\documentclass[a4paper,10pt]{amsart}

\usepackage{amsmath}
\usepackage{amsthm}
\usepackage{amssymb}
\usepackage{wasysym}
\usepackage{graphicx,color,xcolor,tikz}
\usepackage[colorlinks=true, linkcolor=blue]{hyperref}
\usepackage[notref,notcite]{showkeys}
\hypersetup{urlcolor=blue, citecolor=blue}
\usepackage{caption}
\usepackage{epstopdf}

\newtheorem{Thm}{Theorem}[section]
\newtheorem{Lem}{Lemma}[section]
\newtheorem{Def}{Definition}[section]
\newtheorem{Prop}{Proposition}[section]
\newtheorem{Rem}{Remark}[section]

\newtheorem{Cor}{Corollary}[section]
\newcommand{\p}{\partial}
\newcommand{\R}{\mathbb R}

\let\epsilon=\varepsilon
\title{A relaxation model for liquid-vapor
phase change with metastability}
\author{
Fran\c{c}ois James}
\address{Math\'ematiques -- Analyse, Probabilit\'es, Mod\'elisation -- Orl\'eans (MAPMO),
Universit\'e d'Orl\'eans \& CNRS UMR 7349, BP 6759, F-45067 Orl\'eans Cedex 2, France}
\email{francois.james@univ-orleans.fr}

\author{
H\'el\`ene Mathis
}

\address{Laboratoire Jean Leray, Universit\'e de Nantes \&CNRS UMR 6629, 
BP 92208, F-44322 Nantes Cedex 3, France}
\email{helene.mathis@univ-nantes.fr}



\begin{document}

%--------------------------------------------------
% Abstract
%--------------------------------------------------
\begin{abstract}
  We propose a model that describes phase transition
  including meta\-stable states present in the van der Waals Equation
  of State.
  From a convex optimization problem on the Helmoltz free
  energy of a mixture, we deduce 
  a dynamical system that is able to depict the mass
  transfer between two phases, for which equilibrium states are either
  metastable states, stable states or {a coexistent state}.
  The dynamical system is then used as a relaxation source 
  term in an isothermal 4$\times$4 two-phase model.
  We use a Finite Volume scheme that
  treats the convective part and the source term in a fractional step
  way. Numerical results illustrate the ability of the model to capture
  phase transition and metastable states.
\end{abstract}

\maketitle

\tableofcontents

%--------------------------------------------------
\section{Introduction}
\label{sec:introduction}
%--------------------------------------------------

In the last decades considerable research has been devoted to the
simulation of liquid-vapor phase change, which are of major importance
in several industrial applications.
For instance liquid-vapor flows are present in water circuit of
pressurized water reactors in which the water  can be submitted to
saturation pressure and temperature \cite{Delhaye81, Delhaye98}.
The phenomena we are interested in are complex phase changes
including the possible appearance of metastable states.
An example is the metastable vapor which
is a gaseous state where the pressure is higher
than the saturation pressure. Such states are very unstable and a
very small perturbation brings out a droplet of liquid inside the
gas \cite{Landau}. This phenomenon can appear at saturated pressure (or at
saturated temperature for metastable liquid). It is the case in
pressurized water reactor during a loss of coolant accident when
sudden vaporization occurs due to the drop of pressure inside the
superheated liquid \cite{imre10}.

We focus in this paper on situations where
the heterogeneity of the fluid and the thermodynamical
conditions allow the diphasic flow to be described by a compressible
averaged model using Euler type equations. 
Other models can be considered, which account for very smale
scale by means of Korteweg type tensors including dispersive and
dissipative effects. Such models allow to preserve metastable states
but give only a microscopic description of the flow, see 
\cite{slemrod83,slemrod84,knowles91,bedjaoui05,LeFloch10,zeiler12}
and the references cited herein.
In the averaged models framework, one can distinguish between
one-fluid and two-fluid models.

The one-fluid model approach consists in describing the fluid flow
as a single substance that can be present in its vapor or
its liquid phase.
Assuming that the thermodynamical equilibrium is reached
instantaneously (quasi-static process), then
the Euler system has to be closed by an Equation of
State (EoS) able to depict either the pure phases (liquid or vapor)
and the phase transition. A typical example is the van der Waals EoS,
which is well-known to depict stable and metastable states below the
critical temperature. 
However this EoS is not valid in the so-called
spinodal zone where the pressure is a decreasing function of the
density. This forbids the use of instantaneous kinetic exchanges,
since the pressure is always given by the EoS,
and a decreasing pressure leads to a loss of hyperbolicity in Euler equations, hence
to instabilities and computational failure.

To overcome this defect, and recover that
 the phase transition happens at constant pressure, temperature
and chemical potential, the van der Waals pressure is commonly corrected
by the Maxwell equal area rule \cite{callen85}. 
This construction leads to a correct constant
pressure in the spinodal zone but removes the metastable
regions. 

Another way to provide a unique EoS able to depict pure phases and phase 
transition has been studied in \cite{jaouen01, faccanoni07, faccanoni12,HM10, mathis10}.
It consists in considering that each phase is depicted by its own convex 
energy (that is its own monotone pressure law).
According to the second principle of thermodynamics,
the mixture equilibrium energy corresponds to the inf-convolution
of the partial energies. As a result the mixture equilibrium energy 
coincides with the convex hull of the minimum of the two partial energies.
The resulting pressure law of the mixture turns out to be composed of the 
monotone branches of the liquid and vapor pressure laws and is constant in the 
phase transition zone. Hence it is clear that such a construction prevents
the appearance of metastable phases.

Still in the context of one-fluid models, it is also possible to drop out the assumption of
instantaneous equilibrium. The model involves then relaxation terms, which can be of various
forms, but are the expression of a pulling back force to the equilibrium. We have to consider
extended versions of the Euler system, which is supplemented by partial
differential equations on additional quantities such as
the volume fraction of the vapor phase or partial masses.
This approach has been used in \cite{BH05,HS06, saurel08,pelanti14,LMSN14} and in
\cite{corlifan04,fanlin13} in the isothermal case.
In the later references the question of preservation of metastable
states is adressed.

We mention briefly another way to describe diphasic flows, which consists in considering a two-fluid
approach to model liquid-vapor phase change. 
Initially developed to
depict the motion of multicomponent flows \cite{BN86}, such a modeling 
assumes that the fluid can locally be present under both
phase. Hence the model admits two pressures, two velocities and two
temperatures and is supplemented by additional
equations on the volume fraction.
Phase transition is achieved by chemical and mechanical
relaxation processes, in the limit where the kinetics is considered infinitely fast, see for instance \cite{saurel08,Zein10,pelanti14}.

One of the present drawbacks of the averaged models (one or two-fluid) with relaxation is
that there is no global agreement on the equations satisfied by the fractions and on the
transfer terms \cite{drew83}. 
Moreover the preservation of the metastable states seems to be
out of reach in this framework.

We intend in this paper to provide a model able to depict
liquid-vapor phase change and metastable states of a single component,
say some liquid in interaction with its own vapor.
We focus on a one-fluid description of the motion while the phase
transition is driven by transfer terms that will be
coupled to fluid equations through a finite relaxation speed.

The modelling of the phase transition is the core of this work. 
For the sake of simplicity we assume the system to be isothermal.
We propose transfer terms obtained through the minimization
of the Helmholtz free energy of the two-phase system.
We use for both phases the same equation of state
 which has to be non monotone, typically the reduced van der Waals equation.
This choice allows us to recover all possible equilibria: pure phases (liquid or vapor), metastable states
and coexistence states characterized by the equality of pressures and
chemical potentials. These are physically admissible states, but the set of equilibrium states also contains
the spinodal zone, which is irrelevant. Thus the key point now is to characterize the physical stability, and
hence admissibility, of these states. It turns out that this has to be done in
terms of their dynamical behaviour. 
More precisely, we design a dynamical system which is able to depict all
the {\sl stable} equilibria of the system as attraction points. In particular we show that metastable states and mixtures
have different basins of attraction, so that they can be differentiated only by their long-time behaviour with respect to
initial conditions. Hence there is no hope to recover both metastable states and coexistent states under the assumption
of instantaneous equilibrium, which amounts precisely to choose {\sl a priori} this long-time behaviour.

This dynamical system is used as a transfer term in an isothermal
two-phase model in the spirit of \cite{guillard05} and
\cite{chalons09}. The extended Euler system we obtain in this way provides a regularization 
of the isothermal Euler equation with van der Waals EoS, which takes the form of a mixture zone surrounding the 
physical interfaces, see Section \ref{sec:numer-appr} below.

The paper is organized as follows. 
Section~\ref{sec:thermo} is devoted to the thermodynamics of a
two-phase fluid. We provide the definitions of the common potentials 
and give some details on the reduced van der Waals model.
Assuming that both phases follow the same non monotone EoS, 
we describe the thermodynamical
equilibrium as the result of a minimization process on the Helmoltz free
energy of the two-phase fluid.
The section ends with the study of the equilibria of the optimality
system.
Section~\ref{sec:dynam-syst} is the core of this work. It is devoted to the construction of
the dynamical system based on the results of the previous section. A few numerical simulations
illustrate the ability of the system to catch both the Maxwell line and the metastable states in the
van der Waals EoS.
The dynamical system is then plugged as a source term in a $4\times 4$
isothermal model in section~\ref{sec:model}. We provide a study of the homogeneous system, which is conditionally hyperbolic.
However we prove that for smooth solutions the hyperbolicity regions are invariant domains of the system with relaxation.
In section~\ref{sec:numer-appr} we present several numerical illustrations 
which assess the ability of the model to deal with metastable states.
They are obtained using a classical finite volume
schemes that treats the convective part and the source terms with a
time-splitting technique.

%--------------------------------------------------
% Section 1 : thermodynamics of pure phase and binary mixture
%--------------------------------------------------

%--------------------------------------------------
\section{Thermodynamics and the van der Waals EoS}
\label{sec:thermo}
%--------------------------------------------------

%--------------------------------------------------
\subsection{Thermodynamics of a {single fluid}}
\label{sec:therm-single-phase}
%--------------------------------------------------

Consider a fluid of mass $M\geq 0$ occupying a volume $V\geq 0$, 
assumed to be at constant temperature $T$. If the fluid is homogeneous
and at rest, its behavior is entirely described by its mass, its
volume and its Helmholtz free energy $F$.
According to Gibbs' formalism \cite{Gibbs}, the fluid is at equilibrium if
its Helmholtz free energy is a function, also denoted by $F$, of its
mass $M$ and volume $V$:
	\begin{equation}\label{eq:E_extensive}
  F:(M,V) \to F(M,V).
	\end{equation}
Notice that we do not address yet the stability of equilibrium states.
At this level the involved quantities are extensive, which means that
they share the same scalings as the volume $V$. This corresponds 
to the notion of homogeneous sample: for twice the volume, the mass is
double, and the energy as well. The mathematical consequence of this 
notion is that extensive variables have to be related through
positively homogeneous functions of degree 1 (PH1), namely
\begin{equation}\label{eq:PH1}
  \forall \lambda>0, \quad F(\lambda M, \lambda V) = \lambda F(M,V).
\end{equation}
We assume in addition, and without loss of generality, that the energy
function $F$ belongs to
$\mathcal C^2(\R^+ \times \R^+)$.
\begin{Rem}
  The regularity of $F$ seems to preclude phase transitions, but this
  is not the case because 
  no convexity assumption is made at this stage. We shall come back to
  this in more details in the next section. 
\end{Rem}
We introduce now intensive quantities, by which we mean that they do
not depend on the volume scaling. A typical example is the density
  $\rho=M/V$, but such 
quantities also appear as derivatives of the the equilibrium relation
\eqref{eq:PH1}, which are homogeneous of degree $0$ by construction.
In this way two fundamental quantities are to be considered, namely
the pressure $p$ and the chemical potential $\mu$, defined by
\begin{equation}\label{eq:P}
  p = -\dfrac{\p F}{\p V}(M,V), \qquad \mu = \dfrac{\p F}{\p M}(M,V).
\end{equation}
Notice that these quantities are defined only when the system is at
equilibrium, and we recover the classical thermodynamic relation for
isothermal flows
\begin{equation}\label{eq:relation_thermo}
  dF = -pdV + \mu dM.
\end{equation}
Another classical property of thermodynamical potentials is the
so-called Gibbs relation,
which results from the Euler relation for PH1 functions:
\[
F(M,V) = \nabla F(M,V) \cdot
\begin{pmatrix}
  M\\V
\end{pmatrix}.
\]
Using \eqref{eq:P}, we obtain
\begin{equation}\label{eq:gibbs_extensive}
  F(M,V) = \mu M -p V.
\end{equation}

For most of the following computations it is useful to rewrite the preceding relations using intensive variables. 
For a fixed volume $V$, we denote $f$ the Helmoltz free energy per
  unit volume, which is a function of the density $\rho = M/V$:

\begin{equation}\label{eq:f_intensive2}
  f(\rho)=f\left( \dfrac M
    V\right)=\dfrac1VF(M,V)=F\left( \frac{M}{V}, 1 \right).
\end{equation}
We keep the notations $p$ and $\mu$ the pressure and the chemical
potential as functions of the density $\rho$:
\begin{equation}\label{eq:p_rho}
  p(\rho) = p\left( \dfrac M V\right)=-\dfrac{\p E}{\p V}\left(\dfrac M
    V, 1 \right),\qquad
  \mu(\rho) =\mu\left( \dfrac M V\right)=\dfrac{\p E}{\p M}\left(\dfrac
    M V, 1 \right).
\end{equation}
Thus we obtain an intensive form of the Gibbs relation
\eqref{eq:gibbs_extensive}
\begin{equation}\label{eq:gibbs_intensive}
  f(\rho)=\rho \mu(\rho) - p(\rho).
\end{equation}
On the other hand, from the definitions of $p,\;\mu$ and $f$ we obtain
easily the following relations
\begin{equation}\label{eq:fprim_rho}
  \mu(\rho) = f'(\rho),\quad p(\rho) = \rho f'(\rho)-f(\rho), \quad
  \rho \mu'(\rho)=p'(\rho) .
\end{equation}

%--------------------------------------------------
\subsection{Example: the van der Waals EOS}
\label{sec:van-der-waals}
%--------------------------------------------------

The extensive energy of a van der Waals (monoatomic) fluid is
\begin{equation}\label{eq:E_vdW}
  F(M,V) = -\dfrac{a M^2}{V} + RT\left(M\log\dfrac{M}{V-Mb}-M\right),
\end{equation}
where $R$ stands for the perfect gases constant and $a$ and $b$ are 
positive constants ($a=b=0$ leads to the perfect gases law). 
Since we consider the isothermal model, the
temperature $T$ is a parameter here. We refer to \cite[Ch. 7]{Landau} for
a justification of this law from statistical thermodynamics.
The constant $b$ is proportional to the proper volume of a particule
such that $V>Mb$, and the potential $\dfrac{a}{V}$ depicts the
interaction between
particules.

The pressure and the chemical potential associated read
\begin{equation}\label{eq:p_mu_vdw}
  \begin{aligned}
    p\left( \dfrac M V\right)=& -\dfrac{a M^2}{V^2} +
    RT\dfrac{M}{V-Mb}, \\
    \mu\left( \dfrac M V\right) =& -2\dfrac{a M}{V} +
    RT\log\dfrac{M}{V-Mb} +
    RT\dfrac{Mb}{V-Mb}.
  \end{aligned}
\end{equation}

The intensive quantities are
\begin{equation}
  \begin{aligned}
    f(\rho) &= - a \rho^2 +RT\rho\left(\log \left(\dfrac \rho
        {1-b\rho}\right) -1\right),\\
    p(\rho) &=  - a {\rho^2} +\dfrac {\rho RT}{1-b\rho},\\
    \mu (\rho) &= -2{a}{\rho} + RT\log\dfrac{\rho}{1-b\rho} +   
    RT\dfrac{b\rho}{1-b\rho}.
  \end{aligned}
\end{equation}

\begin{figure}
  \caption{Phase diagram for the van der Waals EoS in the $(p,\rho)$
    plan. Below the critical temperature $T_C$, the isotherm curve decreases in
    the spinodal zone which is delimited by the densities $\rho^-<\rho^+$.
    In that area the isotherm is commonly replaced by
    an horizontal segment (dashed line) which coincides with the isobaric line at
    constant pressure $p=p^*$. Such a construction defines the two
    densities $\rho_1^*$ and $\rho_2^*$.}
  \label{fig:vdw_isoth}
  \centering
  \includegraphics[height=6cm]{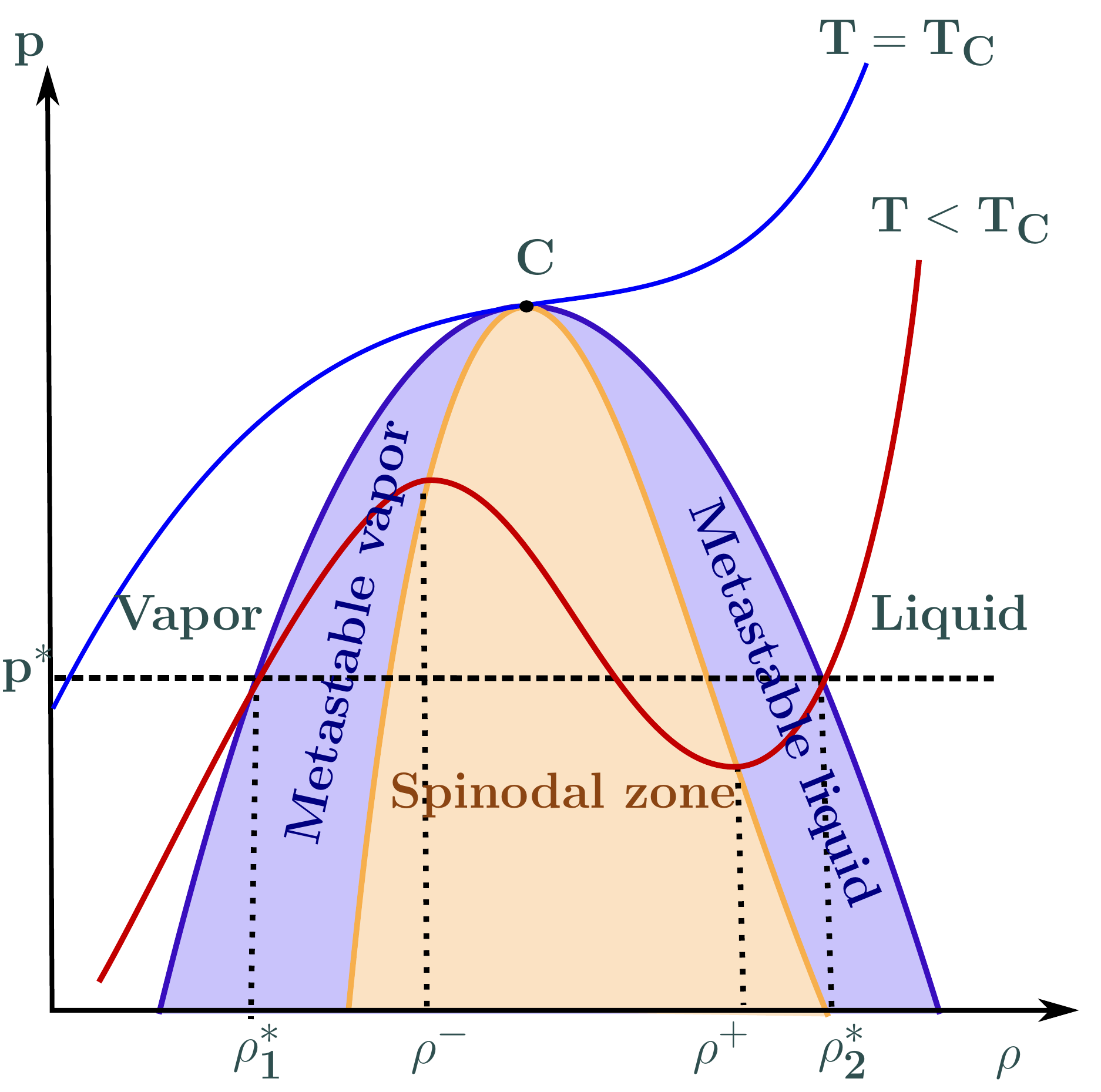}
\end{figure}

The behavior of the isotherm curves in the pressure-density plan is
depicted in Figure \ref{fig:vdw_isoth}. 
The critical temperature $T_c$ is the lower limit of temperatures such
that the pressure is an increasing function of the density. At $T=T_c$
the pressure curve admits an unique horizontal inflection point,
called the critical point, denoted $C$ on Figure \ref{fig:vdw_isoth}. 

In the sequel we will consider the reduced form of the van der Waals
equation (see \cite[Ch. 8]{Landau}).
Denoting $(\rho_c,p_c)$ the coordinates of the critical point one has
\begin{equation*}
  \begin{aligned}
    p'(\rho_c)   &= -2a\rho_c + \dfrac{RT_C}{(1-b\rho_c)^2}=0,\\
    p''(\rho_c)  &= -2a +\dfrac{2bRT_c}{(1-b\rho_c)^3}=0.
  \end{aligned}
\end{equation*}
Considering normalized critical quantities, that is
setting $\rho_c=1$, $p_c=1$ and
$T_c=1$, one obtains the reduced van der Waals EOS with
$R=8/3, \, a=3,$ and $b=1/3$ and
\begin{equation}
\label{eq:reduced_vdw}
  \begin{aligned}
    f(\rho) &=-3\rho^2 + \dfrac 8 3 \rho \left( \log(3\rho/(3-\rho)) -
      1\right),\\
    p(\rho) &= -3\rho^2 + \dfrac{8\rho}{3-\rho},\\
    \mu(\rho) &= -6\rho + \dfrac 8 3 \log(3\rho/(3-\rho)) + \dfrac 8 3
    (3\rho/(3-\rho))
  \end{aligned}
\end{equation}
for $0<\rho< 3$.

Below the critical temperature $T_c$ the pressure is not monotone with
respect to the density (see the red curve on Figure
\ref{fig:vdw_isoth}):
in a region called spinodal zone, delimited for a given temperature by
the densities $\rho^-<\rho^+$, the pressure decreases with respect to
the density and thus leading to unstable states. According to
\eqref{eq:fprim_rho} the energy $f$ is non-convex in the spinodal zone.

In that region the isotherm is commonly replaced by the Maxwell area
rule in order to recover that phase transition happens at constant
pressure and constant chemical potential. 
The Maxwell construction is commonly applied
 on the pressure (see \cite{callen85} for instance) in such a way that 
the two zones  delimited by the van der Waals isotherm and the Maxwell line
(above and below the Maxwell line respectively) have the same area.
This is not the case in Figure \ref{fig:vdw_isoth} because in our context,
the equal area rule is obtained on the chemical potential, see 
Section \ref{sec:equil-constr-opt} below (Proposition \ref{prop:area_rule_mu}).
In any case, the idea is that the isotherm curve is replaced locally by the horizontal segment,
the so-called Maxwell line,
that coincides with some isobaric line $p=p^*$. Such a construction
defines two densities, denoted $\rho_1^*$ and $\rho_2^*$, as well as the value of $p^*$.
An equivalent way to compute this correction is to build the convex
hull of the function $f$, see \cite{faccanoni07} and
\cite{HM10}.

However this construction removes the admissible regions delimited by
the densities where the
pressure law is still nondecreasing. Such regions are called
metastable regions. At a given temperature these regions correspond to
densities belonging to the range $[\rho_1^*,
\rho^-]$ and $[\rho^+, \rho_2^*]$.

%--------------------------------------------------
\subsection{Thermodynamics of  {phase transition}}
\label{sec:therm-two-phase}
%--------------------------------------------------

The van der Waals model depicts the thermodynamical behavior of a pure
substance under its liquid state, gaseous state and the coexistence
state.
The non convexity of the EoS allows to capture metastable states but
does not give a relevant representation of the coexistence phase.
A convenient way to cope with this problem is to represent the fluid
under consideration as a set of several copies of itself under
different pure phases (liquid or gaseous phases).
Such a representation is used in
\cite{jaouen01,BH05,HS06,faccanoni07,saurel08,HM10,faccanoni12,pelanti14,LMSN14} assuming that
the fluid is described by two copies,  each one
satisfying a convex EoS that differs from the one of another copy.
See also \cite{James97} where such a representation is used in the context
of adsorption-desorption of a mixture.

We adopt here a slightly modified approach.
Let us consider $I$ copies of the pure substance, $I\ge 1$ being some integer \textit{a priori}
arbitrary. Each copy 
is depicted by its mass $M_i\geq 0$ and its volume
$V_i> 0$, and is assumed to be at equilibrium.
Hence a copy can occupy a volume with zero mass.
We suppose that each copy follows the same non-convex energy function
$F(M_i,V_i)$, typically the van der Waals
 extensive energy given by \eqref{eq:E_vdW}. 
This assumption contrasts with the aforementioned references
where different convex energy functions are considered, 
and the number of copies is prescribed.

Thanks to the mass conservation,
the complete system has a total mass $M=\sum_{i=1}^I M_i$.
Assuming that the copies are immiscible,
the total volume is
$V=\sum_{i=1}^I V_i$ (for a mixture of gas, 
one has $V_i= V$, $\forall i=1,\ldots I$,
 and this condition implies Dalton's law, see \cite{callen85}).
Out of thermodynamical equilibrium the free
Helmoltz energy of the system reads
\begin{equation*}
  {\mathcal F}\big((M_i,V_i)_i\big) = \sum\limits_{i=1}^I F(M_i,V_i).
\end{equation*}

Let us fix the total mass $M$ and volume $V$ of the system.
According to the second principle of thermodynamics (see
\cite{callen85} for instance), the thermodynamical equilibrium
corresponds to the
solutions of the constrained optimization problem
\begin{equation*}
  \inf_{\substack{I\ge 1,M_i\ge 0,V_i\ge 0}} \{ \sum_{i=1}^I
  F(M_i,V_i) ; \, \sum_{i=1}^I V_i = V, \, \sum_{i=1}^I M_i = M\}.
\end{equation*}
We stress the fact that the total number $I$ of possible copies is not
fixed \textit{a priori} here. 
However, as a consequence of
Carath\'eodory's theorem in dimension 2 (see
e.g. \cite[Ch. 17]{rockafellar}), we can state the Gibbs phase rule
(see \cite[Ch. 9]{Landau}),
which gives the expected result in this context.
\begin{Lem}[Gibbs phase rule]\label{phaseGibbs}
  We have $I\le 2$.
\end{Lem}
In the sequel we use just as well the term \textit{phase} as the term
\textit{copy}.

Taking into account Lemma \ref{phaseGibbs} and the above intensive
formulations, we rewrite the constrained optimization problem for a
fixed mass $M$ and volume $V$, hence for a fixed $\rho$, as
\begin{equation}\label{eq:pb_opti}
  \inf_{\substack{ \rho_1\geq 0 \\\rho_2 \geq 0}}  \{
  \alpha_1f(\rho_1) + \alpha_2f(\rho_2)\},
\end{equation}
under the constraints
\begin{eqnarray}
  \alpha_1+\alpha_2&=&1, \label{eq:constraint_alf}\\
  \alpha_1 \rho_1+\alpha_2 \rho_2 &=&\rho. \label{eq:constraint_rho}
\end{eqnarray}
Here $\alpha_i=\dfrac{V_i}{V}\in [0,1]$ denotes the volume fraction and
$\rho_i=\dfrac{M_i}{V_i}\geq 0$ the density of the
phase $i=1,2$. 
Notice that we removed the optimization on the phase number $I$, 
since single phases can be recovered 
by $\rho_1=\rho_2$ (and any $\alpha_i$) or one $\alpha_i=0$ (with
undetermined {$\rho_j$, $j\neq i$}). 

We rule out the equality case by noticing that, provided
$\rho_1\ne\rho_2$, we can rewrite the constraints 
\eqref{eq:constraint_alf}-\eqref{eq:constraint_rho} as
\begin{equation}\label{eq:alpha_i}
  \begin{aligned}
    \alpha_1 : \R^+ \times \R^+ \times \R^+ &\to [0,1] 
    &\quad
    \alpha_2 : \R^+ \times \R^+ \times \R^+ &\to  [0,1] \\
    (\rho,\rho_1,\rho_2) &\mapsto\dfrac{\rho-\rho_2}{\rho_1-\rho_2},
    &\quad
    (\rho,\rho_1,\rho_2) &\mapsto -\dfrac{\rho-\rho_1}{\rho_1-\rho_2}.
  \end{aligned}
\end{equation}
We have $\alpha_1 \geq 0$ if and only if $\rho_1 \leq \rho\leq
\rho_2$ or $\rho_1 \geq \rho\geq \rho_2$. 
Therefore, accounting on the reduced form of the van der Waals
  model \eqref{eq:reduced_vdw},  we shall assume in the sequel
 that the densities $\rho$, $\rho_1$ and $\rho_2$ satisfy
 \begin{equation}\label{eq:H1}\tag{$H_1$}
   {0<}\rho_1 \leq \rho \leq \rho_2 <3\quad\text{and }\quad
   \rho_1<\rho_2.
 \end{equation}
This is not a restriction, as we shall see below (see Proposition
\ref{prop:DynProp}).

One can also introduce the mass fraction $\varphi_i=\dfrac{M_i}{M}$,
$i=1,2$ defined by
\begin{equation}
  \label{eq:varphi_i}
  \begin{aligned}
    \varphi_1 : \R^+ \times \R^+ \times \R^+ &\to [0,1] 
    &\quad
    \varphi_2 : \R^+ \times \R^+ \times \R^+ &\to  [0,1] \\
    (\rho,\rho_1,\rho_2) &\mapsto\dfrac{\dfrac{1}{\rho} -
      \dfrac{1}{\rho_2}}{\dfrac{1}{\rho_1} - \dfrac{1}{\rho_2}},
    &\quad
    (\rho,\rho_1,\rho_2) &\mapsto\dfrac{\dfrac{1}{\rho} -
      \dfrac{1}{\rho_1}}{\dfrac{1}{\rho_2} - \dfrac{1}{\rho_1}},
  \end{aligned}
\end{equation}
that satisfy $\varphi_1+\varphi_2=1$ and $\varphi_1\varphi_2\geq 0$ if
and only if
the assumption \eqref{eq:H1} is satisfied. Such quantities will be
useful in the mathematical study of the isothermal two-phase flow
model introduced in Section \ref{sec:model}.

\begin{Rem}
In the aforementioned references
\cite{jaouen01,BH05,HS06,faccanoni07,HM10,faccanoni12} the method
consists in describing the two-phase fluid by a coexistence of two copies of
the same substance. The description can be made either on the
extensive variables or the intensive one.
Unlike our present approach the two copies do not
follow the same EoS: each copy is described by its own strictly convex
extensive energy $F_i$, $i=1,2$, which is a function of the mass $M_i$
and the volume $V_i$ of the phase. Following the second principle of thermodynamics
\cite{callen85}, at equilibrium, the extensive energy of the two-phase fluid is given
by 
\begin{equation}
  \label{eq:infconv}
  F((M,V)) := F_1\square F_2(M,V) = \min_{V_1\geq 0, M_1 \geq 0}
  F_1(M_1,V_1) + F_2(M-M_1, V-V_1),
\end{equation}
under the constraints of mass conservation $M=M_1+M_2$ and
immiscibility (without vacuum) $V=V_1+V_2$.
This operation $\square$ is called inf-convolution operation in the convex
analysis framework \cite{hiriart04}.
In \cite{HM10} the authors investigate the link between the
inf-convolution, the Legendre transform and the $(\max,+)$ algebra.
The Legendre transform of a energy $F$ is a convex function
$(M^*,V^*)\to F^*(M^*,V^*)$ defined by
	\begin{equation*}
F^*(M^*,V^*) = \inf_{M^*\geq 0, V^*\geq 0} \{ M^*M+V^*V-F(M,V)\}.
	\end{equation*}
The inf-convolution is transformed into an addition by the Legendre
transform which implies that 
	\begin{equation*}
(F_1\square F_2)^* = F_1^*+F_2^*.
	\end{equation*}
Moreover in the case of convex lower semi continuous (slc) functions
$F_i$, $i=1,2$, one has
\begin{equation*}
  F_1 \square F_2 = (F_1 \square F_2)^{**} = (F_1^* + F_2^*)^*.
\end{equation*}
It means that the energy $F_{eq}$ of the two-phase fluid at
equilibrium is given by
\begin{equation}\label{eq:convhull}
  F_{eq} = F^{**},
\end{equation}
where $F^{**}$ is the convex hull of the energy $F$.
As it is proved in \cite{faccanoni07} the construction of the convex
hull of the energy $F$ is equivalent to the Maxwell
construction. Hence the operation \eqref{eq:convhull} removes the
metastable regions.
\end{Rem}

%--------------------------------------------------
\subsection{Equilibrium states}
\label{sec:equil-constr-opt}
%--------------------------------------------------
This section is devoted to the characterization of the equilibrium
states of the thermodynamical system, that
is states that realize the infimum in the optimization
problem
\eqref{eq:pb_opti}--\eqref{eq:constraint_alf}--\eqref{eq:constraint_rho}.
The function we minimize is defined by: 
\begin{equation*}
\begin{aligned}
  \tilde{\mathcal{F}} : [0,1]\times [0,1]\times \R^+ \times \R^+ &\to \R\\
  (\alpha_1, \alpha_2, \rho_1,\rho_2) &\mapsto
  \alpha_1f(\rho_1)+\alpha_2f(\rho_2),
\end{aligned}
\end{equation*}
and is $C^1$
on the space $\{\rho_i\geq 0, \, i=1,2\}$.
The constraints \eqref{eq:constraint_alf}--\eqref{eq:constraint_rho} 
are affine, so they are also $C^1$. We
are thus in position to use the Lagrange multipliers characterization
of the infimum (see \cite[Ch. 28]{rockafellar}): $\lambda_\alpha\in
\mathbb R$ and $\lambda_\rho\in \R$
respectively correspond to the constraints~\eqref{eq:constraint_alf}
and~\eqref{eq:constraint_rho}.

Using the definition~\eqref{eq:gibbs_intensive} of the free energy $f$
and the pressure and chemical potential definitions~\eqref{eq:fprim_rho}, 
one deduces the following optimality system of equations
\begin{eqnarray}
  f(\rho_1) + \lambda_\alpha + \lambda_\rho \rho_1 &= 0,\label{eq:lag1}\\
  f(\rho_2) + \lambda_\alpha + \lambda_\rho \rho_2 &= 0,\label{eq:lag2}\\
  \alpha_1 \mu(\rho_1) + \lambda_\rho \alpha_1 &= 0,\label{eq:lag3}\\
  \alpha_2 \mu(\rho_2) + \lambda_\rho \alpha_2 &= 0 \label{eq:lag4}.
\end{eqnarray}

From this optimality system we deduce immediately that there are two
different kinds of equilibria.

\begin{Lem}\label{prop:opti1}
Under hypothesis \eqref{eq:H1}, the equilibrium states are
\begin{enumerate}
\item {\bf Pure liquid or gaseous states}:
  $\alpha_1=0$ (resp. $\alpha_2=0$), with $\rho_2=\rho$, $\rho_1<\rho$
  arbitrary 
  (resp. $\rho_1=\rho$, $\rho<\rho_2$ arbitrary)
\item {{\bf Coexistence states}}: $\alpha_1\alpha_2\neq 0$, with
  $(\rho_1,\rho_2)$ satisfying $\mu(\rho_1)=\mu(\rho_2)$
  and $p(\rho_1)=p(\rho_2)$.
\end{enumerate}
\end{Lem}

\begin{proof}
  The case $\alpha_1=0$ corresponds to the saturation of the
  constraint $\alpha_1+\alpha_2=1$ see~\eqref{eq:constraint_alf}. 
  It leads to $\alpha_2=1$ and
  thus $\rho_2=\rho$ that is only the phase 2 is present. 
  Conversely if $\alpha_2=1$ only the phase 1 remains.
  
  On the other hand let us assume $\alpha_1 \alpha_2 \neq 0$.
  Then~\eqref{eq:lag3} and~\eqref{eq:lag4} lead to the equality of the
  chemical potentials
  \begin{equation*}
    \mu(\rho_1) = \mu(\rho_2) = -\lambda_\rho.
  \end{equation*}
  Then the intensive Gibbs relation~\eqref{eq:gibbs_intensive}
  allows to rewrite the conditions~\eqref{eq:lag1} and~\eqref{eq:lag2}
  as
  \begin{equation*}
    \begin{aligned}
      \rho_i \mu(\rho_i) -p(\rho_i) +\lambda_\alpha + \lambda_\rho
      \rho_i =0, \quad i=1,2.
    \end{aligned}
  \end{equation*}
 Since $-\lambda_\rho=\mu(\rho_i)$, this leads to the pressures
 equality
 \begin{equation*}
   p(\rho_1)=p(\rho_2) = \lambda_\alpha.
 \end{equation*}
\end{proof}

To proceed further we need the following result for coexistence states.
\begin{Prop}\label{prop:area_rule_mu}
  Under hypothesis~\eqref{eq:H1} and if $\alpha_1\alpha_2 \neq 0$, 
  the following propositions are
  equivalent and uniquely define 
  the couple of densities $\rho_1^*<\rho_2^*$.
  \begin{enumerate}
  \item The chemical potentials and the pressures are equal
    \begin{align}
      \label{eq:eq_mu}
      \mu(\rho_1^*) &= \mu(\rho_2^*),\\
      \label{eq:eq_p}
      p(\rho_1^*) &=p(\rho_2^*).
    \end{align}
  \item The  Maxwell's area rule on the chemical potential holds
    \begin{equation}
      \label{eq:maxwell_mu}
      \int_0^1 \mu(\rho_2 + t(\rho_1-\rho_2))dt = \mu(\rho_1^*) =
      \mu(\rho_2^*).
    \end{equation}
  \item The difference of the Helmoltz free energies reads
    \begin{equation}
      \label{eq:e2-e1}
      f(\rho_2^*) - f(\rho_1^*) = \mu(\rho_1^*)(\rho_2^*-\rho_1^*) =
      \mu(\rho_2^*) (\rho_2^*-\rho_1^*).
    \end{equation}
  \end{enumerate}
\end{Prop}

\begin{proof}
  The identities \eqref{eq:maxwell_mu} and \eqref{eq:e2-e1} are
  equivalent since $f'(\rho)=\mu(\rho)$, see~\eqref{eq:fprim_rho}.  
  Now assume \eqref{eq:eq_mu}-\eqref{eq:eq_p}
  hold. Then the intensive Gibbs relation~\eqref{eq:gibbs_intensive} gives
  \eqref{eq:e2-e1}. Conversely \eqref{eq:e2-e1} implies the chemical
  potentials equality and thus the pressures equality. 
  Now the uniqueness of $(\rho_1^*,\rho_2^*)$ follows easily form a
  geometrical argument using the Maxwell's area rule.
\end{proof}

The most famous characterizations of diphasic equilibria are
\eqref{eq:eq_mu}-\eqref{eq:eq_mu} and \eqref{eq:maxwell_mu}, although
the latter is usually written in terms of pressure. We can recover
this form by writing the intensive relations with the specific
volume $\tau=V/M$ instead of $\rho$. The third relation
\eqref{eq:e2-e1} is not so classic but will be useful in the sequel.

The density $\rho_1^*$ (resp. $\rho_2^*$) separates the range of
pure gaseous state and the range of metastable gas (resp. the range
of pure liquid state and the range of metastable liquid),
see Figure \ref{fig:vdw_isoth}.
These densities define the coexistence pressure $p^*$ and the
coexistence chemical potential $\mu^*$:
\begin{equation}\label{eq:mus_ps}
  p^* = p(\rho_1^*)=p(\rho_2^*), \qquad
  \mu^*= \mu(\rho_1^*) = \mu(\rho_2^*).
\end{equation}

We emphasize that the necessary conditions for equilibrium contain
both unstable (spinodal zone), metastable and stable states. Nothing at this stage 
can make the difference, which turns out to be of dynamical nature, in a way we make 
precise now.

%--------------------------------------------------
% Section 2: dynamical system
%--------------------------------------------------

%--------------------------------------------------
\section{Dynamical system and phase transition}
\label{sec:dynam-syst}
%--------------------------------------------------

This section is devoted to the construction and the analysis of a
dynamical system deduced from the optimality conditions given in
Lemma~\ref{prop:opti1} and Proposition~\ref{prop:area_rule_mu}.
We will prove that the equilibria of this dynamical system are either
pure liquid/vapor states, pure metastable states or coexistence states
in the spinodal zone (that is states
satisfying the properties~\eqref{eq:eq_mu}-\eqref{eq:maxwell_mu}
of Proposition~\ref{prop:area_rule_mu}).
We emphasize that the difference between metastable states and coexistent states actually relies 
on the long-time dynamical behaviour of the solutions to the dynamical system. No static characterization can be given.

The section ends with numerical illustrations that assess the ability
of the dynamical system to preserve metastable states.

\subsection{Construction of the dynamical system}
\label{sec:constr-dyn-syst}

We want to construct a dynamical system which equilibria  coincide
with the equilibria of the optimization problem depicted in
Lemma~\ref{prop:opti1}.
To do so, we force the dynamical system to dissipate the Helmoltz free
energy defined by
\begin{equation}
  \label{eq:F}
  \begin{aligned}
    \mathcal F: \R^+ \times \R^+ \times \R^+ &\to \R\\
    (\rho, \rho_1,\rho_2) &\mapsto \alpha_1 (\rho, \rho_1,\rho_2)
    f(\rho_1) + \alpha_2 (\rho, \rho_1,\rho_2)f(\rho_2),
  \end{aligned}
\end{equation}
under the optimization constrains~\eqref{eq:constraint_alf}
and~\eqref{eq:constraint_rho}.

Assuming that $\rho$, $\rho_1$ and $\rho_2$ are now time-dependent
functions, we can compute the derivative of the total Helmoltz
free energy $\mathcal F$ with respect to time and deduce appropriate
time derivatives of $\rho$, $\rho_1$ and $\rho_2$
such that $\mathcal F$ is dissipated in time.
Moreover we want pure states (either liquid, vapor or metastable) to
be equilibria of the dynamical system. 
Hence we have forced the time derivatives of $\rho$,
$\rho_1$ and $\rho_2$ to vanish in case of pure state (that is when
$\alpha_1\alpha_2=0$).

For sake of readability we denote 
\begin{equation}
\label{eq:r}
\mathbf{r}=
\begin{pmatrix}
  \rho\\\rho_1\\\rho_2
\end{pmatrix}
\end{equation}
the vector of admissible densities satisfying the
assumption~\eqref{eq:H1}. 
By the definition~\eqref{eq:alpha_i} of the volume fractions
$\alpha_i,\; i=1,2$, one has
\begin{equation*}
\nabla_{\mathbf r} \alpha_1(\mathbf{r}) = -\nabla_{\mathbf{r}}
\alpha_2 (\mathbf{r})=\dfrac{1}{\rho_1-\rho_2}
\begin{pmatrix}
  1\\
  -\alpha_1(\mathbf{r})\\
  -\alpha_2(\mathbf{r})
\end{pmatrix}.
\end{equation*}
From this and the definition~\eqref{eq:F} of $\mathcal F$,
we easily get the gradient of the free energy
\begin{equation}\label{eq:gradF1}
  \nabla_{\mathbf r} \mathcal F(\mathbf{r})
  = \dfrac{1}{\rho_1-\rho_2}
  \begin{pmatrix}
    f(\rho_1)-f(\rho_2)\\
    \alpha_1(\mathbf{r})(\rho_2(\mu(\rho_2)-\mu(\rho_1))+p(\rho_1)-p(\rho_2))\\
    \alpha_2(\mathbf{r})(\rho_1(\mu(\rho_2)-\mu(\rho_1))+p(\rho_1)-p(\rho_2)).
  \end{pmatrix}
\end{equation}
Note at once
that it can be expressed in terms of relative free energy
\begin{equation}\label{eq:gradF2}
  \nabla_{\mathbf r} \mathcal F(\mathbf{r}) = \dfrac{1}{\rho_1-\rho_2}
  \begin{pmatrix}
    f(\rho_1)-f(\rho_2)\\
    \alpha_1(\mathbf{r})f(\rho_2|\rho_1)\\
    -\alpha_2(\mathbf{r})f(\rho_1|\rho_2)
  \end{pmatrix},
\end{equation}
where the relative
free energy of $\rho_2$ with respect to $\rho_1$ is defined by
\begin{equation}
  \label{eq:relat_f}
  f(\rho_2|\rho_1):=f(\rho_2)-f(\rho_1)-\mu(\rho_1)(\rho_2-\rho_1).
\end{equation}
We turn now to the definition of our dynamical system. Because of the
mass conservation, we obviously impose that
\begin{equation*}
  \dot \rho =0.
\end{equation*}
The main idea to proceed further is that we want the system to
dissipate the total Helmoltz free energy $\mathcal F$.
Combining the definition~\eqref{eq:F} of $\mathcal F$
and the expression of $\nabla_{\mathbf{r}}\mathcal F(\mathbf{r})$,
one computes the time derivative of $\mathcal F$
along some trajectory, that is
\begin{equation}
  \label{eq:Fdot}
    \dot{\mathcal F}(\mathbf{r}) = \dfrac{1}{\rho_1-\rho_2}\bigl(
      \alpha_1(\mathbf{r}) f(\rho_2|\rho_1) \dot \rho_1 - 
      \alpha_2(\mathbf{r}) f(\rho_1|\rho_2) \dot \rho_2
      \bigr).
\end{equation}
We now propose the following dynamical system
\begin{equation}
  \label{eq:syst_dyn2}
  \begin{cases}
    \dot\rho = &0,\\
    \dot\rho_1 = & -(\rho-\rho_1)(\rho-\rho_2)
    f(\rho_2|\rho_1),\\
    %\left(
     %\rho_2(\mu(\rho_2)-\mu(\rho_1)) + p(\rho_1) -p(\rho_2)\right),\\
   \dot\rho_2 = & + (\rho-\rho_1)(\rho-\rho_2)
   f(\rho_1|\rho_2).
   %\left(\rho_1(\mu(\rho_2)-\mu(\rho_1)) + p(\rho_1) -p(\rho_2)\right).
  \end{cases}
\end{equation}
An easy computation shows that this system dissipates $\mathcal
F$ along its trajectories {(under the assumption \eqref{eq:H1})}.
Indeed using the expression~\eqref{eq:Fdot} of $ \dot{\mathcal  F}$
and the equations of $\dot{\rho_1}$ and $\dot{\rho_2}$, one gets
\begin{equation*}
  \dot{\mathcal  F} (\mathbf{r})=
  -(\rho-\rho_1)\left[\alpha_1(\mathbf r)f(\rho_2|\rho_1)\right]^2  
  + (\rho-\rho_2) \left[\alpha_2(\mathbf r)f(\rho_1|\rho_2)\right]^2 .
  % -\dfrac{\rho-\rho_1}{(\rho_1-\rho_2)^2}\left(\dfrac{\p}{\p \rho_1}
  %   \mathcal F(\mathbf{r})\right)^2
  % +  \dfrac{\rho-\rho_2}{(\rho_1-\rho_2)^2}\left(\dfrac{\p}{\p
  %     \rho_2}
  %   \mathcal F(\mathbf{r})\right)^2\leq 0.
\end{equation*}
The multiplicative term $(\rho-\rho_1)(\rho-\rho_2)$ 
in the equations on $\dot \rho_1$ and $\dot \rho_2$
 ensures that the right hand side of the
system vanishes in case of pure states (either pure liquid, vapor or
metastable states). 

\begin{Rem}
  We emphasize that the choice of the right hand side of~
  \eqref{eq:syst_dyn2} is somewhat arbitrary.
  Other terms might be more efficient, but this one was chosen for its
  simplicity and its interpretation in
  terms of relative free energy.
\end{Rem}

One can check easily that $\dot{\overbrace{\alpha_1\rho_1+\alpha_2
    \rho_2}} =0$ so that it is consistent with the total mass
conservation equation $\dot \rho=0$.

An equivalent dynamical system can be written in terms of the time
derivatives of the volume
fractions $\alpha_i$ and of the partial masses $\alpha_i\rho_i$, $i=1,2$.
Accounting for the
constraints~\eqref{eq:constraint_alf}-\eqref{eq:constraint_rho} 
and for the system~\eqref{eq:syst_dyn2}, some straightforward 
computations lead to
\begin{equation}
  \label{eq:dyn_syst_equ}
  \begin{cases}
    \dot\alpha_1 &= -\dot \alpha_2\\
    &= \alpha_1^2 (\rho-\rho_1) f(\rho_2|\rho_1) +
    \alpha_2^2 (\rho-\rho_2) f(\rho_1|\rho_2),\\
    \dot{\overbrace{\alpha_1\rho_1}} &=
    -\dot{\overbrace{\alpha_2\rho_2}}\\
    &=\alpha_1^2\rho_2(\rho-\rho_1)f(\rho_2|\rho_1) + \alpha_2^2 \rho_1
    (\rho-\rho_2)f(\rho_1|\rho_2).
  \end{cases}
\end{equation}
This formulation does not allow to compute single phase
systems, since the determination of the
partial densities $\rho_i$ is impossible when $\alpha_i=0$. Thus we
rather use the dynamical system~\eqref{eq:syst_dyn2}.

%-------------------------------------------------------
\subsection{Equilibria of the dynamical system}
\label{sec:study-dynam-syst}

The major result of this paragraph is that the
attractors of the system are either pure liquid/vapor states,
including metastable states, or
the coexistence state defined
by~\eqref{eq:maxwell_mu}-\eqref{eq:eq_p}, see
Proposition~\ref{prop:area_rule_mu}.

\begin{Prop}\label{prop:DynProp}
  The dynamical system \eqref{eq:syst_dyn2} satisfies the following
  properties.
  \begin{enumerate}
  \item If the assumption \eqref{eq:H1} is satisfied at $t=0$, then it
    is preserved in time  (in particular, the assumption
    $\rho_1<\rho_2$ is preserved).
  \item If $\alpha_1(0)=0$ (resp. $\alpha_1(0)=1$) then $\alpha_1(t)
    =0$ (resp $\alpha_1(t) =1$), for all time.
  \end{enumerate}
\end{Prop}

\begin{proof}
  First we address the preservation of the hypothesis \eqref{eq:H1}. 
  Some straightforward computations lead to
  \begin{equation}
    \label{eq:H1_time}
    \dot\rho_1-\dot\rho_2 =
    (\rho-\rho_1)(\rho-\rho_2)(\rho_1-\rho_2)(\mu(\rho_2)-\mu(\rho_1)).
  \end{equation}
  Thus if $\rho_1-\rho_2$ is zero at $t=0$ then this property is
  preserved for all time and the sign of the difference
  $\rho_1-\rho_2$ is also preserved.
  The property on the volume fraction $\alpha_1$ is proved in the same
  way noting that 
  \begin{equation}
    \label{eq:alpha1_time}
    \dot \alpha_1 (\mathbf{r}) = -\alpha_1(\mathbf{r})
    (\rho-\rho_1)(\rho_1(\mu(\rho_1)-\mu(\rho_2))-p(\rho_1)+p(\rho_2)).
  \end{equation}
\end{proof}
We turn now to the characterization of the equilibria of the
dynamical system. Since the total mass $\rho$ is constant in time, 
they are parametrized by $\rho$. Moreover it arises that their
thermodynamical characterization can be given in terms of attraction
basins.
We refer the reader for graphical references to figures \ref{fig:SpinodalBasin},
\ref{fig:PureBasin} and \ref{fig:MetaBasin}, where attraction basins are drawn as hatched zones.
The last two figures correspond to the gaseous phase, similar pictures can be drawn for the liquid phase.

\begin{Thm}\label{th:attraction}
Assume that the initial data $(\rho(0),\rho_1(0),\rho_2(0))$ of the
system~\eqref{eq:syst_dyn2} satisfy \eqref{eq:H1}.
Then the equilibria are given by
\begin{enumerate}
\item {\bf spinodal zone:} $\rho_1^*<\rho(0)<\rho_2^*$. 
  The unique equilibrium is $(\rho_1^*,\rho_2^*)$ characterized by 
  Proposition \ref{prop:area_rule_mu}, with $\alpha_i$ given by
  \eqref{eq:alpha_i}. 
  The attraction basin is $(0,\rho)\times(\rho,3)$.
\item {\bf pure gaseous states:} 
  $\rho(0)<\rho_1^*$ (resp. {\bf pure liquid states:}
    $\rho(0)>\rho_2^*$).
  The set of equilibria
  is $\{\rho\}\times(\rho,3)$, with $\alpha_1=1$
  (resp. $(0,\rho)\times\{\rho\}$ with $\alpha_2=1$). 
  The attraction basin is $(0,\rho)\times(\rho,3)$.
\item {\bf metastable states:} 
  $\rho_1^*\le\rho(0)\le\rho^-$ (resp. $\rho^+\le\rho(0)\le\rho_2^*$). 
  There are two sets of equilibria
  \begin{enumerate}
  \item {\bf perturbation within the phase:} $\rho_2(0)\le\rho^-$ 
    (resp. $\rho_1(0) \ge \rho^+$). 
    The set of equilibria is $\{\rho\}\times(\rho,\rho^-)$, 
    with $\alpha_1=1$ 
   (resp. $(\rho^+,\rho)\times \{ \rho\}$
      with $\alpha_2=1$). 
    The attraction basin is then
    $(0,\rho)\times(\rho,\rho^-)$
    (resp. $(\rho^+,\rho)\times (\rho,3)$).
  \item {\bf outside perturbation:}
    $\rho^-\le\rho_2(0)\le\rho_2^*$ (resp. $\rho^+\ge\rho_1(0)\ge\rho_1^*$). 
    There is a unique equilibrium $(\rho_1^*,\rho_2^*)$,
    characterized by 
    Proposition \ref{prop:area_rule_mu}, with 
    $\alpha_i$ given by \eqref{eq:alpha_i}. 
    The attraction basin is $(0,\rho)\times(\rho^-,3)$ (resp. $(0,\rho^+)\times(\rho,3)$).
  \end{enumerate}
\end{enumerate}
\end{Thm}

\begin{proof}
  We look for Lyapunov functions for each case.
  
  {\bf Spinodal zone. } Let us define
  \begin{equation}\label{eq:SpinLyap}
    G(\mathbf{r}) = \alpha_1f(\rho_1) + \alpha_2f(\rho_2) -
    \mu^*(\alpha_1\rho_1+\alpha_2\rho_2) + p^*(\alpha_1+\alpha_2),
  \end{equation}
where $\mu^*$ and $p^*$ are defined by the Maxwell area rule
\eqref{eq:mus_ps} on the chemical potential $\mu$,
 and $\alpha_i$ are the functions of
$\rho,\rho_1,\rho_2$
given by \eqref{eq:alpha_i}. The intensive version of the Gibbs
relation \eqref{eq:gibbs_intensive} implies that $G(\mathbf{r}^*)=0$
{where $\mathbf{r}^*=(\rho,\rho_1^*,\rho_2^*)^T$}. 
Straightforward computations show that $\nabla_{\mathbf{r}}G
=\nabla_{\mathbf{r}}\mathcal F - (\mu^*,0,0)^T$,
this implies that $\nabla_{\mathbf{r}}G(\mathbf{r}^*)=0$, and
that $\dot G(\mathbf{r}(t))\le 0$ since the free energy is
dissipated. Hence the function $G$ is a Lyapunov function on the
admissible domain defined by \eqref{eq:H1}, 
which contains the unique equilibrium $\mathbf{r}^*$, see Figure
\ref{fig:SpinodalBasin}.

{\bf Pure stable states} The expected equilibrium states are now 
$\bar{\mathbf{r}}=(\rho,\rho,\rho_2)$ for any $\rho_2\ge\rho$, see
Figure \ref{fig:PureBasin}. We introduce the function 
\begin{equation}\label{eq:PureLyap}
  G^{(1)}(\mathbf{r}) = f(\rho_1) - \bar\mu\rho_1 +
  \bar p,
  \qquad\mbox{where}\quad \bar\mu =\mu(\rho),\, \bar p=p(\rho).
\end{equation}
Once again, the Gibbs relation leads to $G^{(1)}(\bar{\mathbf{r}}) =
0$ for any equilibrium $\bar{\mathbf{r}}$, and we easily have 
$\nabla_{\mathbf{r}}G^{(1)}(\bar{\mathbf{r}})=0$.
 Now we have
 $$
 \dfrac{d}{dt}G^{(1)}(\mathbf{r}(t)) =
 -\big(\mu(\rho_1)-\mu(\rho)\big)(\rho-\rho_1)(\rho-\rho_2)
 f(\rho_2|\rho_1).
 $$
Since $\rho_1\le\rho\le\rho_2$ and $\mu$ is nondecreasing on
$]0,\rho_1^*]$ the right-hand side in the preceding relation is
nonpositive
if $f(\rho_2|\rho_1)\ge 0$. But using again the Gibbs relation and
\eqref{eq:relat_f}, this can be rewritten
\begin{equation}\label{eq:IneqConv}
  f(\rho_2) - f(\rho_1) \ge \mu(\rho_1)(\rho_2-\rho_1).
\end{equation}
This convexity inequality holds true for all $(\rho_1,\rho_2)$
such that $\rho_1\le\rho_1^*$, by the very definition of $\rho_1^*$
(see the definition~\eqref{eq:convhull} of $F^{**}$). Since
we consider pure liquid states, one has $\rho_1\le\rho\le\rho_1^*$.
Hence $G^{(1)}$ is dissipated in time, leading to the conclusion.

{\bf Metastable states} 
Two types of equilibria are encountered in this situation, with two
distinct attraction basins, 
see Figure \ref{fig:MetaBasin}. 
\begin{itemize} 
\item the metastable basin, which appears with \textbackslash\textbackslash\textbackslash\ hatches  in the figure,
  corresponds to perturbation of a given state within the same
  phase. It is actually an attraction basin, because the function
  $G^{(1)}$ defined above
  \eqref{eq:PureLyap} is also a Lyapunov function in this domain.
  Indeed the convexity argument still holds true for any 
  $(\rho_1,\rho_2)$ such that $\rho_2\le\rho^-$, since $f$ is convex
  below $\rho^-$.
\item the unstable basin, corresponding to perturbations of the
  state by the other phase ($///$ hatched zone). 
  This basin is governed by the same 
  Lyapunov function as for the spinodal zone \eqref{eq:SpinLyap}.
\end{itemize} 
\end{proof}

\begin{figure}
  \begin{center}\resizebox{7cm}{!} {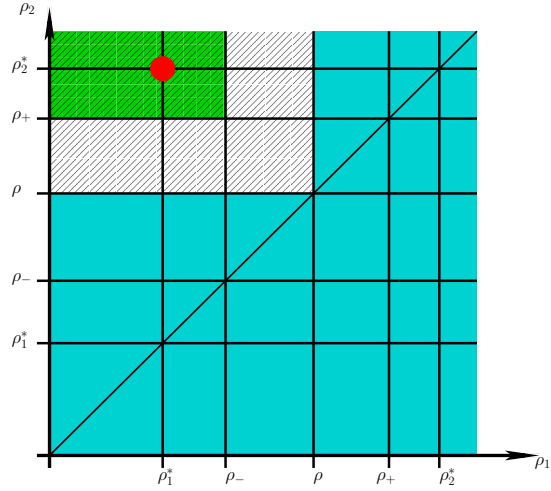}
  \end{center}
  \caption{Spinodal states. The blue area refers to nonattainable
    states according to \eqref{eq:H1}. The attraction basin is the
    hatched area (///). The unique attraction point $(\rho_1^*,\rho_2^*)$ appears in red.
    The green zone corresponds to the invariant hyperbolicity
    region (see Section \ref{sec:dom_hyp}).}
  \label{fig:SpinodalBasin}
\end{figure}

\begin{figure}
  \begin{center}
    \resizebox{6cm}{!} {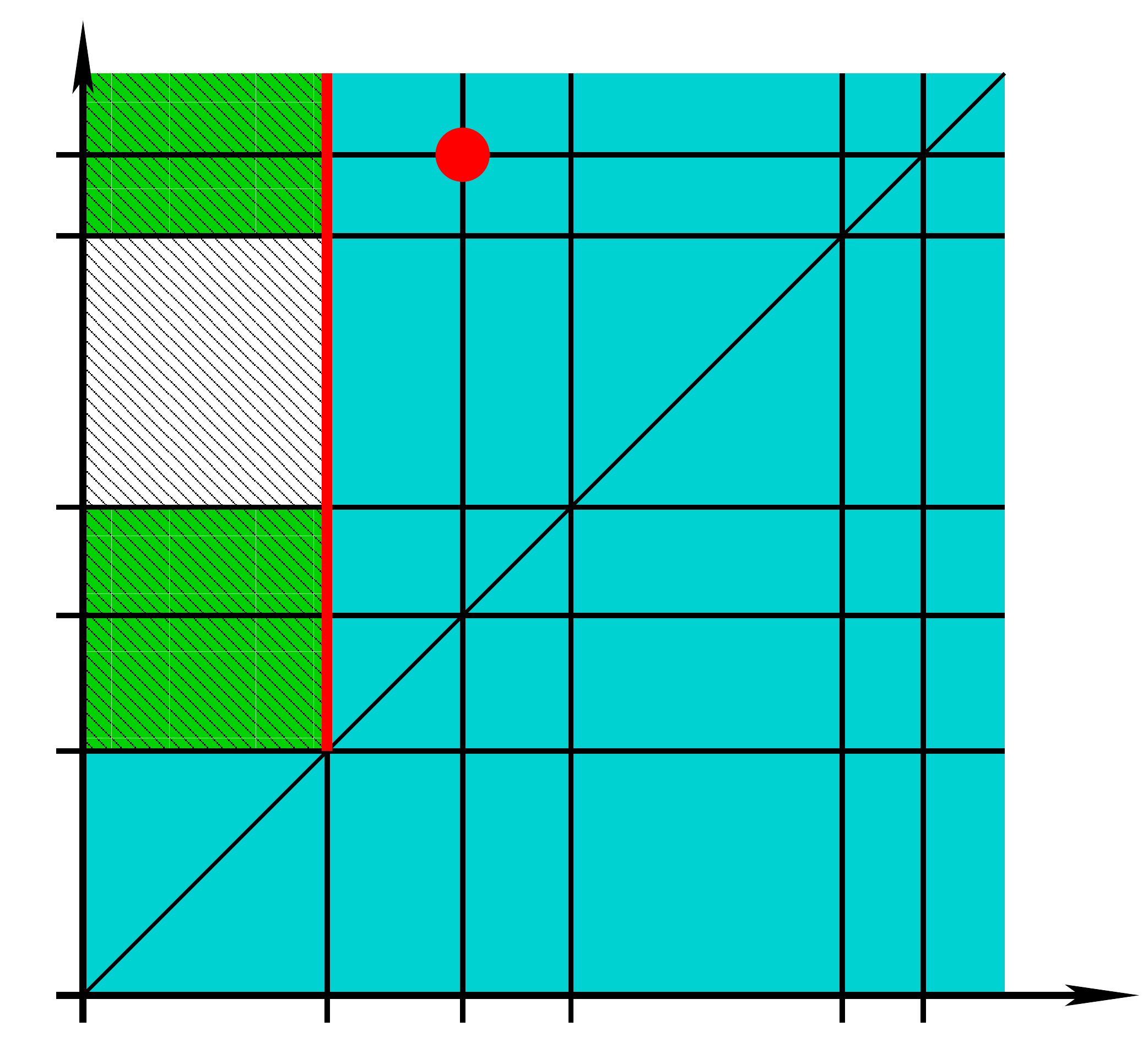}
  \end{center}
  \caption{Pure gaseous states.
    The blue area refers to nonattainable
    states according to \eqref{eq:H1}. The attraction basin is the
    hatched area (\textbackslash\textbackslash\textbackslash). The set of attraction points appears
    in red.
    The two green zones correspond to the invariant hyperbolicity
    regions (see Section \ref{sec:dom_hyp}).}
  \label{fig:PureBasin}
\end{figure}

\begin{figure}
  \begin{center}
    \resizebox{6cm}{!} {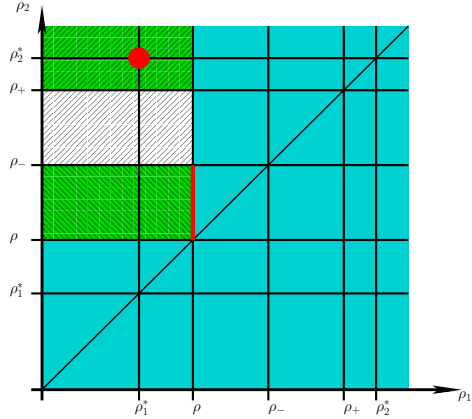}
  \end{center}
  \caption{Pure gaseous metastable states. The blue area refers to nonattainable
    states according to \eqref{eq:H1}. The attraction points appear
    in red. The corresponding attraction basins are the
    hatched areas.  For any state belonging to the hatched area
    (\textbackslash\textbackslash\textbackslash),
    the attraction points are metastable vapor sates such that $\{
    \rho_1=\rho\}$. For any state belonging to the hatched area (///),
    the attraction point is the coexistence state $(\rho_1^*,\rho_2^*)$.
    The two green zones correspond to the invariant hyperbolicity
    regions (see Section \ref{sec:dom_hyp}).}
  \label{fig:MetaBasin}
\end{figure}

\begin{Rem}
This is a formalization of the remark in Landau \& Lifshitz \cite[\S 21]{Landau}
\textit{``[...]we must distinguish between {\bf metastable} and
          {\bf stable} equilibrium states. 
          A body in a metastable state {\bf may not
          return to it} after a sufficient deviation'' }
\end{Rem}

%-------------------------------------------------------
\subsection{Numerical illustrations of the dynamical system behavior}
\label{sec:resol-numer}

We provide in this paragraph some numerical examples to illustrate
the behavior of the dynamical system~\eqref{eq:syst_dyn2}.
First we give the characteristic values of the reduced van der Waals
EoS at a fixed nondimensionalized temperature $T=0.85$. 
Figure \ref{fig:mu_T=0.85} presents the corresponding isothermal
curve in the $(\rho,\mu)$ plan.

The densities $\rho^-$ and $\rho^+$ correspond to the extrema of the
chemical potential.
The densities $\rho_1^*$ et $\rho_2^*$ are obtained by the Maxwell's
equal area rule construction on the chemical
potential~\eqref{eq:maxwell_mu} or equivalently by solving
\eqref{eq:eq_mu}-\eqref{eq:eq_p}.
  
\begin{figure}[h]
  \centering
  \includegraphics[width=12cm]{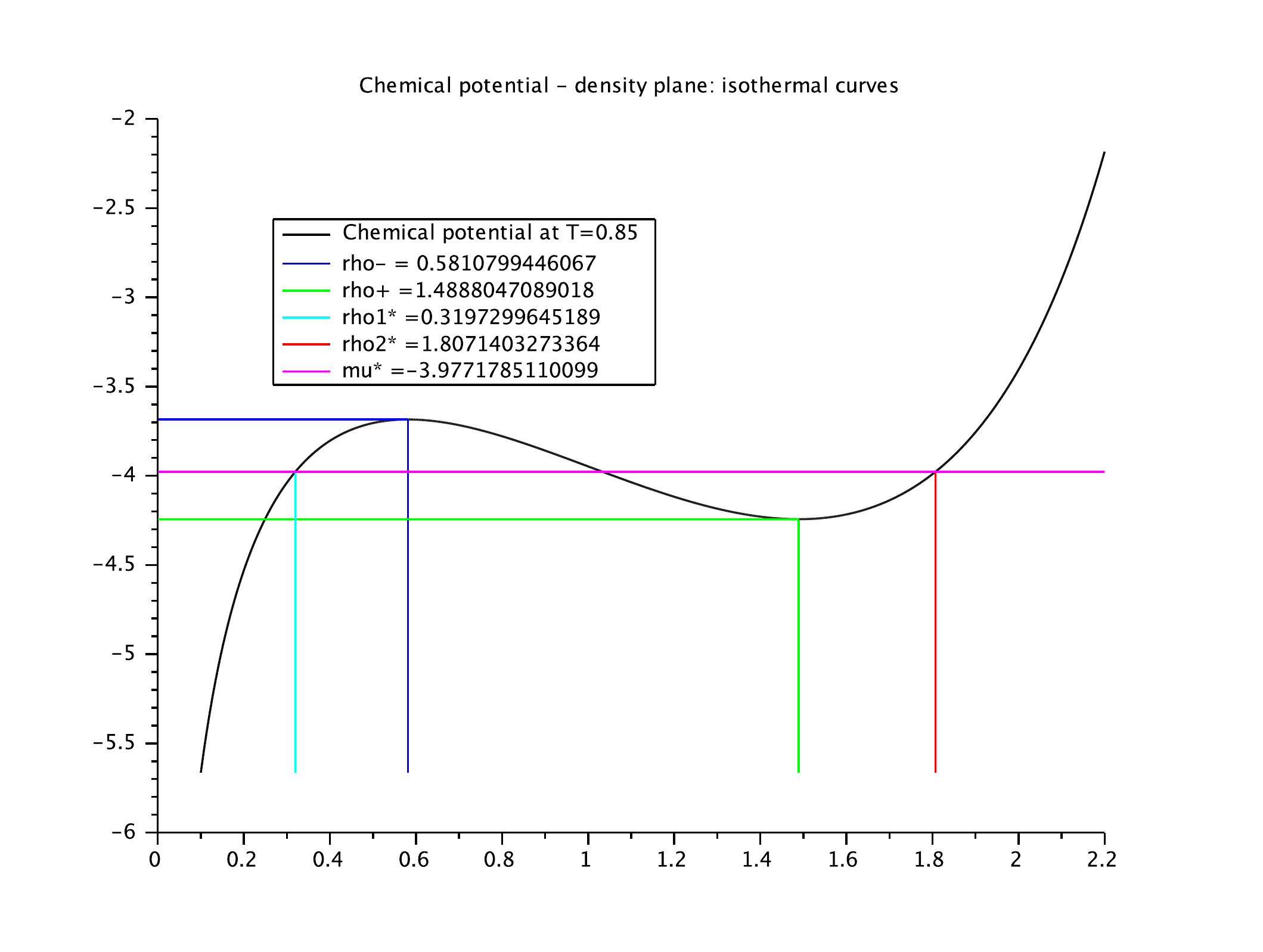}
  \caption{Isotherm curve in the $(\rho,\mu)$ plan at $T=0.85$.}
  \label{fig:mu_T=0.85}
\end{figure}

Table \ref{tab:carac} contains the values of the characteristic
densities and the corresponding values of the pressure and chemical
potential at $T=0.85$.
\begin{table}[h]
  \centering
  \begin{tabular}{|l|c|c|}
    \hline
    $\rho$ & $p(\rho)$ & $\mu(\rho)$\\
    \hline
    $\rho^- = 0.5810799446067$ & $0.62055388470356498 $ & $-
    3.68447967881140137$\\
    \hline
    $\rho^+ =1.4888047089018$ & $0.04962960899844759$ & $-
    4.24339302065563029$\\ 
    \hline
    $\rho_1^* =0.31972996451885$ & $0.504491649787487$ &
    $- 3.97717851100986$\\
    \cline{1-1} 
    $\rho_2^* =1.8071403273364$ & &\\ 
    \hline
  \end{tabular}
  \caption{Characteristic densities at $T=0.85$ and corresponding
    pressure and chemical potential values}\label{tab:carac}
\end{table}
  
Using a Backward Differentiation Formula (provided by Scilab for
stiff problem), we provide numerical illustrations of the attraction
basins depicted in Theorem~\ref{th:attraction}.
To do so, we fix 50 initial values of $\rho(0)$ in a given interval and
set the partial densities $\rho_i(0)$ as perturbations of $\rho(0)$.
Then the dynamical system is solved for a final time $T_f=10^3$ with a
time step set to $10^{-3}$s.
We provide the graphs of the quantities $\alpha_1(T_f)$,
$\rho_i(T_f)$ and $\alpha_1(T_f)p(\rho_1(T_f)) +
\alpha_2(T_f)p(\rho_2(T_f))$, plotted with respect to the density
$\rho(0)$, for each one of the 50 initial conditions
$(\rho(0),\rho_1(0),\rho_2(0))$.
The mixture pressure profile is compared with the classical van der
Waals pressure curve and the Maxwell line.

\textbf{Metastable states}: 
The initial density $\rho(0)$ takes on 50 values in $[\rho_1^*,
\rho^-]$.
For this computation, we set $\rho_1(0) = \rho(0)-0.1$ and
$\rho_2(0)=\rho(0)+0.2$ so that we can observe the perturbation within
and outside the metastable vapor zone.
According to Figure~\ref{fig:sci_pert_meta}-top left,
the mixture pressure presents two different parts: the left part (for
$\rho<0.45$)
remains on the van der Waals pressure curve and a second part
coincides with the Maxwell line. The first part corresponds to the
perturbation of the metastable vapor state within the phase, while the
right part corresponds to the perturbation outside the metastable
vapor, that is when $\rho^-\leq\rho_2(0)\leq \rho_2^*$.
One can check on the volume fraction curve (see
Figure~\ref{fig:sci_pert_meta}-top right) that $\alpha_1(T_f)=1$  for
$\rho(0)\leq 0.45$, that is only the phase 1 is present. Then
$0<\alpha_1<1$ and the corresponding final partial densities
$\rho_{i}(T_f)$, $i=1,2$, coincide with the densities $\rho_1^*$ and
$\rho_2^*$ respectively, which explains that the pressure matches
with the Maxwell line.
\begin{figure}[h]
  \centering
  \includegraphics[width=6cm]{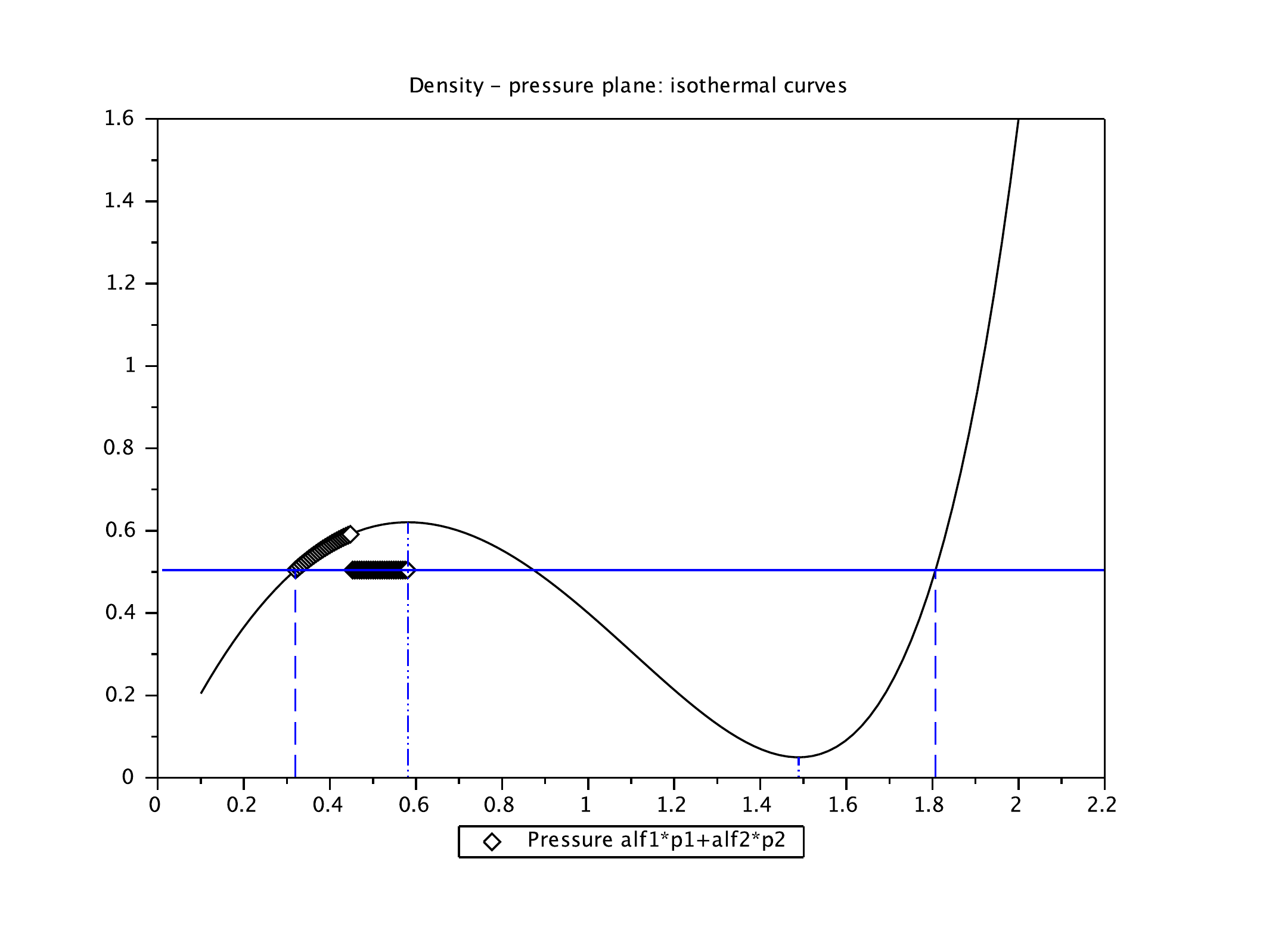}
  \includegraphics[width=6cm]{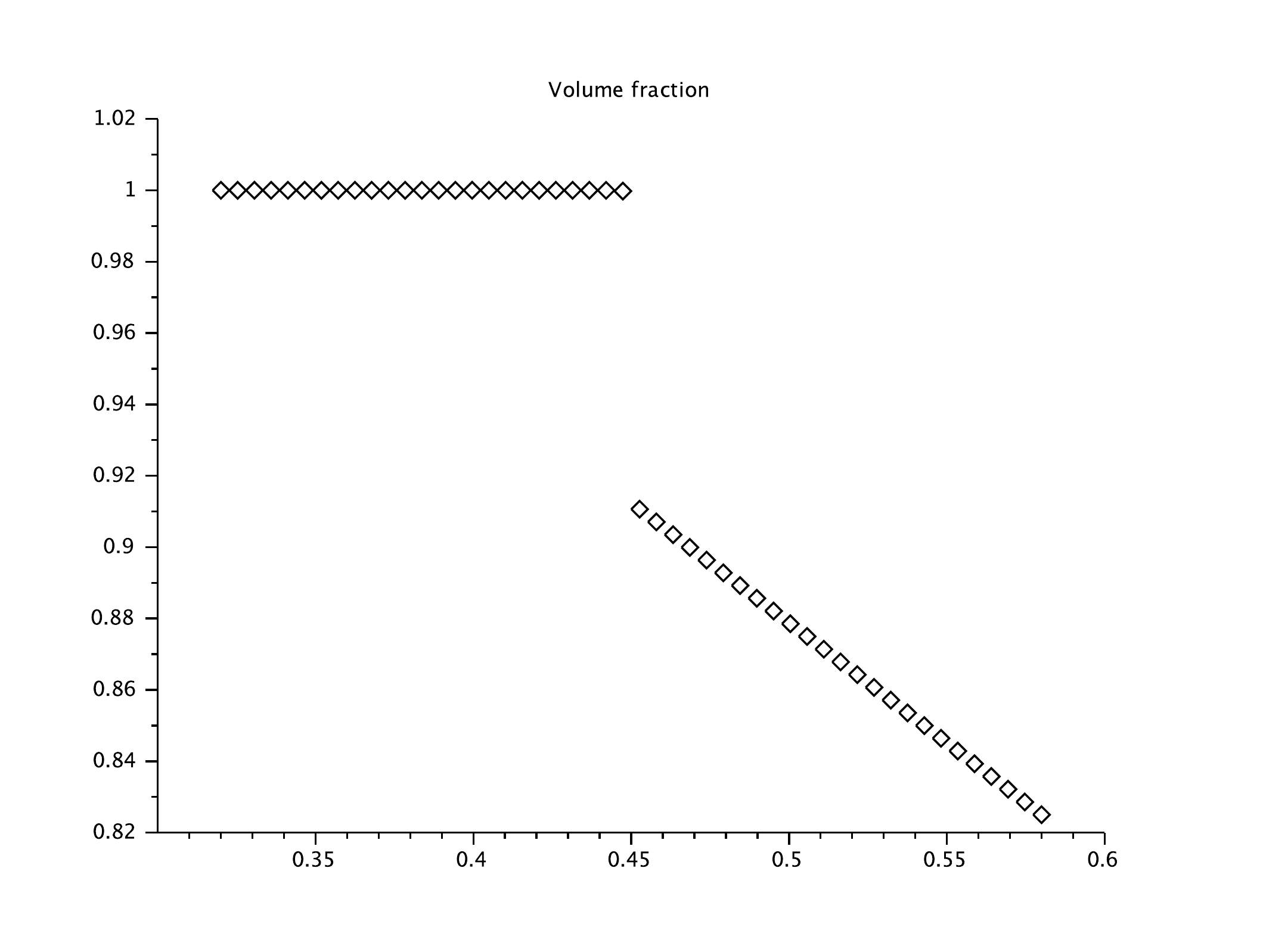}
  \includegraphics[width=6cm]{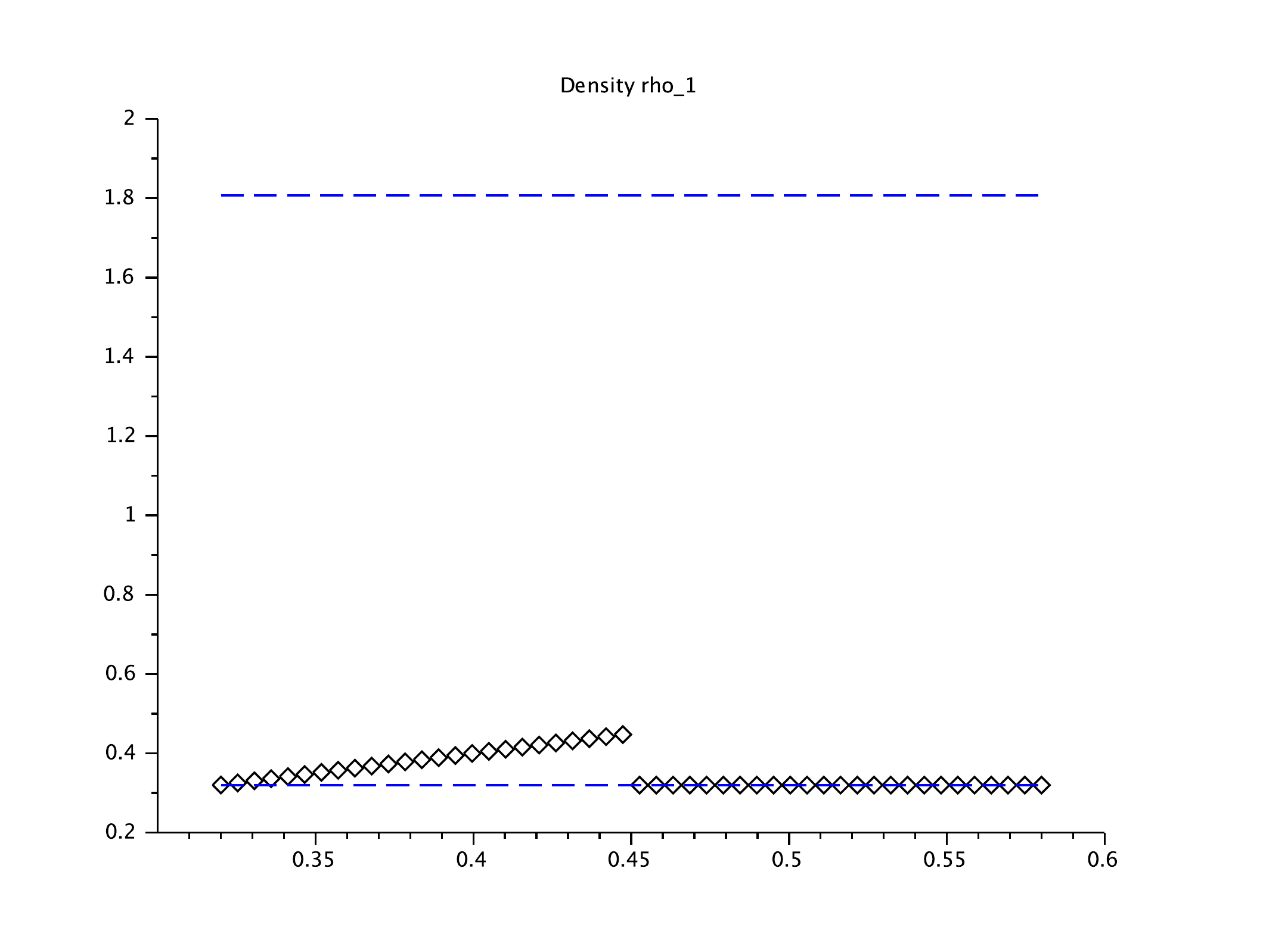}
  \includegraphics[width=6cm]{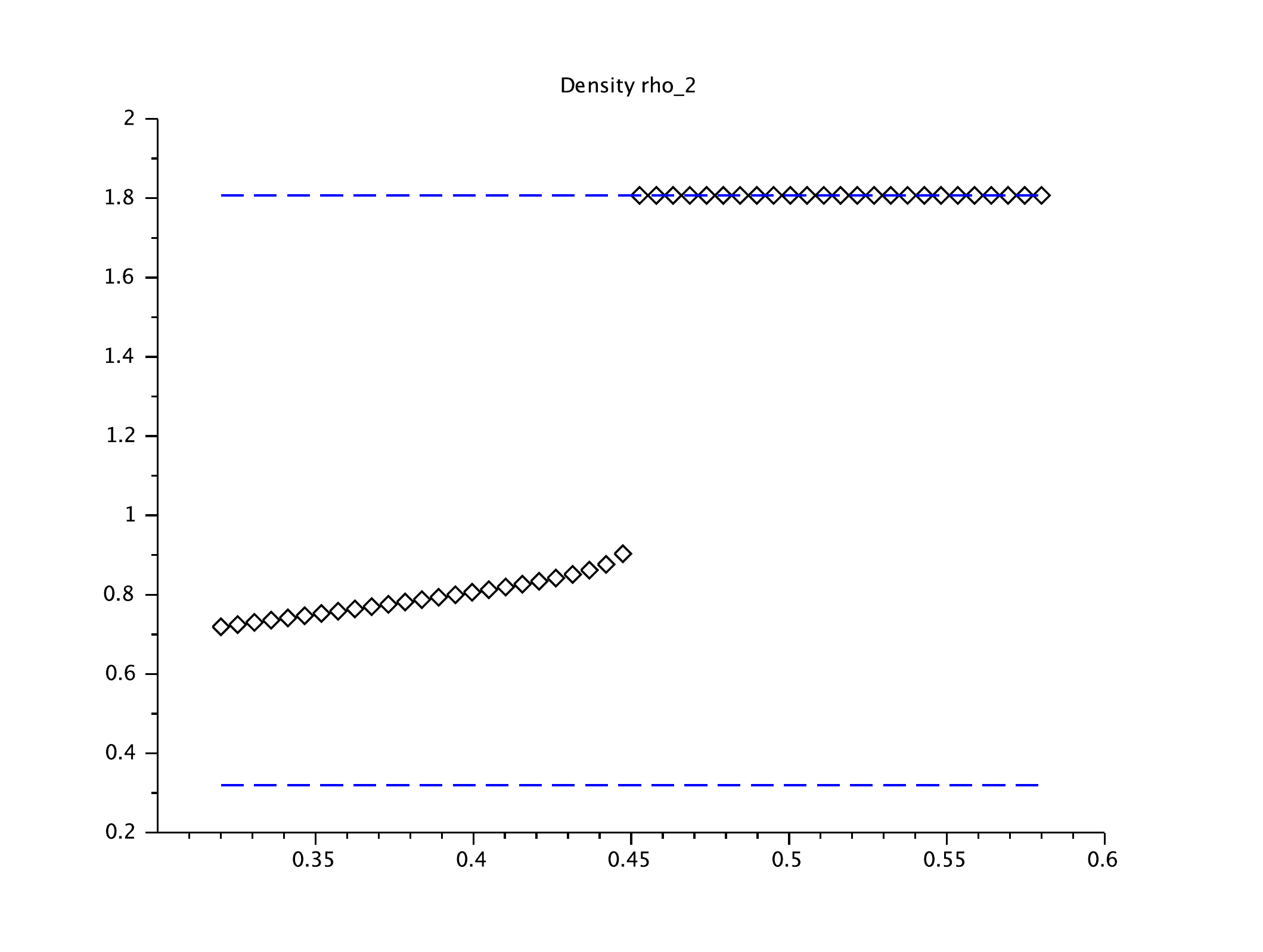}
  \caption{Numerical illustration of a perturbation within the
    metastable vapor zone. The initial density $\rho(0)$ takes on 50 values in
  $[\rho_1^*,\rho^-]$ and $\rho_i(0)$ are perturbations of
  $\rho(0)$ under the assumption~\eqref{eq:H1}. The top-left figure
  corresponds to the mixture pressure at final time $T_f=10^3$. For
  densities $\rho(0)<0.45$, the pressure coincides with the reduced van
  der Waals pressure, while for $\rho(0)>0.45$ it matches with the
  Maxwell line. One notices that the volume fraction $\alpha_1(T_f)$
  is either constant equal to 1 (for $\rho(0)<0.45$), correspondong to
  pure phase 1, or takes on values in ]0,1[, which means that the system
  reached the coexistent state.}
  \label{fig:sci_pert_meta}
\end{figure}

\textbf{Perturbation of the density $\rho$ in the whole domain}:
Figure~\ref{fig:sci_meta} corresponds to an initial density $\rho(0)$
which takes on 50
values between 0.2 and 1.8, while $\rho_1(0) =\rho(0)-0.1$ and
$\rho_2(0) = \rho(0) +0.1$.
One observes that an initial perturbation of the density $\rho(0)$ leads
to final states which belong to either pure vapor/liquid states,
including metastable states, or the coexistent state. Hence the mixture
pressure coincides with the admissible branches of the van der Waals
pressure curve or with the Maxwell line.

However the convergence is not obvious for $\rho(0)$
close to $\rho^-$ on the left (resp. to $\rho^+$ on the right). 
This can be observed in Figure \ref{fig:sci_meta} top-right (plot of the volume fraction), 
where the parts on the left and on the right
should be straight lines. 
Actually, since for $\rho$ close to $\rho^-$ the perturbation chosen is in the spinodal zone, 
we expect the equilibrium to be on the Maxwell line, 
which is not the case. We suppose that the final time is not large enough to ensure the actual convergence.
As a matter of a fact, we observed that the requested time to reach convergence 
is larger in the metastable zone than in the spinodal one.

\begin{figure}[h]
  \centering
  \includegraphics[width=6cm]{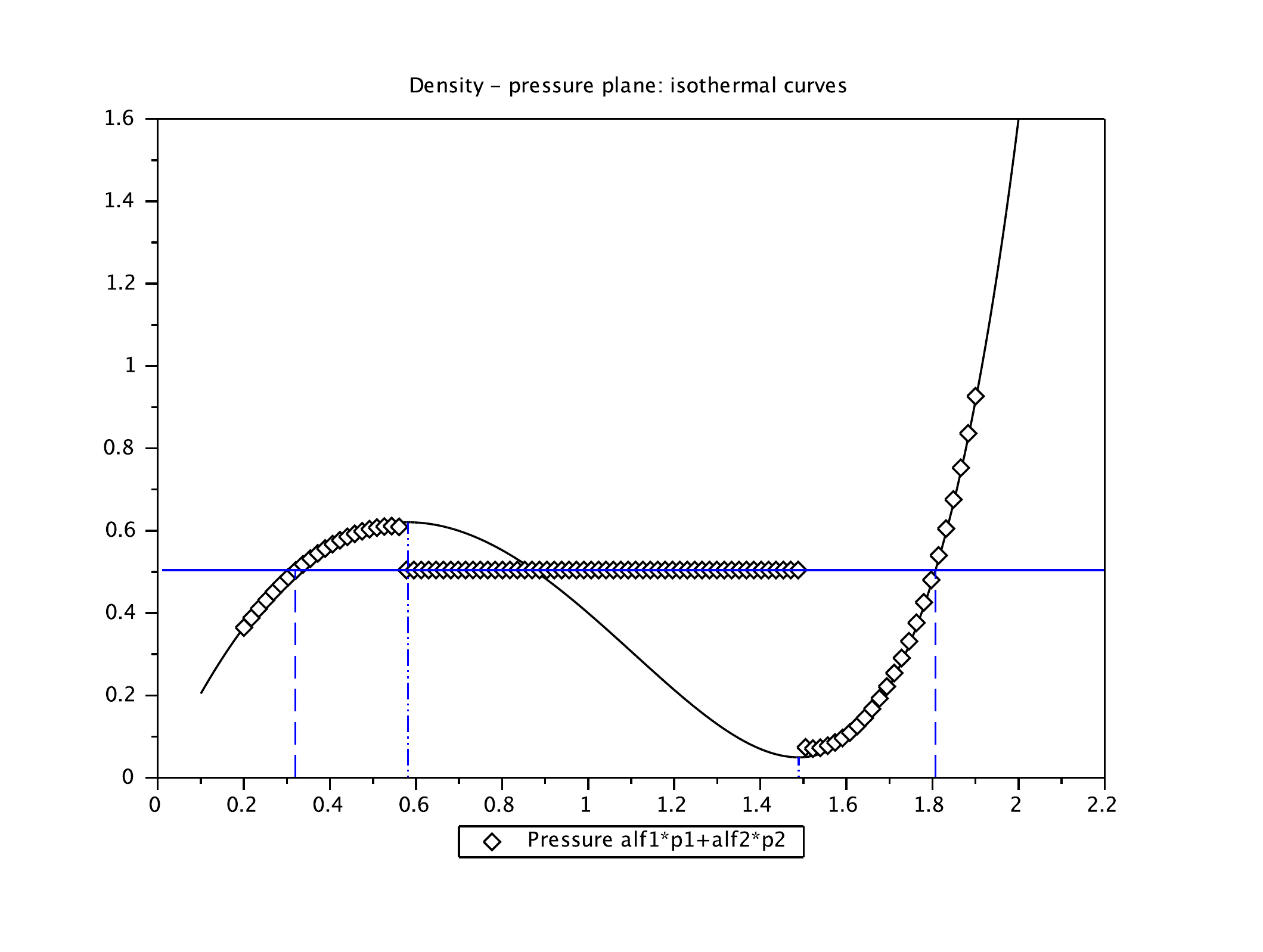}
  \includegraphics[width=6cm]{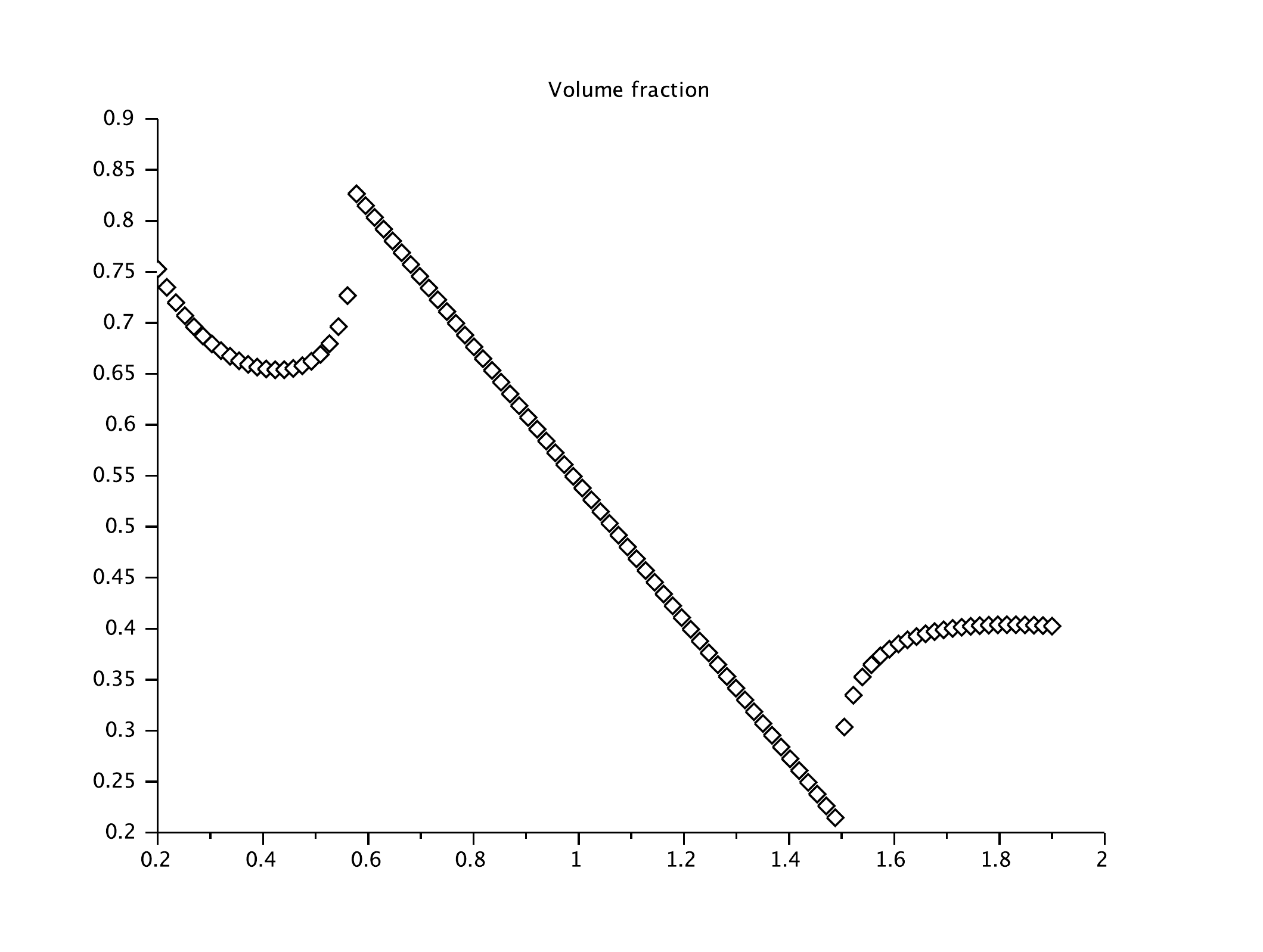}
  \includegraphics[width=6cm]{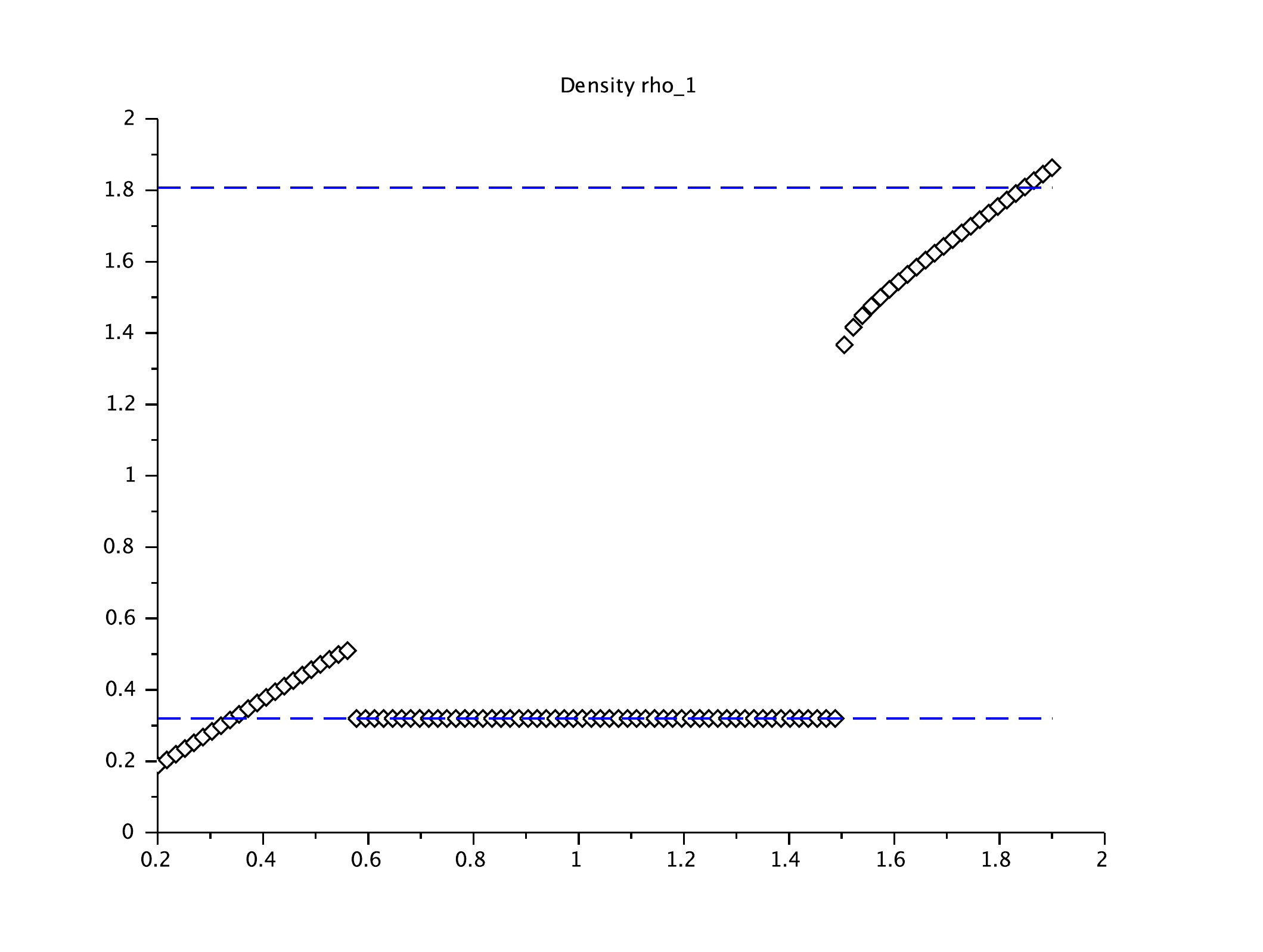}
  \includegraphics[width=6cm]{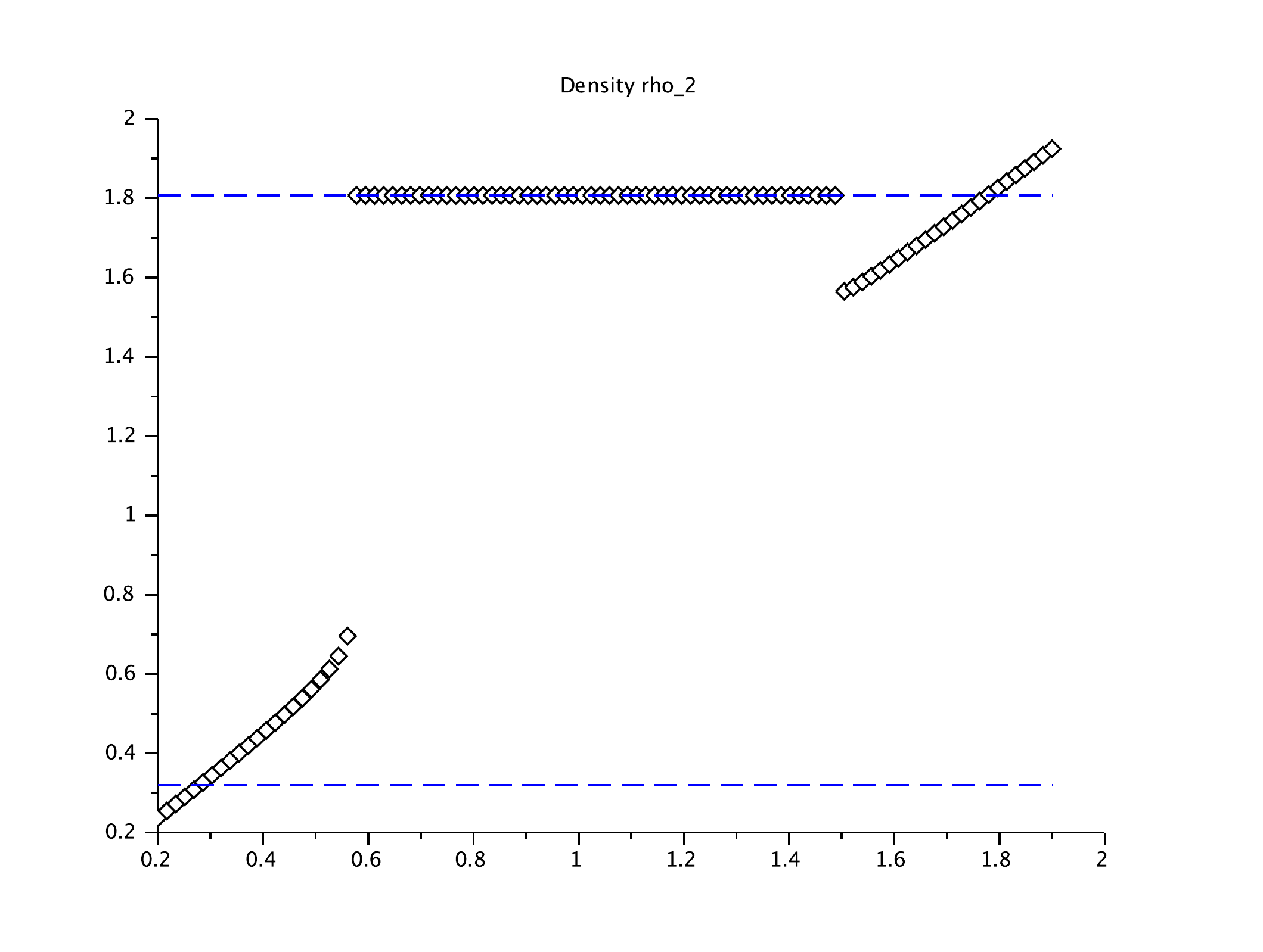}
  \caption{Numerical illustration of a perturbation of $\rho(0)$ in
    the whole domain. The mixture pressure $\alpha_1(T_f)p(\rho_1(T_f)) +
    \alpha_2(T_f)p(\rho_2(T_f))$ coincides with the admissible branches
  of the reduced van der Waals pressure in the pure liquid/vapor
  states, including metastable state. In the spinodal zone it
  corresponds to the Maxwell line.}
  \label{fig:sci_meta}
\end{figure}
%--------------------------------------------------
% Section 2 : pdes model
%--------------------------------------------------
\section{The isothermal model}
\label{sec:model}

This section is devoted to the definition and study of a $4\times 4$ van
der Waals isothermal two-phase flow model.
Since we are interested in the modeling of phase
transitions with possible metasta\-ble states, the liquid-vapor flows that we consider are submitted to
strong thermodynamical perturbations. Hence we propose to depict the
dynamic of the flows by a compressible averaged model, namely Euler
type equations.

In order to 
model phase transitions, the hydrodynamic part of
the model is classically coupled with a relaxation source term which carries on the
mass transfer. Since we wish to take into account possible metastable states, 
the equilibria of the source term have to be either pure liquid/vapor states,
metastable states or the coexistence state given by
\eqref{eq:maxwell_mu}-\eqref{eq:eq_p}.
Hence we propose
a coupling between the dynamical system \eqref{eq:syst_dyn2}
introduced in the previous section and a modified version of the
isothermal two-phase model proposed in \cite{chalons09}.

After defining the model, we study several properties of the system, such as existence of a decreasing energy, hyperbolicity and Riemann invariants
for the homogeneous system. Notice at once that we have only partial results for hyperbolicity, as noticed before in the literature, because of the spinodal zone.
This leads to a formal study of invariant hyperbolicity domains in the last subsection.

\subsection{Definition of the model}
The basic isothermal Euler system  contains the balance equations accounting for the
conservation of the total mass and the total momentum of the two-phase
flow. We propose to extend this system with two equations describing the evolution
 of the partial densities $\rho_1$ and $\rho_2$ which are now functions  of time $t$
and space $x$. The two phases evolve with the same velocity $u$.
The momentum equation involves a pressure term which is the mixture pressure
$\alpha_1p(\rho_1)+\alpha_2p(\rho_2)$. Here $\alpha_1$ and $\alpha_2$ are given by \eqref{eq:alpha_i}
but for the sake of readability we skip this dependence in what follows.

The system we propose is the following
\begin{equation}
  \label{eq:model}
  \begin{cases}
    \p_t \rho + \p_x (\rho u) = 0,\\
    \p_t \rho_1 + \p_x(\rho_1 u
    )=-\dfrac{1}{\varepsilon}(\rho-\rho_1)(\rho-\rho_2)f(\rho_2|\rho_1),\\
    % -\alpha_1\dfrac{1}{\varepsilon}(\rho-\rho_1)(\rho-\rho_2)
    % (\rho_2(\mu(\rho_2)-\mu(\rho_1)) \\
    %\qquad \qquad \qquad \qquad+ p(\rho_1)-p(\rho_2)),\\
    \p_t \rho_2 + \p_x(\rho_2 u
    )=
    \dfrac{1}{\varepsilon}(\rho-\rho_1)(\rho-\rho_2)f(\rho_1|\rho_2),\\
    % -\alpha_2\dfrac{1}{\varepsilon}(\rho-\rho_1)(\rho-\rho_2)
    % (\rho_1(\mu(\rho_2)-\mu(\rho_1)) \\
    % \qquad \qquad \qquad \qquad+ p(\rho_1)-p(\rho_2)),\\
    \p_t (\rho u) + \p_x (\rho u^2 + \alpha_1 p(\rho_1) + \alpha_2 p
    (\rho_2)) = 0,
  \end{cases}
\end{equation}
where $\varepsilon$ is a relaxation parameter that determines the rate
at which the chemical potentials and pressures of the two phases reach
equilibrium.
The chemical potential $\mu$ and the pressure $p$ follow the van der
Waals model \eqref{eq:p_mu_vdw}. 
The source terms on the partial densities equations are exactly those of
\eqref{eq:syst_dyn2}, and involve the relative free energy
  $f(\rho_i|\rho_j)$ which is defined in~\eqref{eq:relat_f}.
System~\eqref{eq:model} is supplemented with initial
conditions on the velocity $u$ and on the densities $\rho$ and
$\rho_i$, $i=1,2$ satisfying the assumption \eqref{eq:H1}.

Combining the mass conservation
  equation~\eqref{eq:model}-1 and the
  equations on the partial densities $\rho_i$, $i=1,2$, one can compute the
  equation satisfied by the volume fraction $\alpha_1$
  \begin{equation}
    \label{eq:eq_alpha}
    \p_t \alpha_1 + u \p_x \alpha_1 =
    \dfrac 1
    \varepsilon\dfrac{(\rho-\rho_1)(\rho-\rho_2)}{\rho_1-\rho_2}
    \left[ \alpha_1 f(\rho_2|\rho_1) -\alpha_2
      f(\rho_1|\rho_2)\right].
  \end{equation}

From~\eqref{eq:model} one can also recover the two equations
on the partial masses $\alpha_i\rho_i$ and deal with a system of the classical form
\begin{equation*}
  \begin{cases}
    \p_t \rho + \p_x (\rho u) = 0,\\
    \p_t (\alpha_1\rho_1) + \p_x(\alpha_1\rho_1 u
    =\dfrac 1
    \varepsilon\dfrac{(\rho-\rho_1)(\rho-\rho_2)}{\rho_1-\rho_2}
    \left[ \alpha_1 \rho_2 f(\rho_2|\rho_1) - \alpha_2 \rho_1 f(\rho_1|\rho_2)\right],\\
    \p_t (\alpha_2\rho_2) + \p_x(\alpha_2\rho_2 u
    =-\dfrac 1
    \varepsilon\dfrac{(\rho-\rho_1)(\rho-\rho_2)}{\rho_1-\rho_2}
    \left[ \alpha_1 \rho_2 f(\rho_2|\rho_1) - \alpha_2 \rho_1 f(\rho_1|\rho_2)\right],\\
    \p_t (\rho u) + \p_x (\rho u^2 + \alpha_1 p(\rho_1) + \alpha_2 p
    (\rho_2)) = 0.
  \end{cases}
\end{equation*}

In the present paper we choose to focus on the system~\eqref{eq:model},
which allows us to define the densities $\rho_i$ even when $\alpha_i=0$, 
which is not the case in the last system.

\begin{Rem}
An interesting feature is that the system boils down to the classical
$p-$system in pure phases that is when $\alpha_1\alpha_2=0$, including the metastable regions.
\end{Rem}

%------------------------------------------------------------
\subsection{Hyperbolicity and entropy for the homogeneous system}
\label{sec:hyperbolicity}
We introduce the mechanical energy 
\begin{equation}
  \label{eq:nrj_math}
  \mathcal E (\rho, \rho_1, \rho_2, u) = \dfrac{\rho u^2}{2} +
  \mathcal F(\rho,\rho_1,\rho_2),
\end{equation}
where the total Helmoltz free energy $\mathcal F$ is defined by
\eqref{eq:F}. The first result we obtain is the decrease in time of this energy. 

\begin{Prop}
  The function $\mathcal E$, defined in \eqref{eq:nrj_math}, satisfies
  the following equation
  \begin{equation}
    \label{eq:comp_bal_law}
    \p_t \mathcal E + \p_x (u(\mathcal E + \alpha_1
    p(\rho_1) + \alpha_2 p(\rho_2)) \leq 0.
  \end{equation}
\end{Prop}

\begin{proof}
  On the one hand, using the notation~\eqref{eq:r}, one has
  \begin{equation*}
      \p_t \mathcal F(\mathbf{r}) = \nabla_{\mathbf{r}} \mathcal F(\mathbf{r})\p_t \mathbf{r}
       = -\nabla_{\mathbf{r}} \mathcal F (\mathbf{r})\p_x (u\mathbf{r})
       +\nabla_{\mathbf{r}} \mathcal F(\mathbf{r}) \cdot Q,
  \end{equation*}
  where $Q = \dfrac 1 \varepsilon \left(0, -(\rho-\rho_1)(\rho-\rho_2)f(\rho_2|\rho_1) ,(\rho-\rho_1)(\rho-\rho_2)
  f(\rho_1|\rho_2) \right)^t $.
  Hence it comes
  \begin{equation*}
    \p_t \mathcal F(\mathbf{r}) = -u \p_x \mathcal F(\mathbf{r}) - \nabla_{\mathbf{r}} \mathcal
      F(\mathbf{r}) \cdot \mathbf{r}\p_x u +\nabla_{\mathbf{r}}
      \mathcal F(\mathbf{r})\cdot Q.
  \end{equation*}
  Now the expression of $\nabla{_\mathbf r} \mathcal F(\mathbf r)$
  given in \eqref{eq:gradF2} leads to
  \begin{equation*}
    \nabla{_\mathbf{r}} \mathcal F(\mathbf{r}) \cdot \mathbf{r} =
    \dfrac{1}{\rho_1-\rho_2 } \left( ((f(\rho_1) -f(\rho_2))\rho +
      \alpha_1 \rho_1 f(\rho_2|\rho_1) - \alpha_2 \rho_2 f(\rho_1|\rho_2)\right).
  \end{equation*}
  Accounting on the definition of the relative free energy
  \eqref{eq:relat_f}, it yields
  \begin{equation*}
    \nabla{_\mathbf r} \mathcal F(\mathbf r) \cdot \mathbf r =
    \mathcal F(\mathbf r) + \alpha_1p(\rho_1)+\alpha_2p(\rho_2).
  \end{equation*}
  Thus one has
  \begin{equation*}
    \begin{aligned}
      \p_t \mathcal F(\mathbf{r}) &= -u\p_x \mathcal F(\mathbf{r}) -
      (\mathcal F(\mathbf{r}) +
      \alpha_1p(\rho_1)+\alpha_2p(\rho_2))\p_x u +\nabla_{\mathbf{r}}
      \mathcal F(\mathbf{r})\cdot Q.\\
      &= -\p_x (\mathcal F(\mathbf{r}) u) - (\alpha_1p(\rho_1)+\alpha_2p(\rho_2))\p_x u +\nabla_{\mathbf{r}}
      \mathcal F(\mathbf{r})\cdot Q.
    \end{aligned}
  \end{equation*}
  On the other hand, a classical Euler type computation gets
  \begin{equation*}
    \p_t \dfrac{\rho u^2}{2} + \p_x \left( u \dfrac{\rho u^2}{2}
    \right) +u \p_x (\alpha_1p(\rho_1)+\alpha_2p(\rho_2))=0.
  \end{equation*}
Combining the previous two relations gives
  \begin{equation*}
    \p_t \mathcal E + \p_x (u(\mathcal
    E+\alpha_1p(\rho_1)+\alpha_2p(\rho_2)))= \mathcal
    F(\mathbf{r})\cdot Q\leq 0,
  \end{equation*}
where the final inequality follows again from the expression~\eqref{eq:gradF2} of $\nabla_{\mathbf{r}}\mathcal F(\mathbf{r})$.
\end{proof}

However, the function $\mathcal E$ is
not convex everywhere, so that it cannot be considered as a mathematical entropy. Indeed 
$\mathcal E$ is convex where $\mathcal F$ is, and we have the following result.

\begin{Thm}
  \label{thm:convex_F}
  The total Helmholtz free energy $\mathcal F$ defined by \eqref{eq:F} is convex
  for $\rho\in]0,3[$, $u\in\R$ and
  \begin{itemize}
  \item $(\rho_1,\rho_2)\in (0,\rho) \times\big((\rho,\rho^-)\cup
    (\rho^+,3)\big)$, if $\rho \le \rho^-$
  \item $(\rho_1,\rho_2)\in (0,\rho^-) \times (\rho^+,3)$, if $\rho \in
    (\rho^-,\rho^+)$
  \item $(\rho_1,\rho_2)\in \big((0,\rho^-)
    \times(\rho^+,\rho)\big)\times (\rho,3)$, if $\rho\geq\rho^+$
  \end{itemize}
\end{Thm}

\begin{proof}
  The function $\mathcal F$ is a convex combination of
  $f(\rho_1)$ and $f(\rho_2)$ where $f$ is the intensive Helmholtz free
  energy \eqref{eq:gibbs_intensive}. 
  By definition of 
  $\rho_-$ and $\rho_+$ (see Figure \ref{fig:vdw_isoth}), $f$ is convex
  on $(0,\rho^-] \cup [\rho^+,3)$, so the result follows.
\end{proof}

We turn now to the determination of the eigenvalues of the homogeneous
system~\eqref{eq:model}.
If we set $\mathbf{Y}=(\rho,\rho_1,\rho_2,u)$, for smooth solutions, the
homogeneous system can be written as
\begin{equation}
  \label{eq:syst_prim}
  \p_t \mathbf{Y} + A(\mathbf{Y}) \p_X \mathbf{Y} = 0,
\end{equation}
where the matrix $A(\mathbf{Y})$ is defined by
\begin{equation}
  \label{eq:AY}
  A(\mathbf{Y}) =
  \begin{pmatrix}
    u & 0 & 0 & \rho\\
    0 & u & 0 & \rho_1\\
    0 & 0 & u & \rho_2\\
    A_1(\mathbf{Y}) & A_2(\mathbf{Y}) & A_3(\mathbf{Y}) & u
  \end{pmatrix},
\end{equation}
and 
\begin{equation}
  \label{eq:ABC}
  \begin{aligned}
    A_1(\mathbf{Y}) &= \dfrac{p(\rho_1)-p(\rho_2)}{\rho(\rho_1-\rho_2)},\\
    A_2(\mathbf{Y}) &=
    \dfrac{\alpha_1}{\rho(\rho_1-\rho_2)}(p(\rho_2)-p(\rho_1)) +
    \dfrac{\alpha_1}{\rho}p'(\rho_1),\\
    A_3(\mathbf{Y}) &=
    \dfrac{\alpha_2}{\rho(\rho_1-\rho_2)}(p(\rho_2)-p(\rho_1)) +
    \dfrac{\alpha_2}{\rho}p'(\rho_2).
  \end{aligned}
\end{equation}

The characteristic equation of $A(\mathbf{Y})$ is given by
$$(u-\lambda)^2(u-c-\lambda)(u+c-\lambda),$$ with the speed of sound
\begin{equation}
  \label{eq:speed_c}
  c:=c(\mathbf{r})=\sqrt{\dfrac 1 \rho
\left(\alpha_1(\rho,\rho_1,\rho_2) \rho_1 p'(\rho_1) +
  \alpha_2(\rho,\rho_1,\rho_2) \rho_2 p'(\rho_2)\right)}.
\end{equation}
Thus we obtain
three distinct eigenvalues for the matrix $A(\mathbf{Y})$:

\begin{equation}
  \label{eq:eigval}
  \lambda_1(\mathbf{Y}) = u-c,\quad
  \lambda_2(\mathbf{Y})=\lambda_3(\mathbf{Y}) = u,\quad
  \lambda_4(\mathbf{Y}) = u+c.
\end{equation}

Note that the eigenvalues are real if $\mathbf r$ satisfies
\begin{equation*}
  \alpha_1(\mathbf r) \rho_1 p'(\rho_1) +
  \alpha_2(\mathbf r) \rho_2 p'(\rho_2) \geq 0.
\end{equation*}
Accounting on relations~\eqref{eq:fprim_rho}, it is equivalent to 
the following hyperbolicity condition
\begin{equation}
  \label{eq:hyp_cond}
    \alpha_1(\mathbf r) \rho_1^2 \mu'(\rho_1) +
  \alpha_2(\mathbf r) \rho_2^2 \mu'(\rho_2) \geq 0.
\end{equation}

The right eigenvectors $r_i(\mathbf{Y})$, $i=1,\ldots, 4$, that satisfy
$A(\mathbf{Y)}r_i(\mathbf{Y}) = \lambda_i(\mathbf{Y})r_i(\mathbf{Y})$
can be chosen as

\begin{equation}
  \label{eq:right_eigvect}
  r_1(\mathbf{Y}) =
  \begin{pmatrix}
   -\dfrac{\rho}{c}\\\\  -\dfrac{\rho_1}{c} \\\\  -\dfrac{\rho_2}{c} \\\\ 1
  \end{pmatrix}, \quad
  r_2(\mathbf{Y}) =
  \begin{pmatrix}
    -\dfrac{A_3}{A_1}\\\\ 0 \\  1 \\ 0
  \end{pmatrix}, \quad
  r_3(\mathbf{Y}) =
  \begin{pmatrix}
     -\dfrac{A_2}{A_1}\\\\ 1 \\  0 \\ 0
  \end{pmatrix}, \quad
  r_4(\mathbf{Y}) =
  \begin{pmatrix}
    \dfrac{\rho}{c}\\\\  \dfrac{\rho_1}{c} \\\\  \dfrac{\rho_2}{c} \\\\ 1
  \end{pmatrix},
\end{equation}
where the quantities $A_1, \; A_2$ and $A_3$ are defined
in~\eqref{eq:ABC}.

% The left eigenvectors $l_i(\mathbf{Y})$, $i=1,\ldots, 4$, of
% $A(\mathbf{Y})$ that satisfy $^tA(\mathbf{Y})l_i(\mathbf{Y}) =
% \lambda_i(\mathbf{Y})l_i(\mathbf{Y})$ are given by

% \begin{equation}
%   \label{eq:left_eigvect}
%   l_1(\mathbf{Y}) =
%   \begin{pmatrix}
%     -\dfrac{A_1(\mathbf{Y})}{c}\\  -\dfrac{A_2(\mathbf{Y})}{c} \\
%     -\dfrac{A_3(\mathbf{Y})}{c} \\ 1
%   \end{pmatrix}, \quad
%   l_2(\mathbf{Y}) =
%   \begin{pmatrix}
%     -\dfrac{\rho_2}{\rho}\\ 0 \\  1 \\ 0
%   \end{pmatrix}, \quad
%   l_3(\mathbf{Y}) =
%   \begin{pmatrix}
%      -\dfrac{\rho_1}{\rho}\\ 1 \\  0 \\ 0
%   \end{pmatrix}, \quad
%   l_4(\mathbf{Y}) =
%   \begin{pmatrix}
%     \dfrac{A_1(\mathbf{Y})}{c}\\  \dfrac{A_2(\mathbf{Y})}{c} \\
%     \dfrac{A_3(\mathbf{Y})}{c} \\ 1
%   \end{pmatrix}.
% \end{equation}

If the densities $\rho,\rho_1$ and $\rho_2$ satisfy \eqref{eq:H1} and
\eqref{eq:hyp_cond},
the matrix $A(\mathbf{Y})$ is diagonalizable in
$\R$ and its eigenvectors span the whole space $\R^4$ so that the
system is hyperbolic.

%------------------------------------------------------------
\subsection{Structure of the waves}
\label{sec:struct-waves}

In this paragraph we study the structure of the waves. 
Assuming that the densities $\rho,\rho_1$ and $\rho_2$ satisfy
\eqref{eq:H1} and \eqref{eq:hyp_cond},
one can observe that the waves are either
genuinely non linear or linearly degenerate.

Straightforward computations lead to the following property which will
be useful in the sequel.
\begin{Prop}
  The speed of sound $c$,
  function of state $\mathbf{r}$, satisfies the following properties
  relations
\begin{eqnarray}
    \label{eq:p_c}
    \nabla p(\mathbf{r}) \cdot \mathbf{r}
    &=&\rho c^2(\mathbf{r}),\\
    \nabla c(\mathbf{r})&=& \dfrac{1}{2 c(\mathbf{r})}
    \begin{pmatrix}
      -\dfrac{c(\mathbf{r})^2}{\rho}+ \dfrac{1}{\rho
        (\rho_1-\rho_2)}(\rho_1 p'(\rho_1) -\rho_2p'(\rho_2))\\
      \dfrac{\alpha_1(\mathbf{r})}{\rho} \left(
        \dfrac{\rho_2p'(\rho_2) -
          \rho_1p'(\rho_1)}{\rho_1-\rho_2}
        +p'(\rho_1)+\rho_1p''(\rho_1) \right)\\
      \dfrac{\alpha_2(\mathbf{r})}{\rho} \left(
        \dfrac{\rho_2p'(\rho_2) -
          \rho_1p'(\rho_1)}{\rho_1-\rho_2} +p'(\rho_2)+\rho_2p''(\rho_2) \right)
    \end{pmatrix}.
    \label{eq:grad_c}
\end{eqnarray}
\end{Prop}

Let us start with the waves associated to the wave speed $u-c$ and
$u+c$.

\begin{Prop}
  The characteristic fields associated to the waves speed
  $\lambda_1(\mathbf{Y})=u-c$
  and $\lambda_4(\mathbf{Y})=u+c$ are genuinely non linear \textit{i.e.} 
  $\nabla_\mathbf{Y} \lambda_1(\mathbf{Y}) \cdot r_1(\mathbf{Y})\neq 0$ and 
  $\nabla_\mathbf{Y} \lambda_4(\mathbf{Y}) \cdot r_4(\mathbf{Y})\neq 0$ for admissible state
  vector $\mathbf{Y}$ that is for densities $(\rho,\rho_1,\rho_2)$
  satisfying \eqref{eq:H1} and \eqref{eq:hyp_cond}.
\end{Prop}

\begin{proof}
  We introduce the notation $D(\mathbf{r})= (c(\mathbf{r}))^2$.
  We consider the wave associated to the eigenvalue
  $\lambda_1(\mathbf{Y})$.
  One has
  \begin{equation}
    \label{eq:gen_non_lin_field}
    \begin{aligned}
      \nabla_\mathbf{Y}\lambda_1(\mathbf{Y})\cdot
      r_1(\mathbf{Y}) &=
      \begin{pmatrix}
        \nabla_\mathbf{r}c(\mathbf{r}) \\ 1
      \end{pmatrix}
      \cdot r_1(\mathbf{Y})\\
      &=-\dfrac{1}{2c(\mathbf{r})^2} \left(\rho\dfrac{\p D}{\p \rho} +
        \rho_1 \dfrac{\p D}{\p \rho_1} +\rho_2 \dfrac{\p D}{\p \rho_2}
      \right) +1\\
      &=-\dfrac{1}{2\rho c(\mathbf{r})^2}(\alpha_1(\mathbf{r})\rho_1^2p''(\rho_1)
      + \alpha_2(\mathbf{r})\rho_2^2p''(\rho_2))+1.
    \end{aligned}
  \end{equation}
  The densities are assumed to be strictly positive.
  Under the hypothesis~\eqref{eq:H1}  and \eqref{eq:hyp_cond} the mass
  fractions $\alpha_i$ are positive and the second derivative of
  the van der Waals pressure \eqref{eq:p_mu_vdw} is a strictly negative
  function of the density. Thus 
  $\nabla_\mathbf{Y}\lambda_1(\mathbf{Y})\cdot r_1(\mathbf{Y}) \neq 0$.
  Similarly we can state that
  $\nabla_\mathbf{Y}\lambda_4(\mathbf{Y})\cdot r_4(\mathbf{Y}) \neq 0$
  that conclude the proof.
\end{proof}

We now study the wave associated to the speed $u$.

\begin{Prop}
  The characteristic fields associated to the waves
  $\lambda_2(\mathbf{Y})=\lambda_3(\mathbf{Y})=u$
  linearly degenerate \textit{i.e.} 
  $\nabla_\mathbf{Y} \lambda_2(\mathbf{Y}) \cdot r_2(\mathbf{Y})= 0$ and 
  $\nabla_\mathbf{Y} \lambda_3(\mathbf{Y}) \cdot r_3(\mathbf{Y})= 0$ for admissible state
  vector $\mathbf{Y}$ that is for densities $(\rho,\rho_1,\rho_2)$
  satisfying \eqref{eq:hyp_cond}.
\end{Prop}

\begin{proof}
  We deduce from the eigenvalues \eqref{eq:eigval} the relation
  \begin{equation}
    \label{eq:lin_deg_field}
    \nabla_\mathbf{Y}\lambda_i(\mathbf{Y}) \cdot r_i(\mathbf{Y}) =
    (0,0,0,1)^T\cdot r_i(\mathbf{Y}),
  \end{equation}
  for $i=\{1,2\}$.
  Then introducing the right eigenvectors~\eqref{eq:right_eigvect}
  in~\eqref{eq:lin_deg_field}, it is easily checked that
  $\nabla_\mathbf{Y}\lambda_i(\mathbf{Y}) \cdot r_i(\mathbf{Y}) =0$ 
  for  $i=\{1,2\}$ and this complete the proof.
\end{proof}

We now address the determination of the Riemann invariants of the
system. These computations are made easier using the following property.

\begin{Prop}
  The mass and volume fractions $\alpha_1$ and $\varphi_1$, defined
  by~\eqref{eq:alpha_i} and~\eqref{eq:varphi_i}, satisfy the
  following non conservative equations
  \begin{equation}
    \label{eq:edp_fraction}
    \begin{aligned}
      \p_t \alpha_1 + u\p_x \alpha_1=0,  \\
      \p_t \varphi_1 + u\p_x \varphi_1=0.
    \end{aligned}
  \end{equation}
\end{Prop}

\begin{Prop}
  The Riemann invariants associated to the wave of speed $u$ are 
  \begin{equation}
    \label{eq:RI_u}
    \{u,\bar p\},
  \end{equation}
  {with $\bar p(\mathbf{r}) = \alpha_1(\mathbf r) p(\rho_1) +
    \alpha_2(\mathbf r) p(\rho_2)$.}
  The volume and mass fractions are Riemann invariants associated to the wave of speed $u\pm c$: 
  \begin{equation}
    \label{eq:RI_u+c}
    \{\alpha_1,\varphi_1\}.%,u\mp g(\mathbf{r})\},
  \end{equation}
%  where the function
%\textcolor{red}{devrait appara\^itre dans la d\'emonstration, il manque une formule pour $g$ ici}
%  $g(\mathbf{r})$ is defined by
%  \begin{equation}
%    \label{eq:g}
%  \nabla_\mathbf{r} g(\mathbf{r})= \dfrac{1}{\rho
%      c(\mathbf{r)}}\nabla_\mathbf{r} \bar p(\mathbf{r}).
%  \end{equation}
\end{Prop}

\begin{proof}
  Because the field associated to the speed $u$ is linearly
  degenerate, $u$ is clearly a Riemann invariant for this wave.
  Using the gradient of $\bar p$ with respect to $\mathbf{r}$, a straightforward
  computation gives
  \begin{equation}
    \label{eq:RI_u_p}
    \nabla_\mathbf{Y} \bar p(\mathbf{Y}) \cdot r_2(\mathbf{Y})=0.
  \end{equation}
  
  On the other hand the volume fractions $\alpha_i$ and the mass
  fractions $\varphi_i$ satisfy the equations~\eqref{eq:edp_fraction}. Thus the
  fractions are Riemann invariants for the waves of speed $u\pm c$.
%  In order to determine the last Riemann invariant one can observe
%  that $u+g(\mathbf{r})$ satisfies the following equation
%  \begin{equation}
%    \label{eq:RI_u+c_g}
%    \p_t (u+ g(\mathbf{r})) + (u+c)\p_x (u+ g(\mathbf{r}))=0.
%  \end{equation}
%  Thus $u\pm g(\mathbf{r})$ is a Riemann invariants for the wave of speed $u\mp
%  c$.
\end{proof}
The characterization of the third Riemann invariant for the waves of speed $u\pm c$ is more intricate
and is not addressed here.

\subsection{Invariant domains of hyperbolicity for the relaxed system}
\label{sec:dom_hyp}
%%%%%%%%%%%%%%%%%%%%%%%%%%%%%%%%%%%%%%%%%%%%%%%%%%%%%%%%%%%%%%%%%%%%%%%%%%%%%%%%%%%%%%%%%%%

According to Section \ref{sec:hyperbolicity}, it is
clear that the \textit{homogeneous} system \eqref{eq:model} is hyperbolic if
and only if the densities $\rho$, $\rho_1$ and $\rho_2$ satisfy
\eqref{eq:H1} and the constraint \eqref{eq:hyp_cond} on the speed of sound.
Nonetheless we are interested in the study and the numerical
approximation of the \textit{whole} relaxed system \eqref{eq:model}, that is
taking into account the relaxation term with a finite relaxation parameter
$\varepsilon > 0$. 
Actually the domains of hyperbolicity of \eqref{eq:model} strongly
depend on the attraction basins of the dynamical system
\eqref{eq:syst_dyn2}.
In the present section, we introduce the notion of invariant domains
in the same spirit as in \cite{Chueh77} for diffusive systems. 
We show that invariant domains $\Omega$ of
hyperbolicity for the relaxed system  \eqref{eq:model} are subsets of
the attraction basins of the dynamical system \eqref{eq:syst_dyn2}.
First note that the hyperbolicity of the homogeneous system
\eqref{eq:model} solely depends on the densities
$\mathbf{r}(t,x)=(\rho,\rho_1,\rho_2)^t(t,x)$, $\forall (t,x)\in
\R^+\times \R$, according to the constraint
\eqref{eq:hyp_cond} on the speed of sound and not on the velocity $u(t,x)$.
Hence we consider the following definition of an invariant region.
\begin{Def}
  \label{def:dom_inv}
  Let $\Omega=\{\mathbf{r}=(\rho,\rho_1,\rho_2)\in ]0,3[^3| \; 0<\rho_1\leq \rho\leq
  \rho_2 \text{ and } \rho_1<\rho_2\}$ a subset of the phase space
  $(\rho,\rho_1, \rho_2)$ with a
  Lipschitz continuous boundary $\p \Omega$. The region $\Omega$
  is said to be a invariant domain if
  \begin{equation*}
    \{\forall x\in \R, \; \mathbf{r}(0,x) \in \Omega \}
    \Leftrightarrow \forall t>0, \; \{ \forall x\in \R,\;
    \mathbf{r}(t,x) \in \Omega\}.
  \end{equation*}
\end{Def}
We now define some kind of indicator function for such a domain $\Omega$: let $S$ be defined by
\begin{equation}
  \label{eq:indicS}
  \begin{array}{rccl}
    S : & ]0,3[^3 &\to &\R\\
      &\mathbf{r }=(\rho,\rho_1,\rho_2)&\mapsto &\rho s(\rho_1/\rho,\rho_2/\rho),
  \end{array}
\end{equation}
where $s(\beta_1,\beta_2) = 1-\mathbf{1}_{\{\beta_i^-\le\beta_i\le\beta_i^+\}}$.
Obviously we have
  \begin{equation*}
    S(\mathbf{r})=0 \Leftrightarrow \mathbf{r}\in \Omega.
  \end{equation*}
Next we introduce the nonnegative quantity $J$
  \begin{equation}\label{eq:SJ}
    \begin{array}{lccl}
      J : &\R^+ &\to &\R\\
      &t &\mapsto &\int_\R S(r(t,x))dx.
    \end{array}
  \end{equation}

\begin{Prop}
  Consider $\Omega$ a subset of the phase space
  $(\rho,\rho_1, \rho_2)$ with a Lipschitz continuous boundary $\p
  \Omega$ and the associated function $S$ defined by
  \eqref{eq:indicS}.
  Then one has the following properties:
  \begin{enumerate}
  \item In the sense of distributions we have
    \begin{equation}\label{Stokes}
      \langle\nabla S,\phi\rangle = [S]\int_{\partial\Omega}{\bf n}\phi(\sigma)\,d\sigma,
    \end{equation}
    where $d\sigma$ is the surface measure on $\partial\Omega$, ${\bf n}$
    the outer normal of $\Omega$ and 

    \noindent $[S]=S_{\rm out}-S_{\rm in}$ is the jump of $S$
    across the boundary $\partial\Omega$.
  \item The function $S$ is positively homogeneous of degree 1 so that
    it verifies the Euler relation
    $S(\mathbf{r})=\nabla_\mathbf{r}S(\mathbf{r})\cdot \mathbf{r}$.
  \item The function $S$ satisfies
    \begin{equation}
      \label{eq:lcS}
        \langle\nabla_{\mathbf{r}}S(\mathbf{r}), \p_t \mathbf{r} +
        \p_x (u\mathbf{r})\rangle
        = \p_t S(\mathbf{r}) + \p_x (uS(\mathbf{r})).
    \end{equation}
  \end{enumerate}
\end{Prop}
\begin{proof}
  The first item is a consequence of the Stokes theorem. 
  By construction the function $S$ is positively homogeneous of degree
  1. Then it satisfies the Euler relation given in the second item.
  Finally following the same steps as in the energy estimate
  \eqref{eq:comp_bal_law}, we have
\begin{equation*}
  \begin{aligned}
    &\nabla_{\mathbf{r}}S(\mathbf{r}) \p_t \mathbf{r} +
    \nabla_{\mathbf{r}}S(\mathbf{r}) \p_x (u\mathbf{r})\\
    & =\p_t S(\mathbf{r}) + u\p_x S(\mathbf{r}) + \nabla_{\mathbf{r}}
    S(\mathbf{r})\cdot \mathbf{r}\p_x u\\
    &= \p_t S(\mathbf{r}) + \p_x (uS(\mathbf{r})), 
  \end{aligned}
\end{equation*}
where we use the above Euler relation for $S$ to obtain the last equality.
\end{proof}

We now relate the definition of an invariant domain to the functions $S$
and $J$ through several propositions.

\begin{Prop}
  \label{prop:caract_invdom1}
  The domain $\Omega$ is an invariant region if and only if
  \begin{equation*}
    \{ J(0) = 0 \Rightarrow \forall t>0, \quad J(t)=0\}.
  \end{equation*}
\end{Prop}
The proof of the Proposition \ref{prop:caract_invdom1} relies on the
following Lemma.
\begin{Lem} 
  \label{lem:dom_inv}
  Let $\Omega$ be a subset of the phase space
  $(\rho,\rho_1,\rho_2)$ and $S$ defined by \eqref{eq:indicS}. Then one has
  \begin{equation*}
    \forall x \in \R, \quad \mathbf{r}(.,x) \in \Omega \Leftrightarrow
    J(.)=\int_{\R} S(\mathbf{r}(.,x))dx=0.
  \end{equation*}
\end{Lem}
The proof of the Lemma relies on the definition of $S$ and its
positivity.

\begin{proof}[Proof of Proposition \ref{prop:caract_invdom1}]
  According to the Lemma \ref{lem:dom_inv} and by the definition of $J$
  \eqref{eq:SJ} of the quantity $J$, it follows
  \begin{equation*}
    \forall x\in \R, \quad r(0,x) \in \Omega \Leftrightarrow J(0)=0.
  \end{equation*}
  Using the Lemma \ref{lem:dom_inv} again, one gets
  \begin{equation*}
     \forall x\in \R, \quad r(t,x) \in \Omega \Leftrightarrow J(t)=0.
  \end{equation*}
  Combining these two equivalences leads to the conclusion.
\end{proof}

\begin{Prop}
  \label{prop:J_inv}
  Let $\Omega$ be a subset of the phase space
  $(\rho,\rho_1,\rho_2)$ and $J$ given by \eqref{eq:SJ}. Assume $J$
  is differentiable. Then it follows
  \begin{equation*}
    \left\{ \dfrac{d}{dt}J(t) \leq 0\right\} \Rightarrow \Omega \text{ is an
      invariant domain.} 
  \end{equation*}
\end{Prop}
\begin{proof}
  Assume that $J(0)=0$. By assumption on the time derivative of $J$, 
  $J(t)\leq J(0)$ and $J(t)\geq 0$
  by positivity. Thus $J(t)=0, \; \forall t$. Hence according to Proposition
  \ref{prop:caract_invdom1}, the domain $\Omega$ is invariant.
\end{proof}

\begin{Cor}
  Let $\Omega$ be a subset of the plan $(\rho,\rho_1,\rho_2)$ and $S$
  and $J$ the associated functions given by \eqref{eq:SJ}
  and \eqref{eq:indicS}.
  Denote $$Q=\dfrac 1 \varepsilon \left(0,
    -(\rho-\rho_1)(\rho-\rho_2)f(\rho_2|\rho_1)
    ,(\rho-\rho_1)(\rho-\rho_2)  f(\rho_1|\rho_2) \right)^t $$
  the right-hand side of the relaxed model \eqref{eq:model}.
  Then for any $\mathbf{r}\in \Omega$ such that $\lim\limits_{x\to +\infty}
  \mathbf{r}(.,x) = \lim\limits
_{x\to -\infty}
  \mathbf{r}(.,x)$, 
  one has the following assertions
  \begin{equation*}
    \langle Q,\nabla_\mathbf{r}S(\mathbf{r})\rangle \leq 0
    \quad\Rightarrow\quad
    \dfrac{d}{dt}J(t) \leq 0
    \quad\Rightarrow\quad
    \Omega \text{ is an invariant domain.}
  \end{equation*}
\end{Cor}
\begin{proof}
  Since $Q$ is the right-hand side of the relaxed model
  \eqref{eq:model}, it yields
  \begin{equation*}
    \langle Q,\nabla_\mathbf{r}S(\mathbf{r})\rangle = \langle \p_t
    \mathbf{r} + \p_x (u\mathbf{r}), \nabla_\mathbf{r}S(\mathbf{r})\rangle.
  \end{equation*}
  According to \eqref{eq:lcS} it follows that if $\langle
  Q,\nabla_\mathbf{r}S(\mathbf{r})\rangle \leq 0$ then
  \begin{equation*}
    \langle \p_t
    \mathbf{r} + \p_x (u\mathbf{r}),
    \nabla_\mathbf{r}S(\mathbf{r})\rangle \leq 0.
  \end{equation*}
  Integrating the above inequality on $\R$ gives
  $\int_R \p_t S(\mathbf{r}) dx \leq 0$, that is
  \begin{equation*}
    \dfrac{d}{dt}\int_R S(\mathbf{r}) dx =   \dfrac{d}{dt} J(t) \leq 0.
  \end{equation*}
  Proposition \ref{prop:J_inv} now leads to the conclusion.
\end{proof}

Hence in order to check that $\Omega$ is an invariant domain, one has
solely to verify that
\begin{equation*}
  \langle Q,
  \nabla_{\mathbf{r}} S(\mathbf{r})\rangle\leq 0.
\end{equation*}
Taking the scalar product of $\nabla S$ with the right-hand side $Q$
of the relaxation system we obtain
\begin{equation*}
\langle Q,
  \nabla_{\mathbf{r}} S(\mathbf{r})\rangle =\dfrac1\epsilon(\rho-\rho_1)(\rho-\rho_2)
\big(-\partial_{\rho_1}Sf(\rho_2|\rho_1)+\partial_{\rho_2}Sf(\rho_1|\rho_2)\big).
\end{equation*}
Using equation \eqref{Stokes}, we obtain formally
\begin{equation*}
\langle Q,
  \nabla_{\mathbf{r}} S(\mathbf{r})\rangle=\dfrac1\epsilon(\rho-\rho_1)(\rho-\rho_2)
\rho\big(-n_1f(\rho_2|\rho_1)+n_2f(\rho_1|\rho_2)\big),
\end{equation*}
where ${\bf n_\rho}=(n_1,n_2)$ is the outer normal of the
domain $\Omega_\rho=\{(\rho_1,\rho_2) |(\rho,\rho_1,\rho_2) \in \Omega
\}$. Note that the domain $\Omega_\rho$ are rectangles in the phase
space $(\rho_1,\rho_2)$.
Now checking the sign of $\langle  Q,
\nabla_{\mathbf{r}} S(\mathbf{r})\rangle$ is quite straightforward on
each part of the boundary $\p\Omega_\rho$.

\begin{figure}
  \begin{center}\resizebox{7cm}{!} {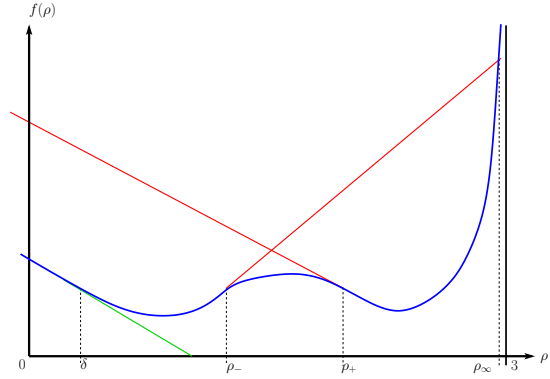}
  \end{center}
  \caption{Some reference points on the graph of the Helmoltz energy
    $f(\rho)$. The blue curve is a sketch of the graph of $f(\rho)$. The green
    line is the tangent to the graph of $f$ at the point $\delta$.
    The red lines are the tangents of the blue curve at $\rho^-$ and
    $\rho^+$.
    The dashed line is the tangent of the graph of $f$ at
    $\rho_\infty$ defined by $f(\rho_\infty,\rho^-)=0$.
    Depending on the position of the tangent line to the blue curve at
    a given point $\rho$, one can determine the sign of $f(.|\rho)$,
    see Proposition \ref{prop:graphf}.}
  \label{fig:graphf}
\end{figure}

\begin{figure}
  \begin{center}\resizebox{7cm}{!} {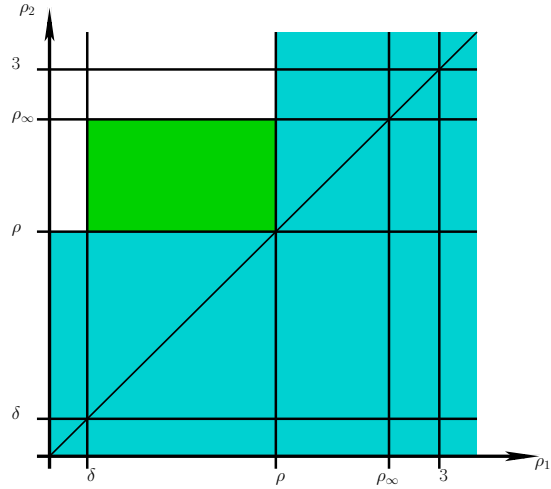}
  \end{center}
  \caption{\textit{A priori} estimate for $(\rho_1,\rho_2)$: the blue
    area refers to nonattainable states according to \eqref{eq:H1},
    the green zone is an invariant domain, providing in particular
    that void cannot appear.}
  \label{fig:Apriori}
\end{figure}
To characterize the invariant regions we study the sign of
the relative Helmoltz free energy $f(.|.)$ defined by
\eqref{eq:relat_f}.
According to Figure \ref{fig:graphf}, one can determine the sign of $f(.|.)$ according to the following proposition.
\begin{Prop}
  \label{prop:graphf}
  Let $\delta\in ]0,\rho_1^*[$ and $\rho_\infty\in [\rho_2^*,3[$ such
  that $f(\rho_\infty|\rho^-)=0$. Then the relative Helmoltz free
  energy satisfies
  \begin{itemize}
  \item $f(\rho|\delta)\geq 0$, $\forall \rho\in ]\delta,\rho_\infty[$,
  \item $f(\rho|\rho^-)\geq 0$ (resp. $\leq 0$), $\forall \rho\in
    ]\delta,\rho^-[$ (resp. $\rho\in ]\rho^-,\rho_\infty[$),
  \item $f(\rho|\rho^+)\leq 0$ (resp. $\geq 0$), $\forall
    \rho\in]\delta,\rho^+[$ (resp. $\rho\in ]\rho^+,\rho_\infty[$),
  \item $f(\rho|\rho_\infty)\geq 0$, $\forall\rho\in ]\delta, \rho_\infty[$.
  \end{itemize}
\end{Prop}
\begin{proof}
 The sign of $f(.|a)$ for any remarkable density $a$ depends on the
 positition of the tangent to the graph of $f$ at the point $a$, see
 Figure \ref{fig:graphf}.
 If the tangent at the point $a$ is below the curve (resp. above), one
 has
 $f(.|a)=f(.)-f(a)-f'(a)(.-a)\geq0$ (resp. $\leq 0$).
\end{proof}

We first state a global \textit{a priori} estimate which ensures that
if \eqref{eq:H1} is satisfied at $t=0$ then it is preserved for any time by
the relaxed system \eqref{eq:model}.
\begin{Prop}
  \label{prop:aprioriestimate}
  Let $\delta\in ]0,\rho_1^*]$ and $\rho_\infty \in [\rho_2^*,3[$ (see
  Figure \ref{fig:graphf}).
  Then, for any $0<\rho<3$, the domain 
  \begin{equation*}
   \Omega_\rho :=\{ (\rho_1,\rho_2) \in (\delta,\rho)\times(\rho,\rho_\infty)\}
  \end{equation*}
  is invariant.
\end{Prop}

\begin{proof}
  One has to check the sign of $\langle Q,
  \nabla_{\mathbf{r}} S(\mathbf{r})\rangle$ on each side of the
  green rectangle domain, see Figure \ref{fig:Apriori}.
  On the sides $\{\rho_1=\rho\}$ and $\{\rho_2=\rho\}$ of the
  rectangle, $Q$ vanishes.
  Now for the side $\{\rho_1=\delta\}$, the outer normal is $\mathbf{n_\rho}=(-1,0)$ and 
  \begin{equation*}
    \langle Q,\nabla_{\mathbf{r}} S(\mathbf{r}) \rangle = \dfrac 1
    \varepsilon (\rho-\delta)(\rho-\rho_2)\rho f(\rho_2|\rho^-).
  \end{equation*}
  The product $ (\rho-\delta)(\rho-\rho_2)\rho$ is nonpositive and
  $$f(\rho_2|\rho^-)=f(\rho_2) -
  f(\rho^-)-f'(\rho^-)(\rho_2-\rho^-)\geq 0,$$ for
  $\rho_2>\rho_2^*$, see the green tangent line on Figure \ref{fig:graphf}.
  A similar argument involving $f(\rho_1|\rho_\infty)$ works for the
  side $\{ \rho_2=\rho_\infty\}$.
\end{proof}

We turn now to determine the invariant domains of hyperbolicity of the
relaxed model \eqref{eq:model} depending on the value of the density
$\rho$.
\begin{Prop}
  \label{prop:invdom}
  Fix $0<\delta<\rho_1^*$ and let $\rho_\infty$ be such that
  $f(\rho_\infty|\rho^-)=0$. The following subsets $\Omega_\rho$ are
  invariant domains of hyperbolicity :
  \begin{enumerate}
  \item \textbf{Pure gaseous stable zone}. For any $\delta<\rho<\rho_1^*$, 
    \begin{equation*}
      \Omega_\rho :=\{ (\rho_1,\rho_2) \in ]\delta,\rho^-[\times]\rho^+,\rho_\infty[\},
    \end{equation*}
    see Figure \ref{fig:PureBasin}.
  \item \textbf{Pure liquid stable zone}.
    For any $\rho_\infty >\rho>\rho_2^*$, 
    \begin{equation*}
      \Omega_\rho :=\{ (\rho_1,\rho_2) \in ]\rho^+,\rho[\times]\rho,\rho_\infty[\}.
    \end{equation*}
  \item \textbf{Metastable zones}.
    For any $\delta<\rho<\rho_\infty$,
    \begin{equation*}
      \Omega_\rho :=\{ (\rho_1,\rho_2) \in ]\delta,\min(\rho,\rho^-)[\times]max(\rho,\rho^+),\rho_\infty[\},
    \end{equation*}
    see Figure \ref{fig:MetaBasin} for the metastable vapor zone.
  \item\textbf{Spinodal zone}. For any $\rho^-<\rho<\rho^+$, 
    \begin{equation*}
      \Omega_\rho :=\{ (\rho_1,\rho_2) \in ]\delta,\rho^-[\times]\rho^+,\rho_\infty[\},
    \end{equation*}
    see Figure \ref{fig:SpinodalBasin}.
  \end{enumerate}
\end{Prop}

\begin{proof}
  We only give the proof for the spinodal zone. Following the proof of
  Proposition \ref{prop:aprioriestimate} one has to check that
  $\langle Q,\nabla_{\mathbf{r}}S(\mathbf{r})\rangle$ on the boundary of
  the green domain $\Omega_\rho $ of Figure \ref{fig:SpinodalBasin}.
  On the side $\{\rho_1=\rho^- \}$, the outer normal is
  $\mathbf{n}_\rho=(1,0)^t$. Thus one has
  \begin{equation*}
        \langle Q,\nabla_{\mathbf{r}} S(\mathbf{r}) \rangle =
        \dfrac{1}{\varepsilon}(\rho-\rho^-)(\rho-\rho_2)\rho(-f(\rho_2|\rho^-)).
  \end{equation*}
  The product $(\rho-\rho^-)(\rho-\rho_2)\rho$ is nonpositive since
  $\rho>0$, $\rho>\rho^-$ (because $\rho$ is fixed in
  $[\rho^-,\rho^+[$) and $\rho_2>\rho$ thanks to hypothesis \eqref{eq:H1}.
  Moreover $f(\rho_2|\rho^-)\leq 0$ for any $\rho_2\in
  ]\rho^-,\rho_\infty[$ according to Proposition \ref{prop:graphf}. Indeed
  the red tangent line of the blue curve of $f$ at $\rho^-$ is above
  the graph of $f$.
  Hence $ \langle Q,\nabla_{\mathbf{r}} S(\mathbf{r}) \rangle\leq
  0$ on the side $\{\rho_1=\rho^- \}$.
  The same kind of arguments are used to prove that 
  $ \langle Q,\nabla_{\mathbf{r}} S(\mathbf{r}) \rangle\leq 0$ 
  on the three other sides of $\Omega_\rho$.
\end{proof}
Note that this characterization of invariant domains of hyperbolicity
is formal since it relies on the smoothness of the densities
$\mathbf{r}$.  A possible way to generalize to weak solutions is to follow the definition of Hoff \cite{Hoff85}.

%------------------------------------------------------------
% Section 3 : Numerical approximation
%------------------------------------------------------------

\section{Numerical approximation}
\label{sec:numer-appr}
This section is devoted to numerical experiments. 
We do not wish to elaborate here on efficient numerical schemes for
this problem, but merely to illustrate some
typical behaviours of the model. Hence we limit ourselves to a simple
finite volume scheme, coupled to a time-splitting method for the
source terms. Considering the
stiffness of the problem, a complete numerical study is mandatory but
far beyond the aim of this paper. We emphasize that we did not implement any specific strategy for
the non hyperbolicity of the homogeneous system. However for all the cases we present, the computed sound velocity is
real. This does not prevent from possible losses of hyperbolicity, probably due to lack of convergence in the source term 
treatment, see in particular sections \ref{sec:bulle-double-choc} and \ref{sec:perturb-meta}.

\subsection{Definition of the splitting strategy}
\label{sec:defin-splitt-strat}
We rewrite the system~\eqref{eq:model} in a more compact form, considering the following Cauchy problem
\begin{equation}
  \label{eq:cauchy}
  \begin{aligned}
    & \p_t W + \p_xF(W) = S(W),\\
    & W(t=0,x) = W_0(x), \, \forall x \in \R,
  \end{aligned}
\end{equation}
where 
\begin{equation*}
  \begin{aligned}
    W &=(\rho, \rho_1,\rho_2, \rho u)^T,\\
    F(W) &=(\rho u, \rho_1 u,\rho_2 u, \rho u^2 + \alpha_1p(\rho_1) +
    \alpha_2 p(\rho_2))^T,\\
    S(W) &= \dfrac{1}{\varepsilon}
    (0,-(\rho-\rho_1)(\rho-\rho_2) f(\rho_2|\rho_1),
    (\rho-\rho_1)(\rho-\rho_2)f(\rho_1|\rho_2), 0),
  \end{aligned}
\end{equation*}
and $\varepsilon$ is the relaxation parameter.
Note that we exclude pure phase initial data, so that the equations on
the partial densities are not multiplied by $\alpha_i$ anymore.

Convective terms and source terms are taken into account by a
fractional step approach. We denote $\Delta t$ the time step and
$\Delta x$ the length of the cell $(x_{i-1/2}, x_{i+1/2})$ on the
regular mesh. Let $W^n$ be the Finite Volume approximation at time
$t^n = n\Delta t, \, n\in \mathbb N$. The approximated solution
$W^{n+1}$ of the Cauchy problem
\begin{equation}
  \label{eq:cauchy_disc}
  \begin{cases}
    \p_t W + \p_xF(W) = S(W), \, t\in (t^n, t^{n+1}),\, x \in\R,\\
    W(t^n,x) = W^n(x), \, \forall x \in \R,
  \end{cases}
\end{equation}
is approximated by splitting the problem in two steps.
The first one corresponds to the convective part
\begin{equation}
  \label{eq:convective}
  \begin{cases}
    \p_t W + \p_xF(W) = 0, \, t\in (t^n, t^{n+1}),\, x \in\R\\
     W(t^n,x) = W^n(x), \, \forall x \in \R,
  \end{cases}
\end{equation}
which provides $W^{n,-}$.
The second steps corresponds to the relaxation process
\begin{equation}
  \label{eq:relax}
  \begin{cases}
    \p_t W =S(W), \, t\in (t^n, t^{n+1}),\, x \in\R\\
     W(t^n,x) = W^{n,-}(x), \, \forall x \in \R,    
  \end{cases}
\end{equation}
which finally gives $W^{n+1}$.

%------------------------------------------------------------

\textbf{Numerical scheme for the convective part.} 
We consider a classical HLL numerical flux.
We adopt the following classical
notation
\begin{equation}
  \label{eq:W_i}
  W_i^n = \dfrac 1 {\Delta x} \int_{x_{i-1/2}}^{x_{i+1/2}} W(t^n, x)
  dx, \, n\geq 0, \, i\in \R.
\end{equation}
The scheme is the following
\begin{equation}
  \label{eq:rusanov}
  \Delta x (W_i^{n,-} - W_i^n) + \Delta t (F_{i+1/2}^n - F_{i-1/2}^n)=0,
\end{equation}
together with
\begin{equation*}
  F_{i+1/2}^n =
  \begin{cases}
    F(W_i^n), & if 0\leq s_L,\\
    \dfrac{s_RF(W_i^n) - s_LF(W_{i+1}^n) + s_L s_R (W_{i+1}^n
      -W_i^n)}{s_R-s_L},
    & if s_L\leq 0 \leq s_R,\\
    F(W_{i+1}^n), & if 0\geq s_R,
  \end{cases}
\end{equation*}
where $s_R = \max(u_i^n+c_i^n, u_{i+1}^n+c_{i+1}^n)$ and
$s_L=\min(u_i^n-c_i^n, u_{i+1}^n-c_{i+1}^n)$.
The time step is subjected to the classical CFL condition
\begin{equation}
  \label{eq:CFL}
  \dfrac{\Delta t}{\Delta x} |\lambda_{\max}|\leq 1,
\end{equation}
where $\lambda_{\max}$ is the maximal speed of wave computed on each
cell of the mesh.

%------------------------------------------------------------
\textbf{Numerical treatment of the source terms.}
The initial condition for this step is $W^{n,-}$ which is assumed to be
admissible, that is $(\rho^{n,-},\rho_1^{n,-},\rho_2^{n,-})\in
\mathcal R$. 
The total density $\rho$ and the momentum $\rho u$ remain unchanged
during this step.
Only the densities $\rho_1$ and $\rho_2$ may vary which leads to the
following system
\begin{equation}
  \label{eq:source}
  \begin{aligned}
    \p_t \rho &= \p_t (\rho u) =0,\\
    \p_t \rho_1 &=
    -\dfrac 1 \varepsilon
    (\rho-\rho_1)(\rho-\rho_2)f(\rho_2|\rho_1),\\
    \p_t \rho_2&=
    \dfrac 1 \varepsilon
    (\rho-\rho_1)(\rho-\rho_2)f(\rho_1|\rho_2).
  \end{aligned}
\end{equation}
At this stage we merely use a classic explicit order 4 Runge-Kutta
method to integrate the source term. 

Such a treatment enforces tough constraints on the time step: the computations were performed
with 1000 iterations using a time step of $10^{-6}$. We emphasize that this does not ensure
the actual convergence to the equilibrium state. This pleads for a more efficient method, for instance
 a semi-implicit scheme in the spirit of \cite{IMEX09}.

%--------------------------------------------------------
% Numerical results
%--------------------------------------------------------
\subsection{Numerical results}
\label{sec:numerical-results}

We present here numerical results that assess the ability of the
model~\eqref{eq:model} to capture phase transition and metastable
states.

We consider the van der Waals pressure
in its reduced form~\eqref{eq:reduced_vdw} at a constant subcritical
temperature $T=0.85$.
At this fixed temperature the extrema $\rho^-$ and $\rho^+$ of the
pressure and the values $\rho_1^*$ and $\rho_2^*$ defined by the Maxwell
construction on the chemical potential are given in
Table~\ref{tab:carac}.

We propose test cases with Riemann initial conditions that is
\begin{equation}
  \label{eq:CI_RP}
  W(t=0,x)=w_0(x)=
  \begin{cases}
  W_L, & \text{if } x\leq 0,\\
  W_R, & \text{if } x >0.
  \end{cases}
\end{equation}

The following test cases are set on the domain $[0,1]$,
with an  uniform mesh of 10000 cells and Neumann boundary
conditions at $x=0$ and $x=1$.

%%_____________________________________________________________________________________________ 
\subsubsection{Riemann problem with phase transition}
\label{sec:riemann-problem-with}
The initial state $W_L$ and $W_R$ are 
\begin{equation*}
  \begin{aligned}
    \rho_L=\rho_{1,L}=0.3, \; \rho_{2,L}=\rho_2^*, \; u_L=0,\\
    \rho_R=\rho_{2,R}=1.9, \; \rho_{1,R} = \rho_1^*, \; u_R=0.
  \end{aligned}
\end{equation*}
The left state is a pure stable gas and the right state is a pure
stable liquid. Various value of $\varepsilon$ are considered. 
The solution is at time $t=0.1$s. 
One can observe the appearance of a mixture zone 
on both sides of the interface, see in particular the pressure profile
Figure~\ref{fig:transition1}-bottom left. The results with $\varepsilon=10^{-4}$
or $10^{-6}$ are similar expect on the velocity profile 
(see Figure~\ref{fig:transition1}-bottom right) where the intermediate state on the left of the
interface is slightly modified for $\varepsilon=10^{-4}$.

\begin{figure}
  \centering
  \includegraphics[width=6cm]{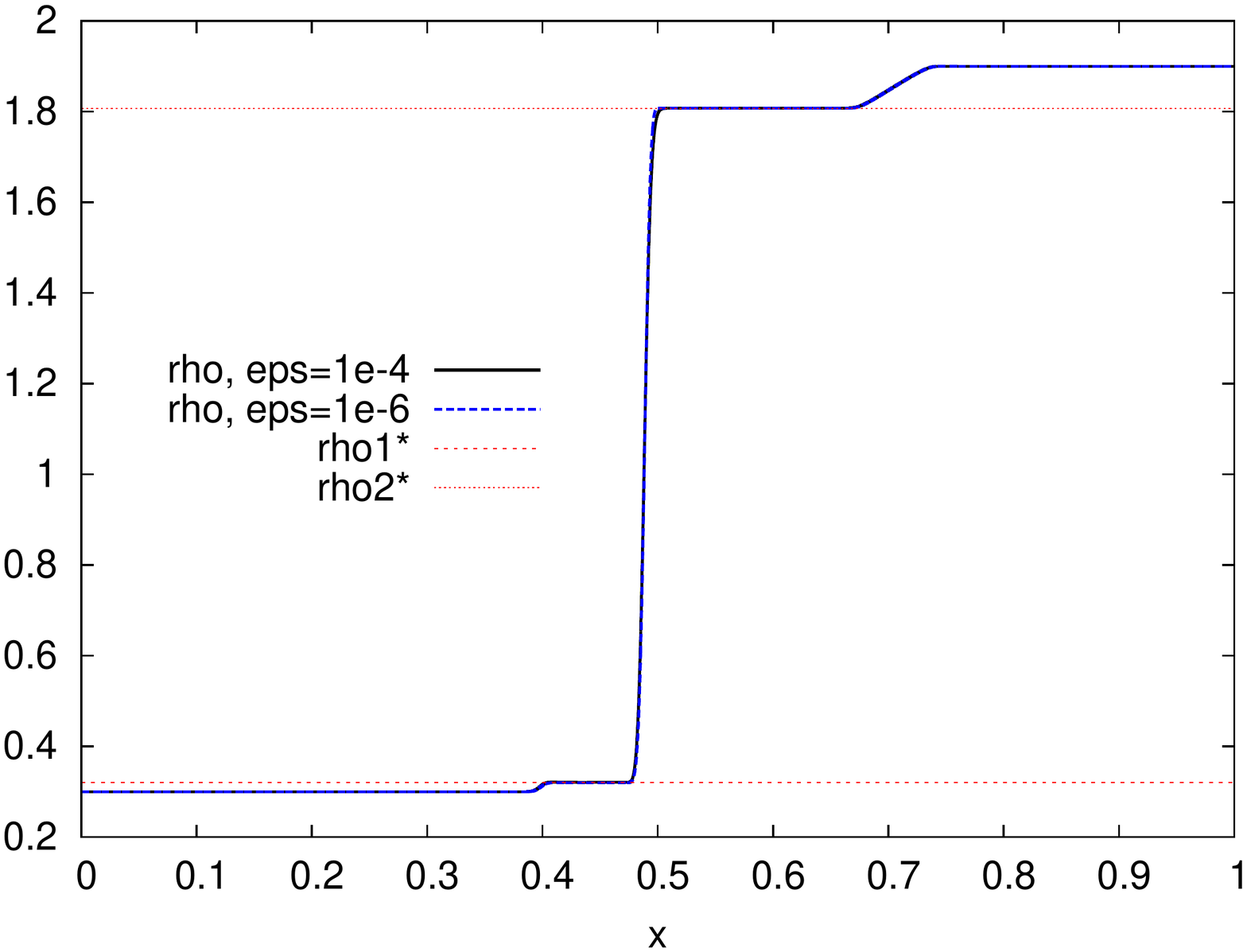}
  \includegraphics[width=6cm]{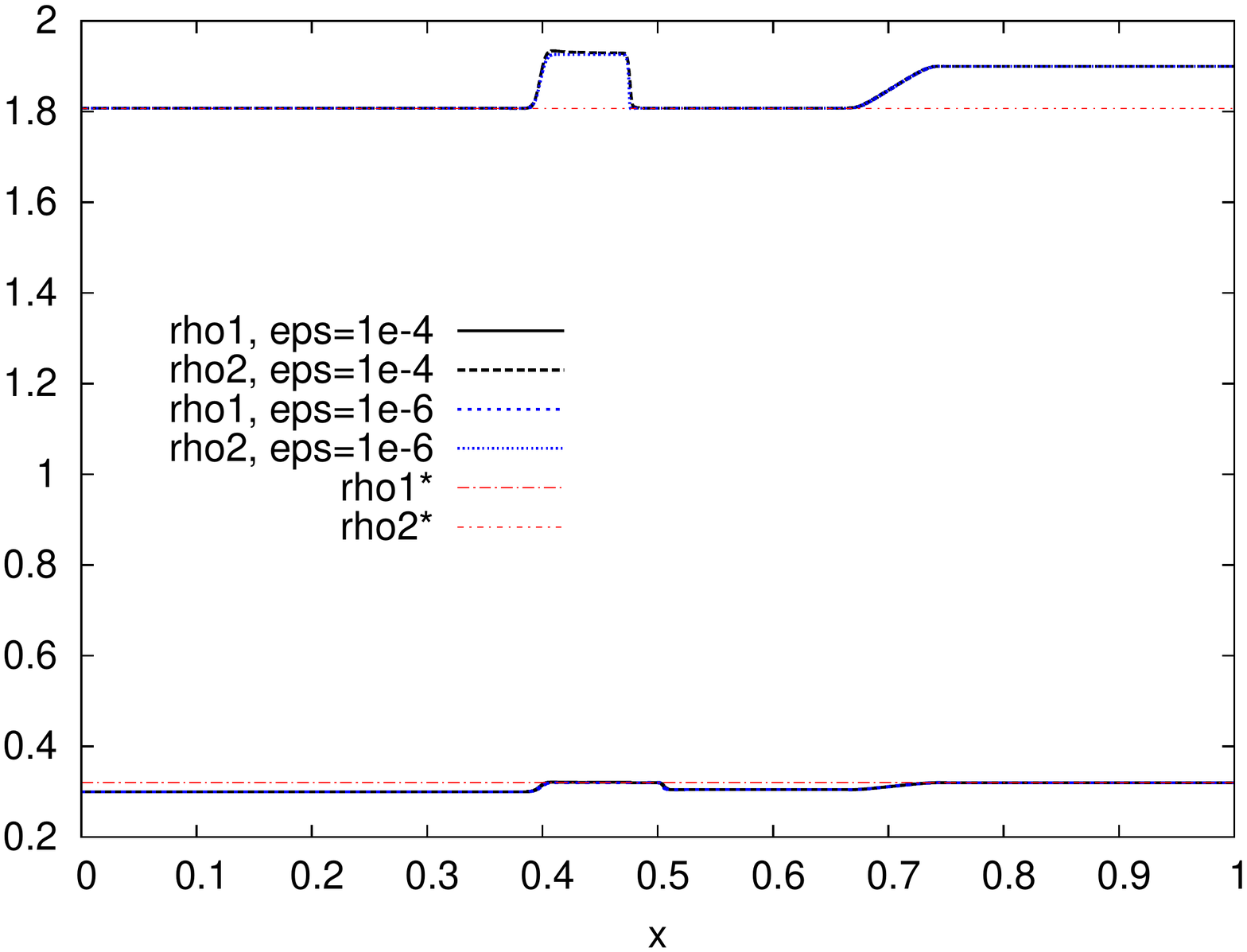}
  \includegraphics[width=6cm]{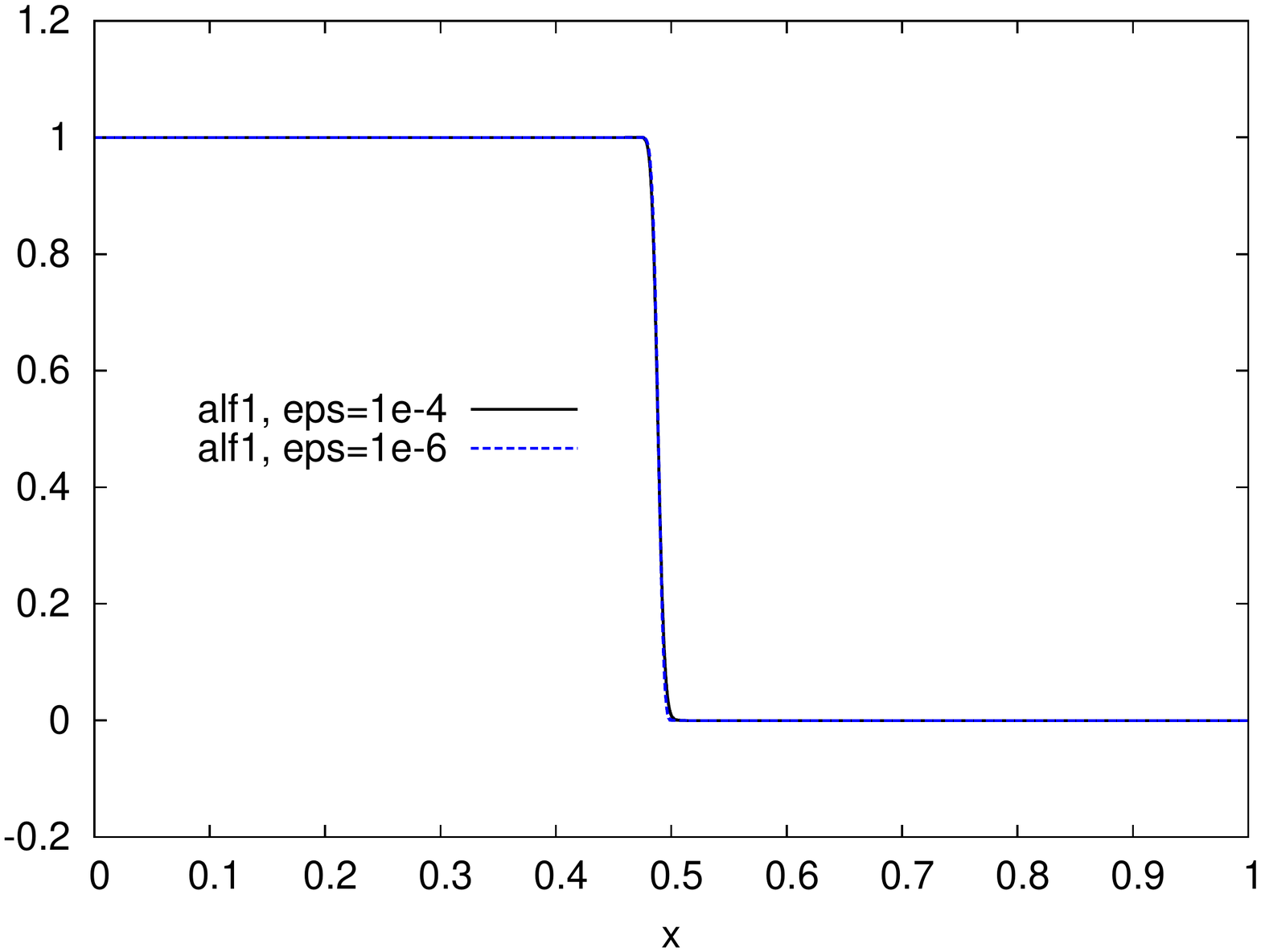}
  \includegraphics[width=6cm]{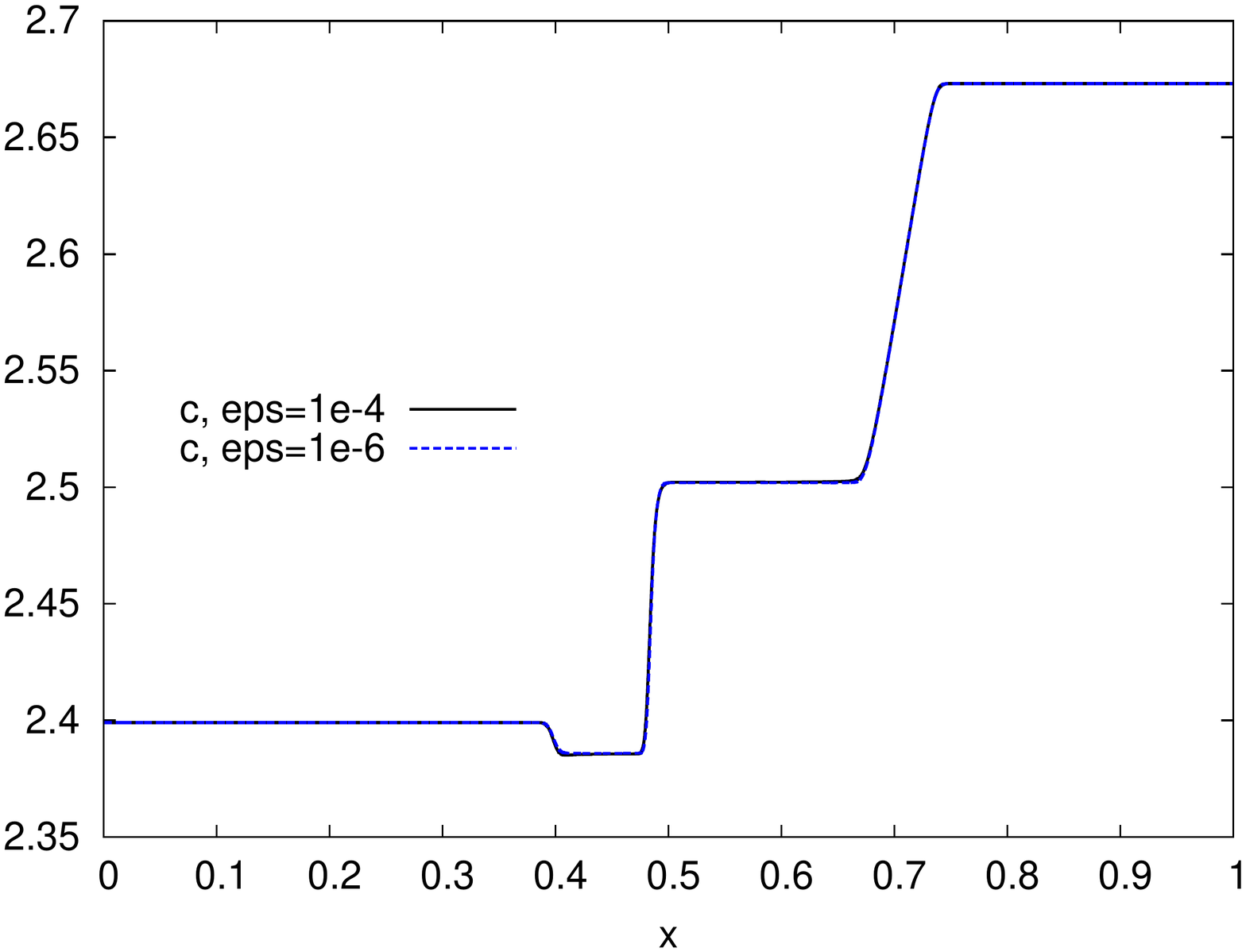}
  \includegraphics[width=6cm]{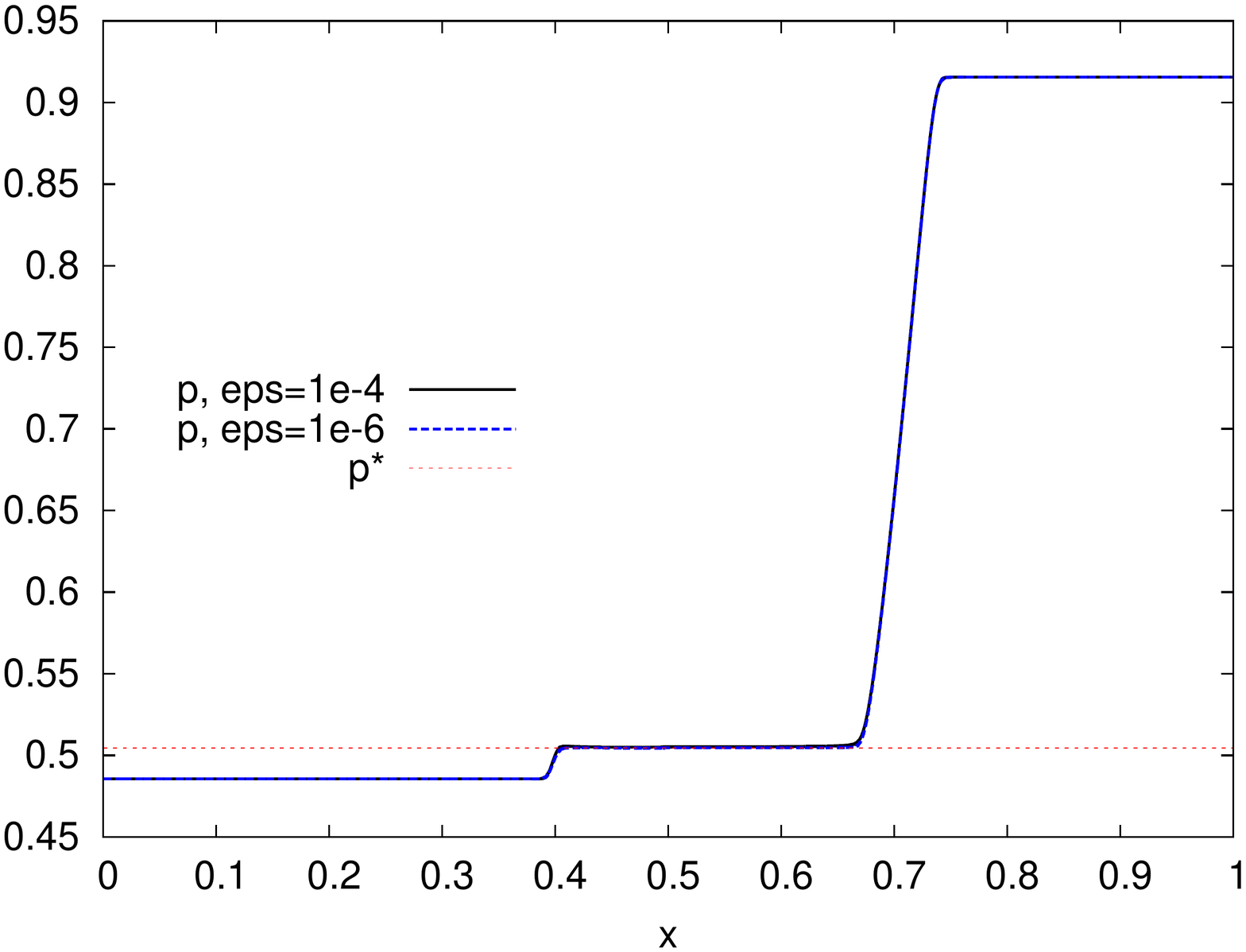}
  \includegraphics[width=6cm]{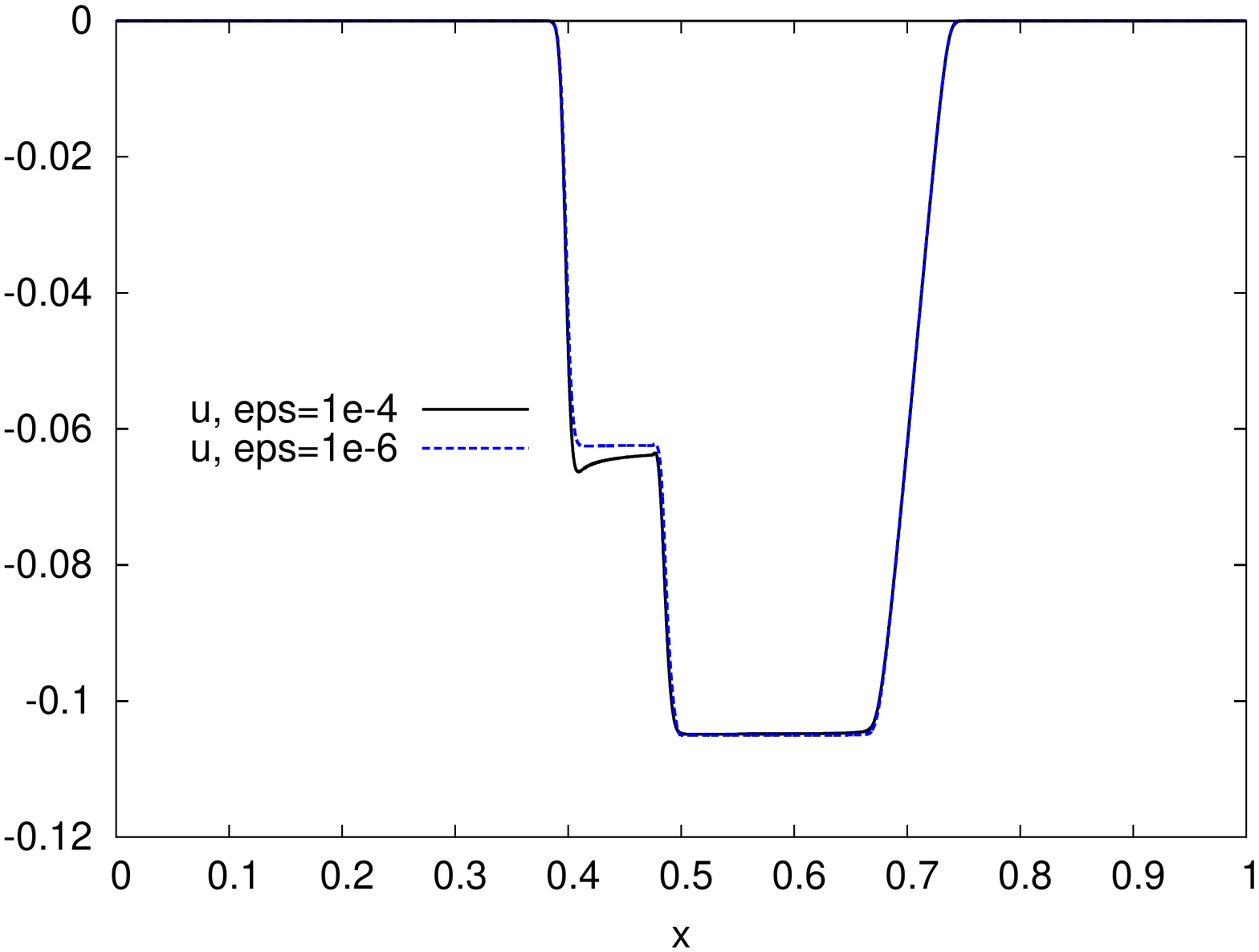}
  \caption{Riemann problem with
    phase transition. From top left to bottom right: density $\rho$,
    densities $rho_1$ and $\rho_2$, volume fraction $\alpha_1$, speed
    of sound $c$, pressure $p$ and
    velocity $u$.}
  \label{fig:transition1}
\end{figure}

%%_______________________________________________________________________ cavitation
\subsubsection{Cavitation by double rarefaction}
\label{sec:avit-double-raref}

The test consists in a liquid state submitted to a double
rarefaction wave. 
The initial state is given by
$\rho=\rho_2^*-10^{-3}$, $\rho_2=\rho_2^*$,
$\rho_1=\rho_1^*$ and the velocities are $u_R=0.3=-u_L$.
The total density corresponds to a metastable liquid,
and the initial volume fraction is $\alpha_1\simeq0.000672$, which means
that phase 2 is predominant.
The solution is computed at time $t=0.1$s. 
We observe on the plot of the volume fraction
(Figure~\ref{fig:cavit1}-second line left) that a
bubble of stable vapor appears around the interface $x=0.5$.
The value of $\varepsilon$ does not modify the profile of the bubble.
However the pressure profile is sharper for $\varepsilon=10^{-6}$ than for $10^{-4}$, see
Figure~\ref{fig:cavit1}-bottom left.

\begin{figure}
  \centering
  \includegraphics[width=6cm]{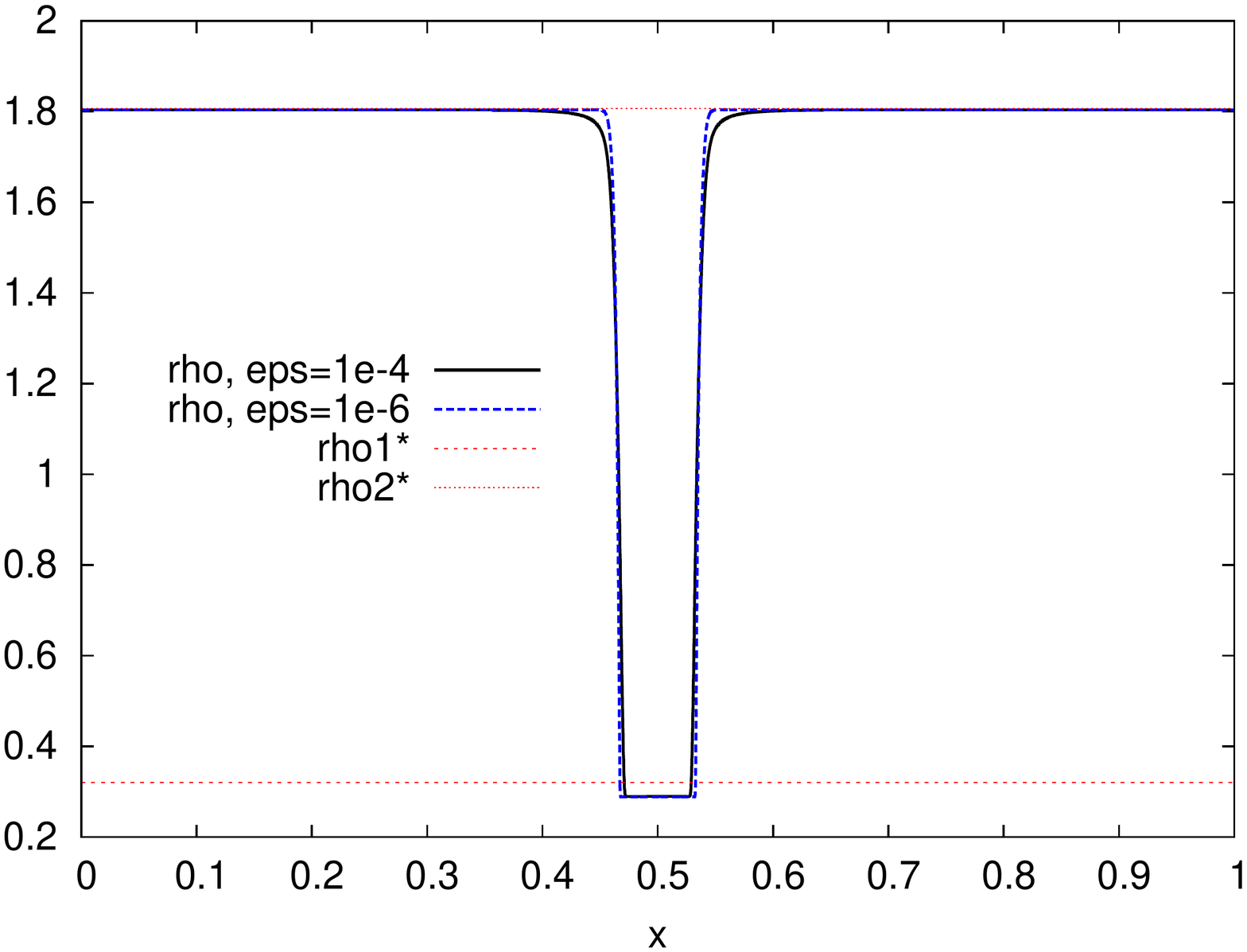}
  \includegraphics[width=6cm]{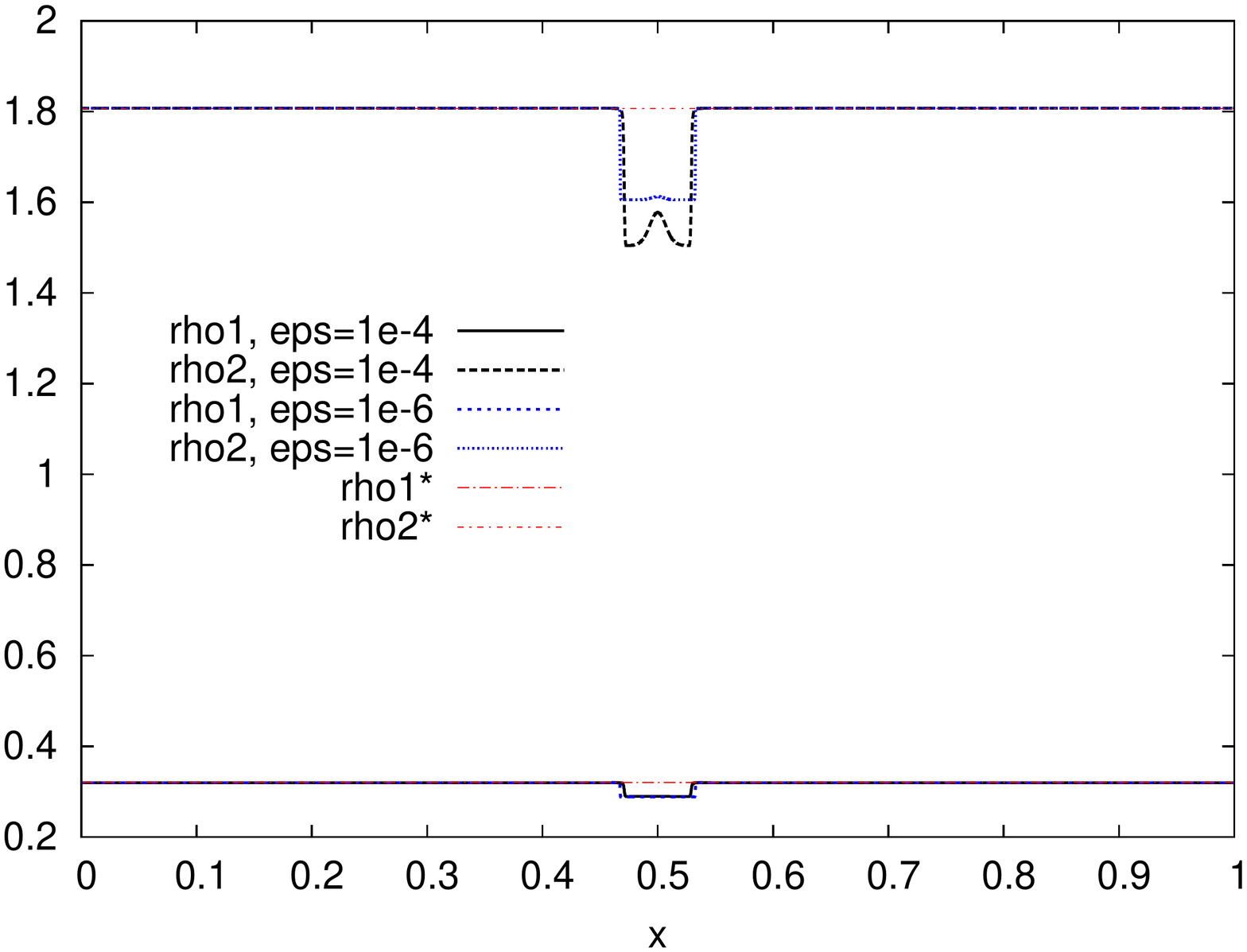}
  \includegraphics[width=6cm]{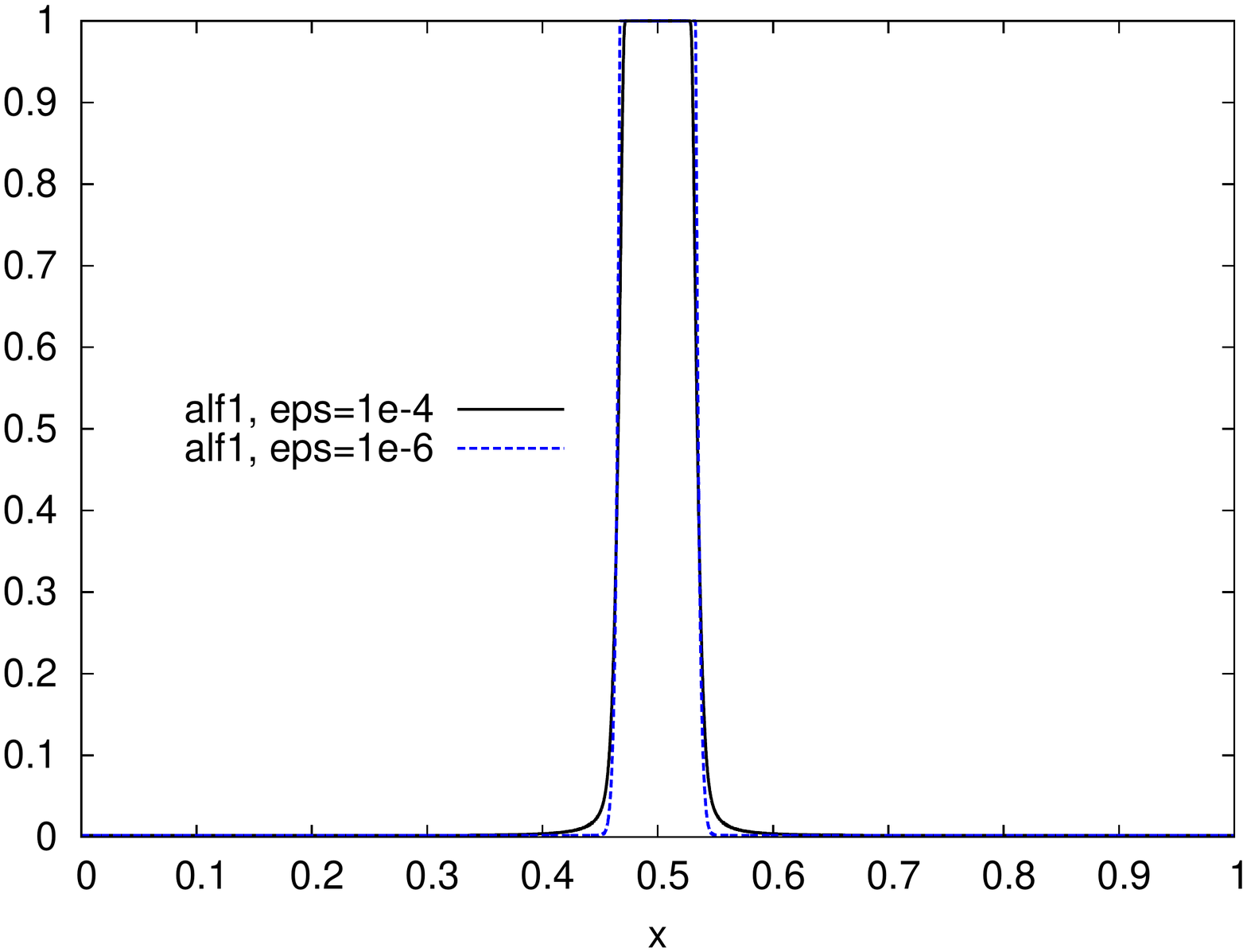}
  \includegraphics[width=6cm]{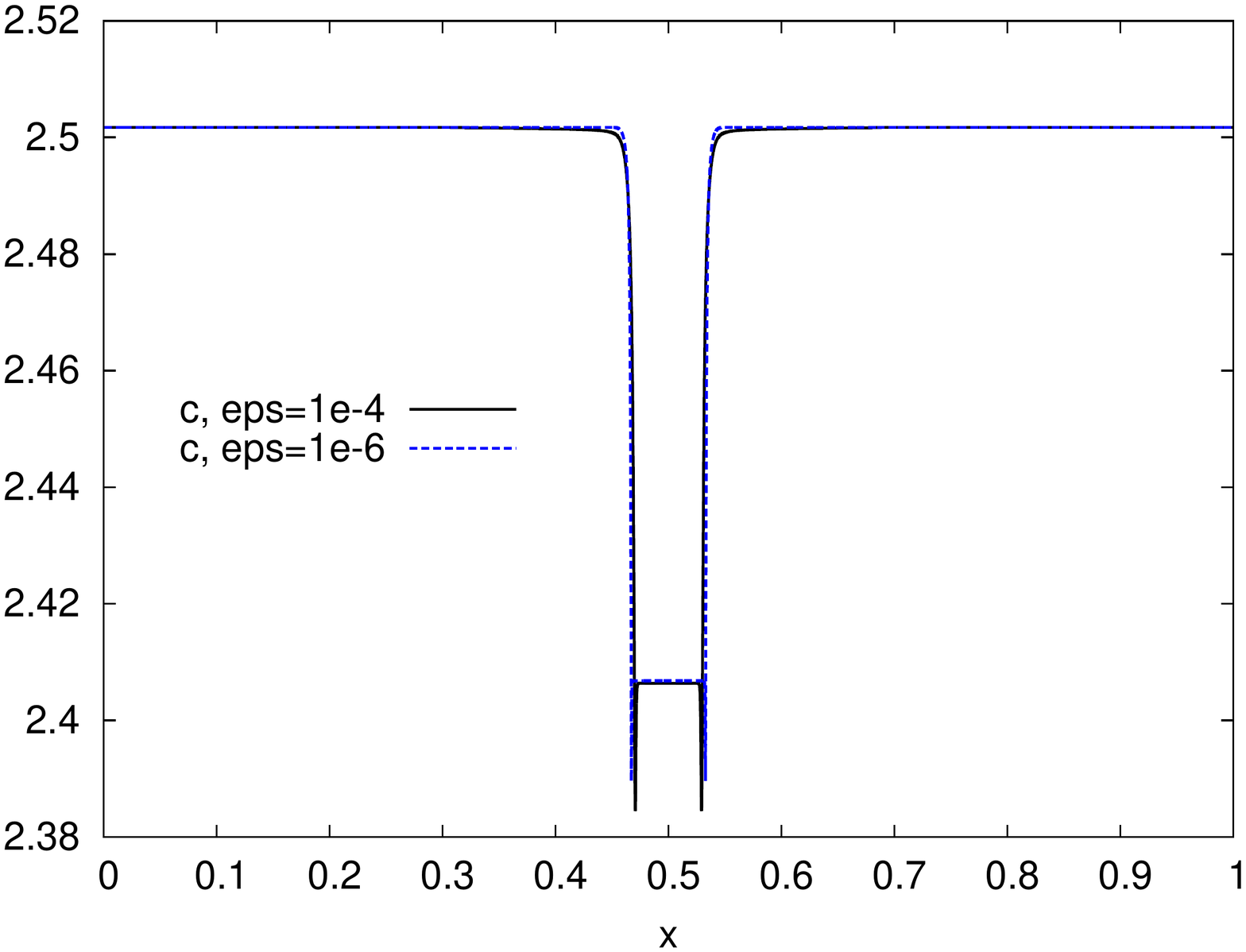}
  \includegraphics[width=6cm]{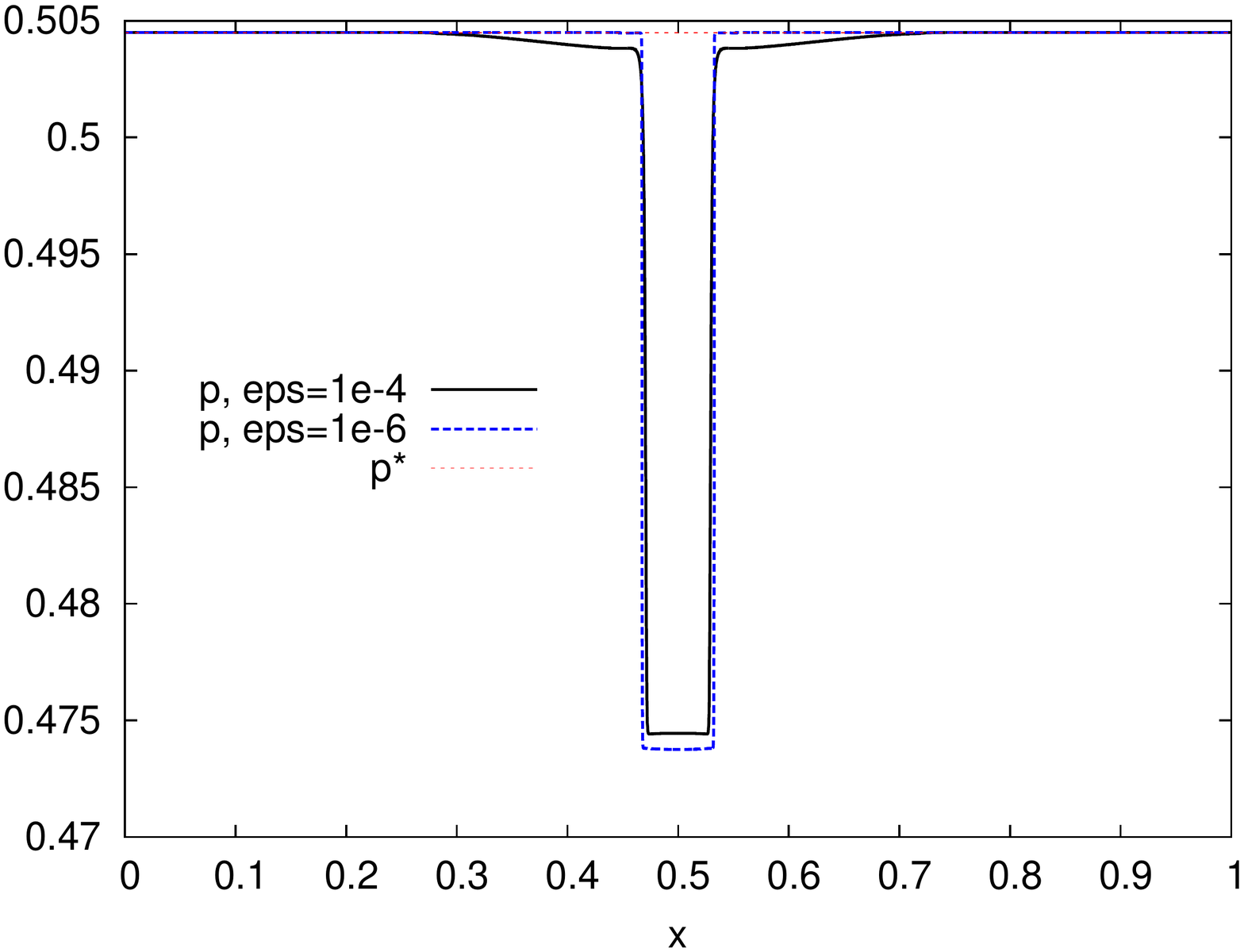}
  \includegraphics[width=6cm]{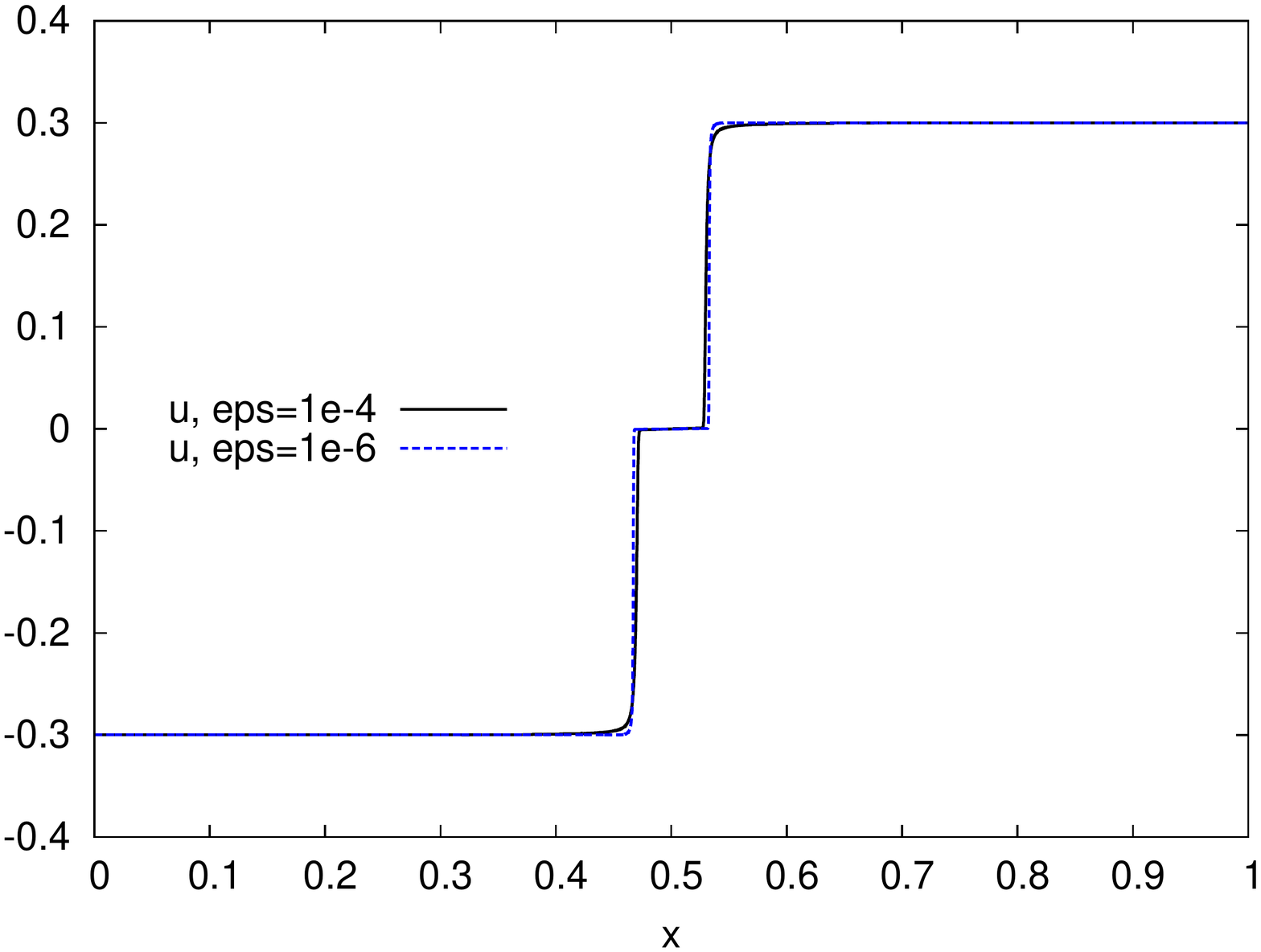}
  \caption{Cavitation by double rarefaction. From top left to bottom right: density $\rho$,
    densities $rho_1$ and $\rho_2$, volume fraction $\alpha_1$, speed
    of sound $c$, pressure $p$ and
    velocity $u$.}
  \label{fig:cavit1}
\end{figure}

%%_______________________________________________________________________ nucleation
\subsubsection{Nucleation by double shock}
\label{sec:bulle-double-choc}

The test consists in a pure stable gaseous state submitted to a double
shock wave. The initial state is given by
$\rho=\rho_1=0.3$,
$\rho_2=1$ and the velocities are $u_R=-0.3=-u_L$.
The solution is computed at time $t=0.4$s.
Note that phase 2 is not present initially but is fixed in the spinodal zone.
We observe two different behaviours of the solution depending on the value of 
$\varepsilon$.
For $\varepsilon=10^{-6}$, the profile of the volume fraction (see
Figure~\ref{fig:nucl1}-second line left)
shows that a droplet of liquid appears around $x=0.5$ which is not 
the case for $\varepsilon=10^{-4}$ where there is no droplet.
For $\varepsilon=10^{-6}$, one observes on the density plot (see
Figure~\ref{fig:nucl1}-top left) that
the liquid state inside the droplet admits a density  $\rho$ close to $\rho_2^*$ with small oscillations.
The droplet is surrounded by 
two mixture areas with $\rho=\rho_1^*$.
The pressure curve inside the droplet presents oscillations  (see
Figure~\ref{fig:nucl1}-bottom left)
which might be due to a loss of hyperbolicity due to the lack of
accuracy in the approximation of the source \eqref{eq:SJ}.

\begin{figure}
  \centering
  \includegraphics[width=6cm]{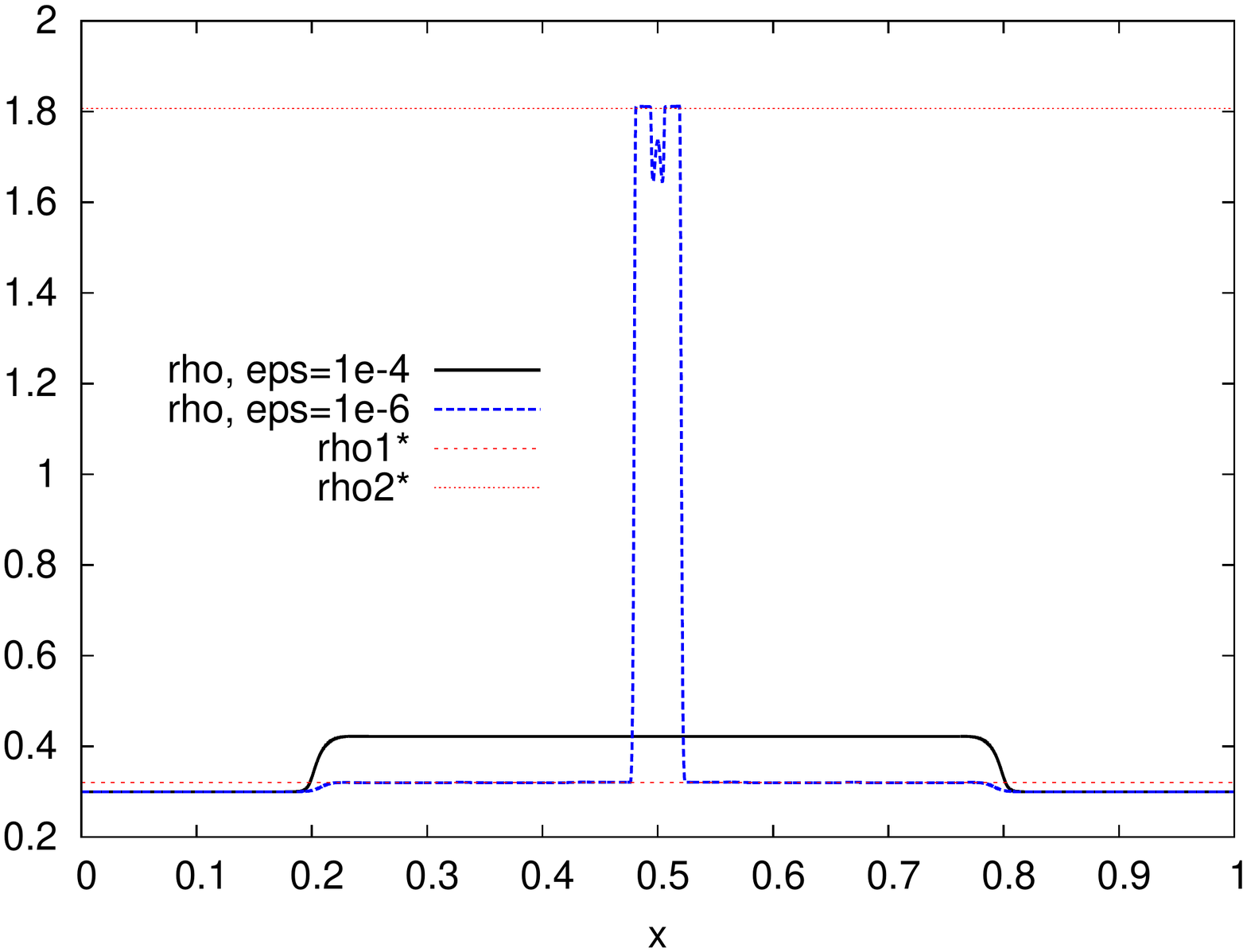}
  \includegraphics[width=6cm]{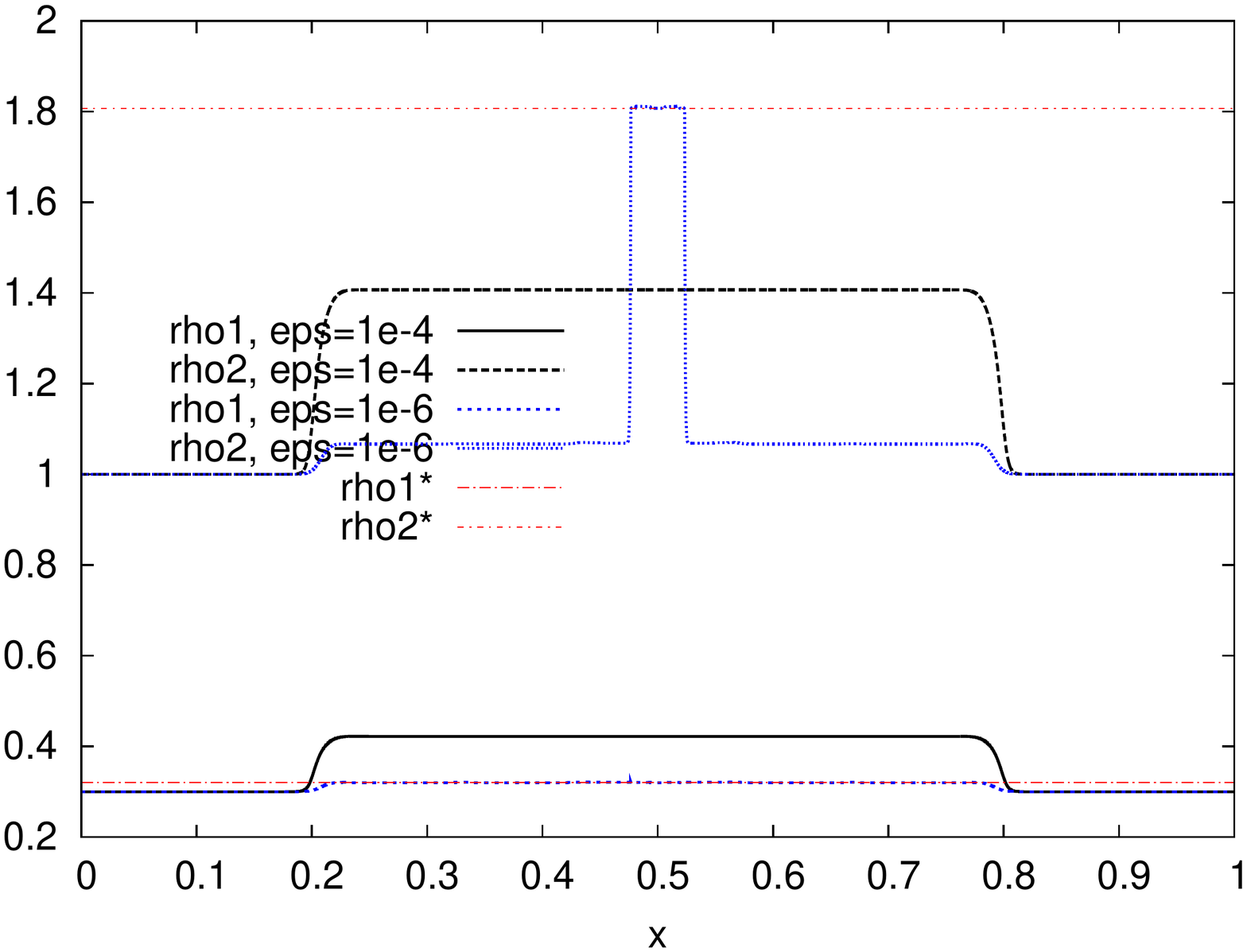}
  \includegraphics[width=6cm]{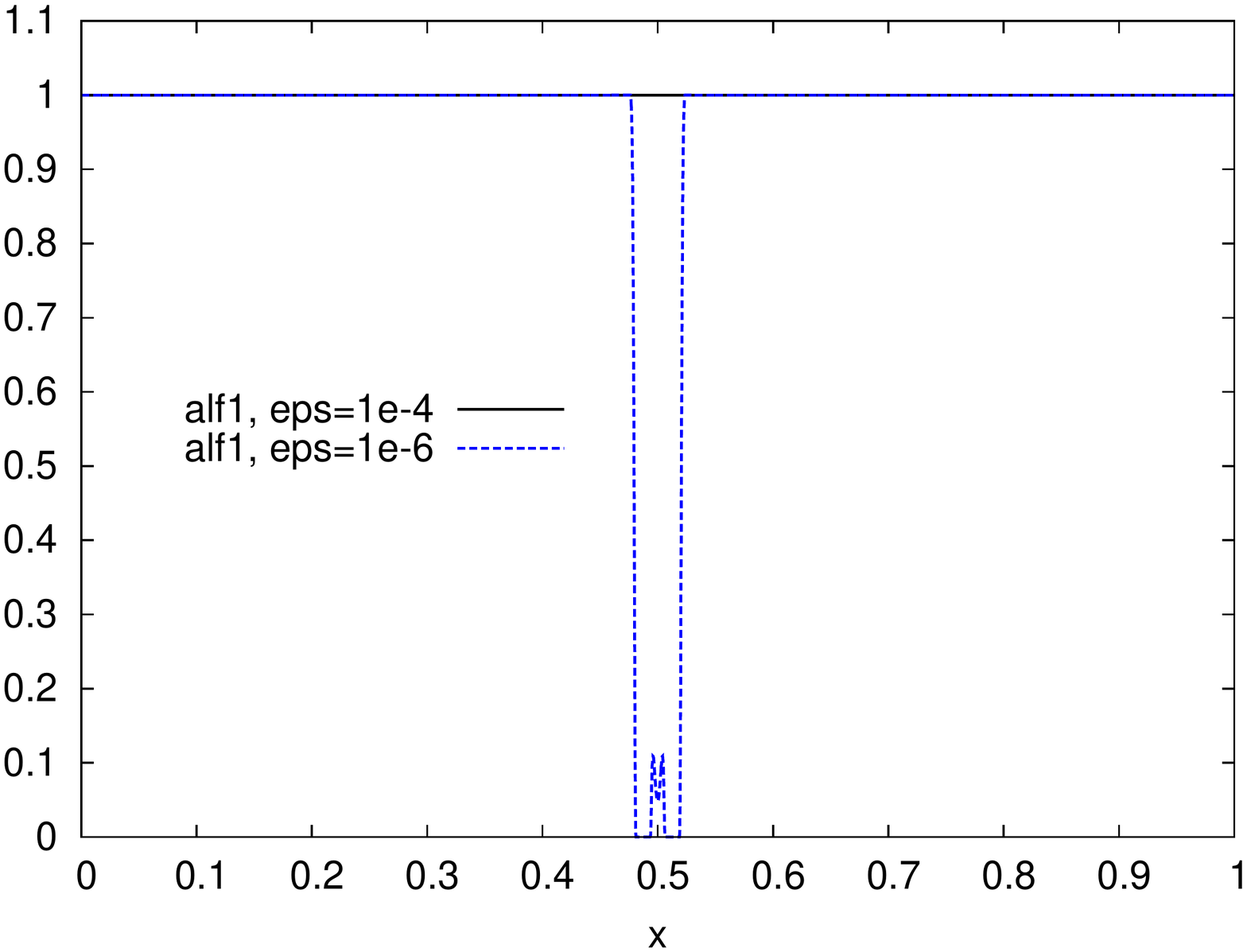}
  \includegraphics[width=6cm]{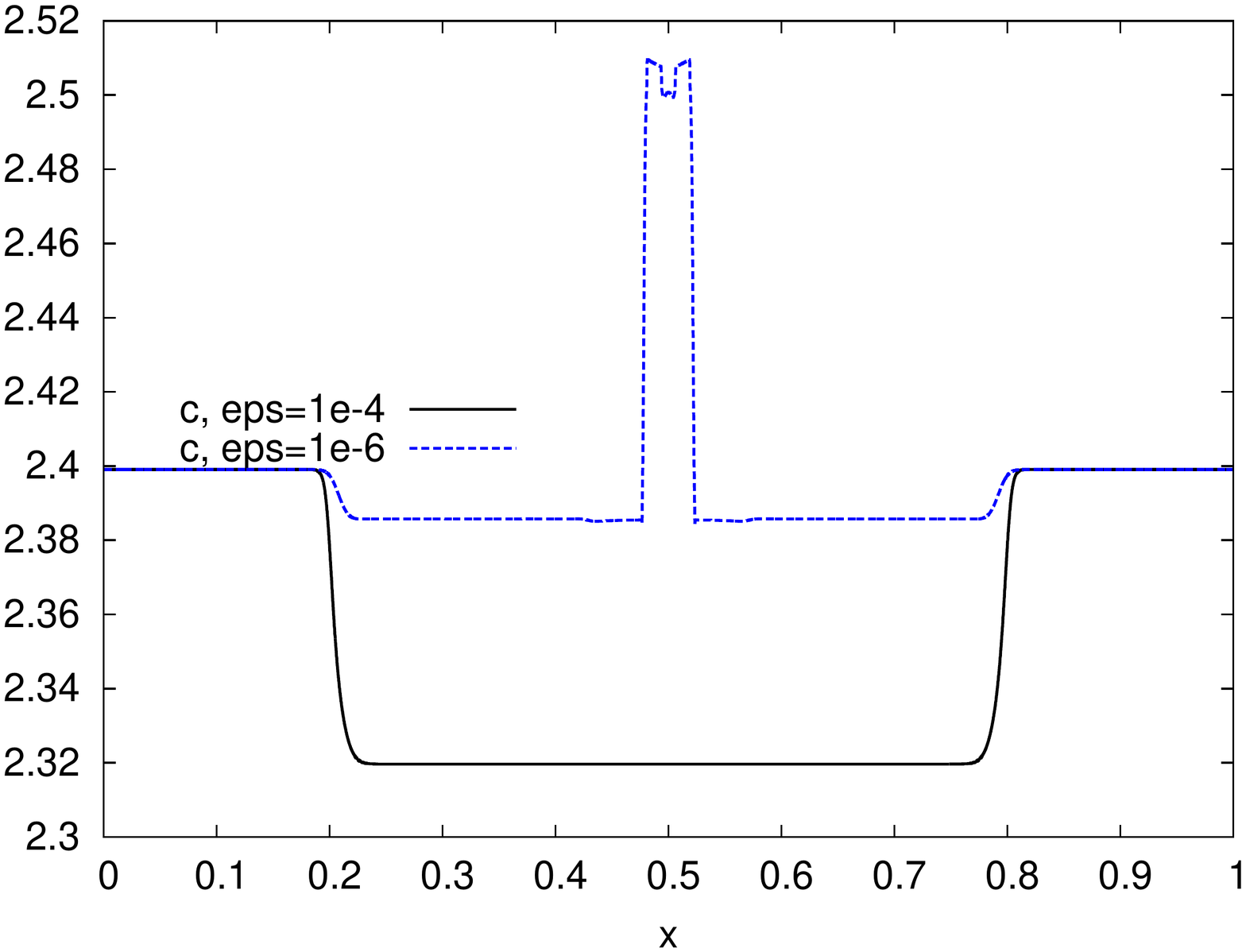}
  \includegraphics[width=6cm]{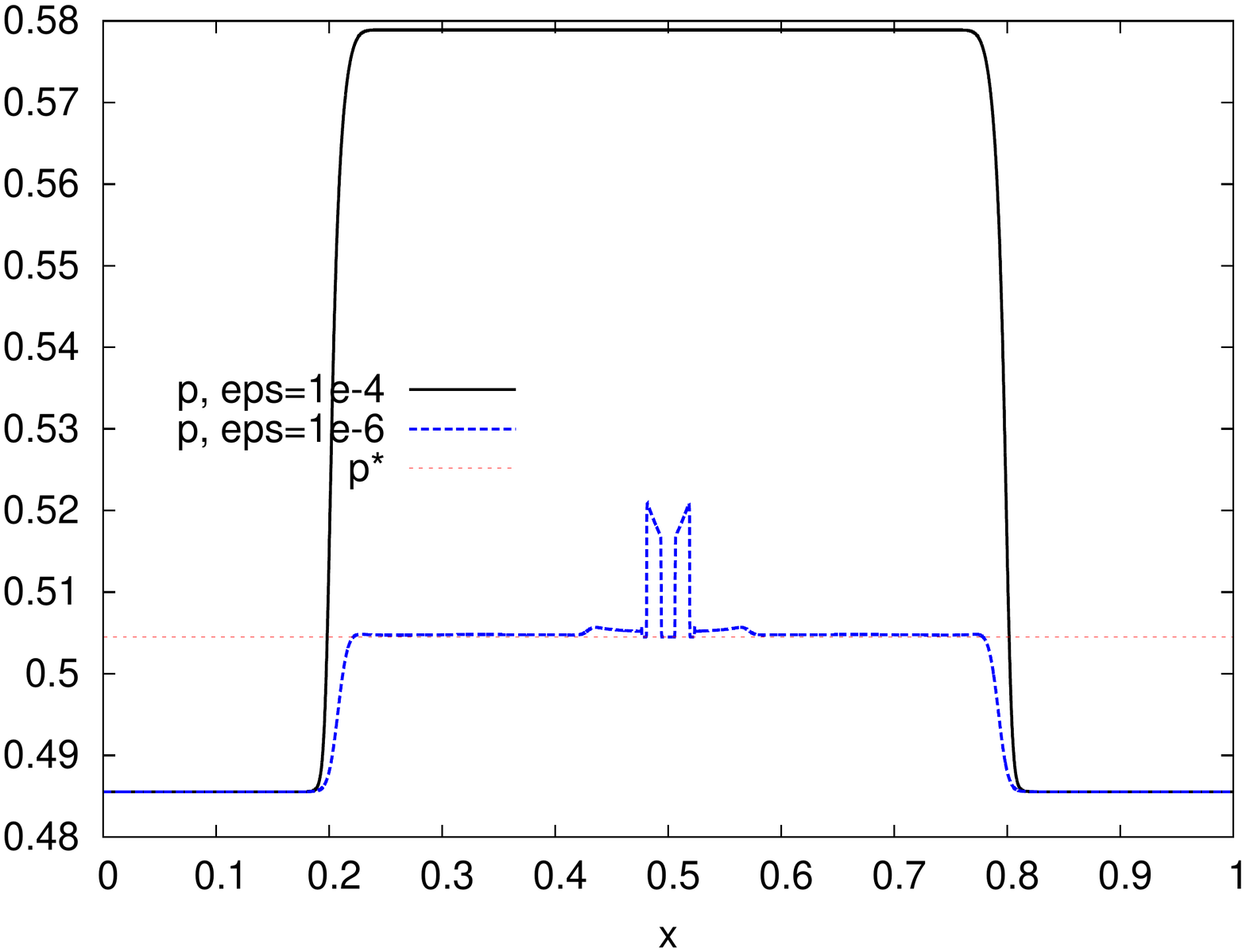}
  \includegraphics[width=6cm]{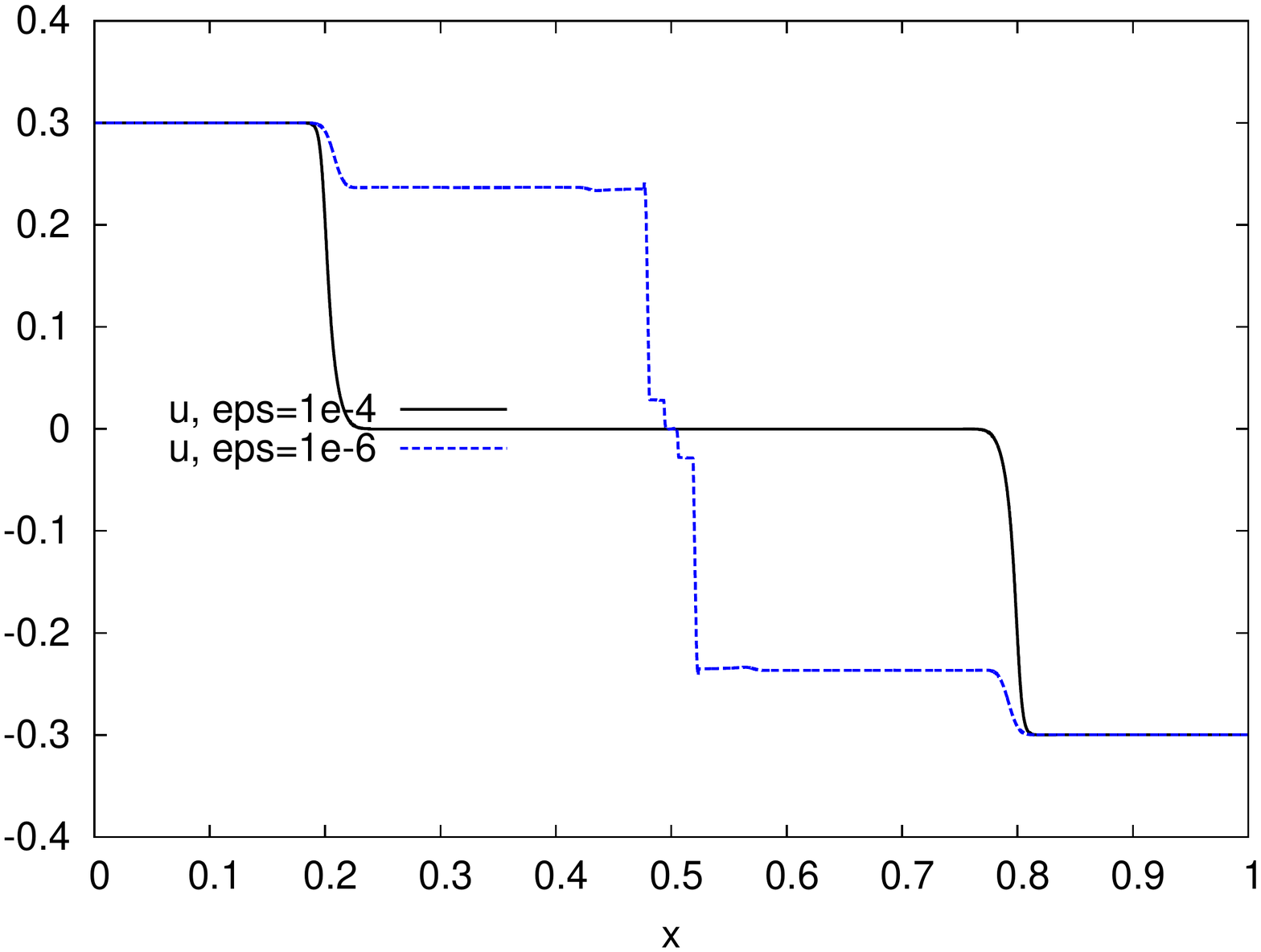}
  \caption{Nucleation by double shock. From top left to bottom right: density $\rho$,
    densities $rho_1$ and $\rho_2$, volume fraction $\alpha_1$, speed
    of sound $c$, pressure $p$ and
    velocity $u$.}
  \label{fig:nucl1}
\end{figure}

%%_______________________________________________________________________ perturbation metastable
\subsubsection{Acoustic perturbation of a metastable state}
\label{sec:perturb-meta}

This example consists in a constant metastable vapor state, with a perturbation in the velocity.
The initial state is 
$\rho=0.42$,
$\rho_1=0.32$, $\rho_2=0.52$ and the velocities are $u_L=0.4$, $u_R=0$.
Both densities $\rho_1$ and $\rho_2$ are in the metastable state, 
and we impose a compression from the left with velocity $0.4$. 

The compression induces the appearance of droplet of pure liquid which moves from 
the left to the right, see the time evolution on Figure \ref{fig:meta1}. 
With a smaller velocity perturbation the structure of the waves is similar, 
but there is no creation of a droplet at the interface. 
One can check on Figure~\ref{fig:meta1}- top left that the density $\rho$ inside the droplet is larger than
$\rho_2^*$. The droplet is surrounded by two areas with a mixture state with $\rho=\rho_1^*$.
The velocity and pressure profiles exhibit spikes on both sides of the droplet, which amplitude decreases when $\varepsilon$ decreases,
see Figures  \ref{fig:meta1}, \ref{fig:meta2} and \ref{fig:meta4}.
Notice also that the pressure in the mixture zone is not at the value $p^*$. It seems that when $\varepsilon$ decreases
the value is closer, this may indicate that the source term has not reached the equilibrium state yet.

\begin{figure}
\begin{center}
  \includegraphics[width=6cm]{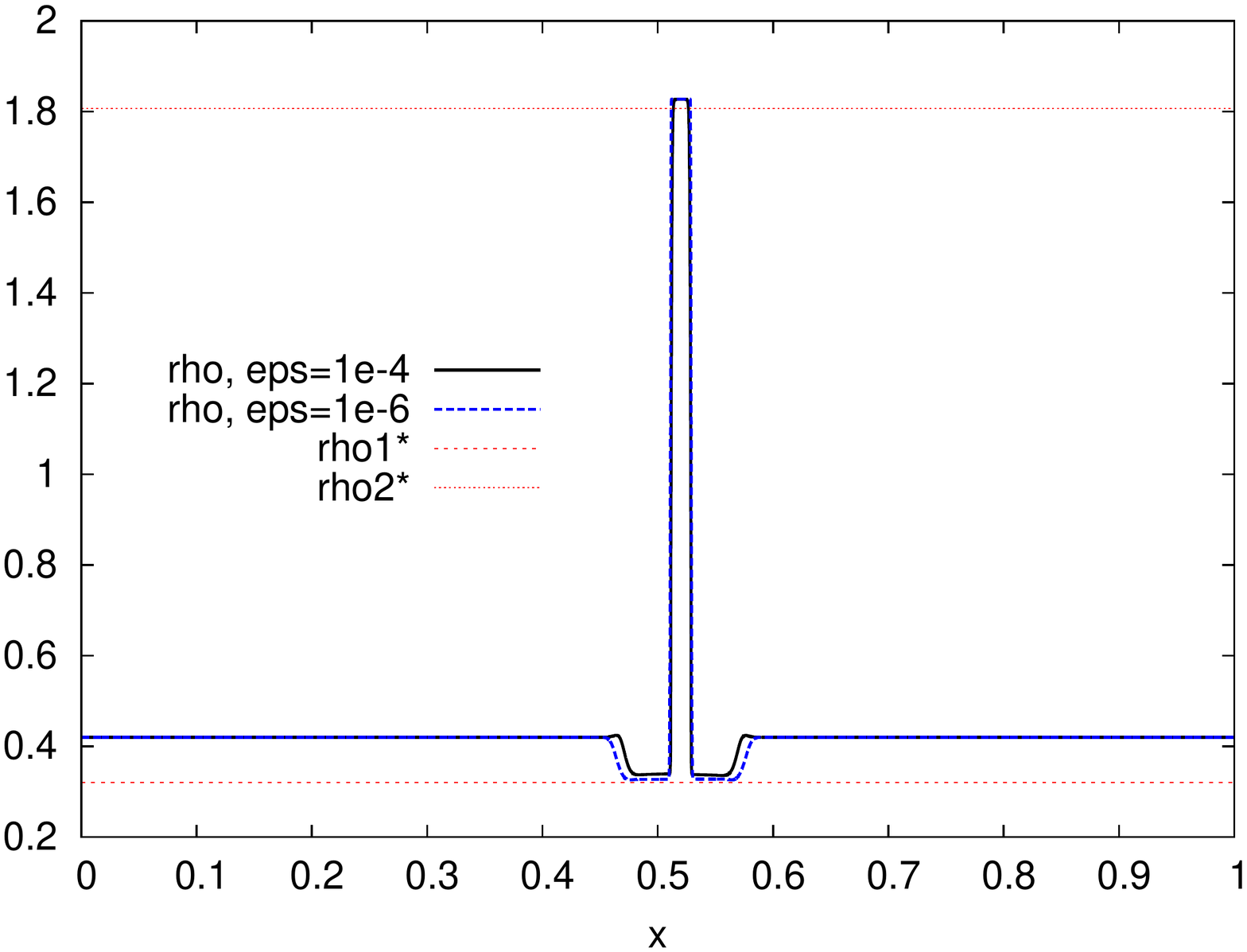}
  \includegraphics[width=6cm]{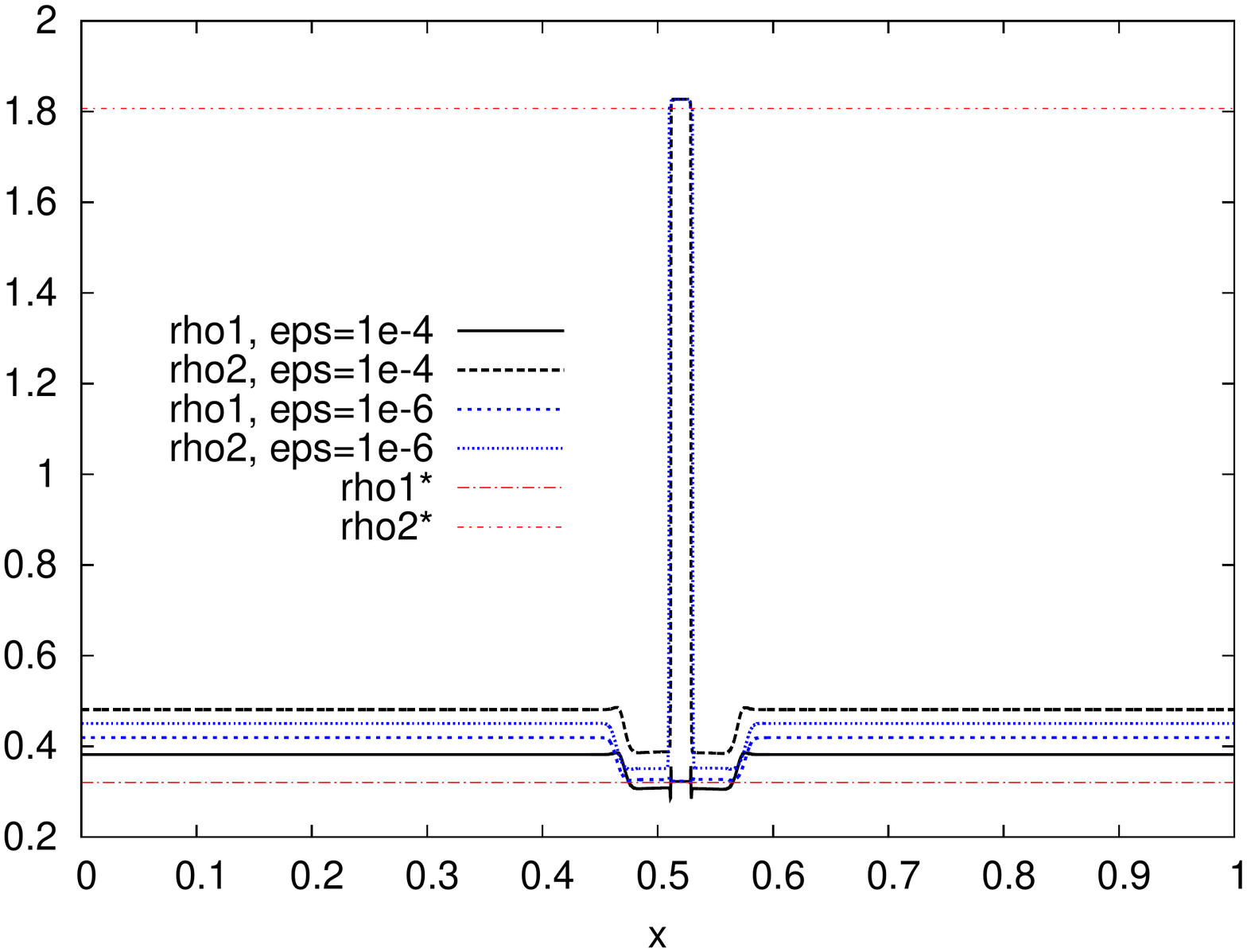}
  \includegraphics[width=6cm]{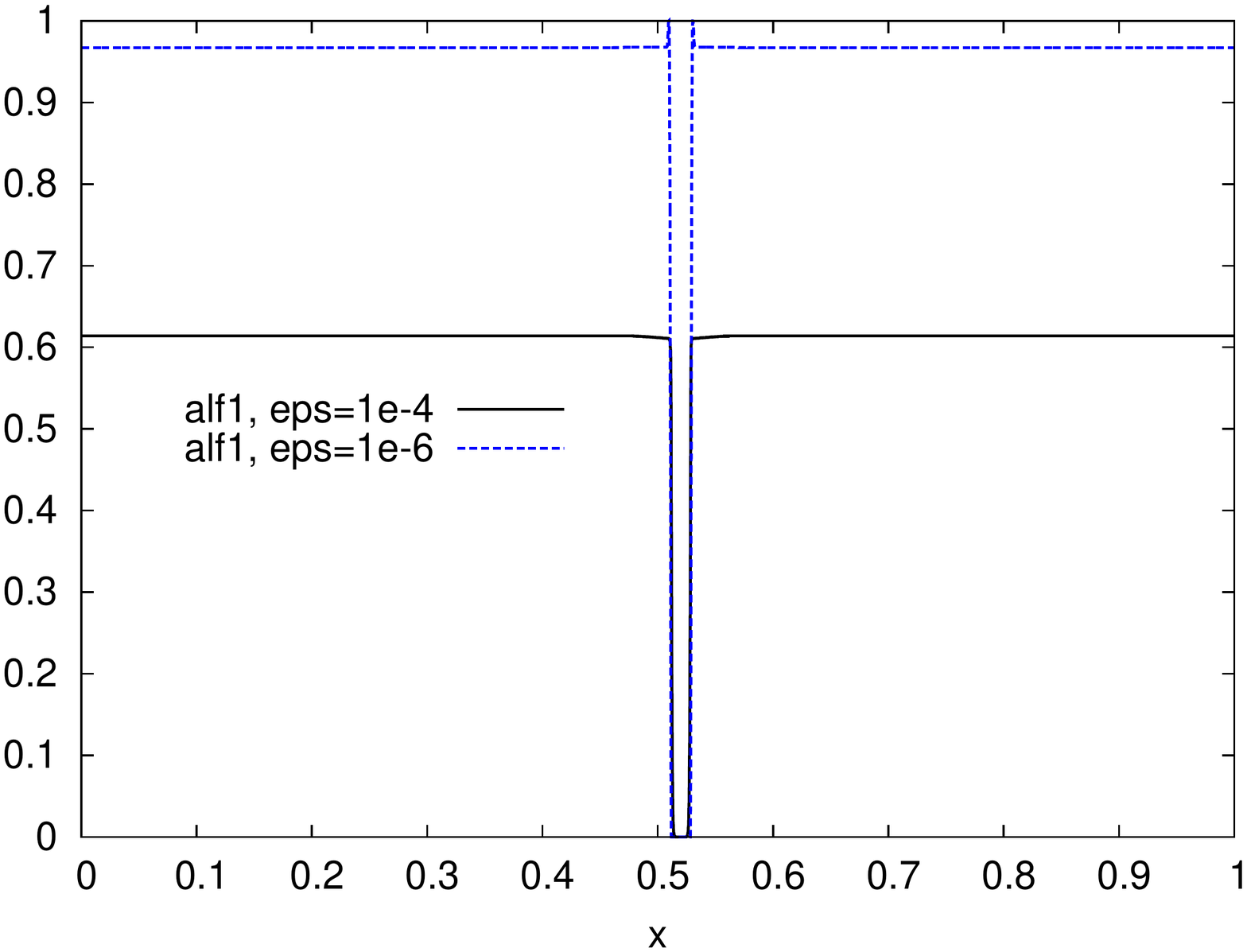}
  \includegraphics[width=6cm]{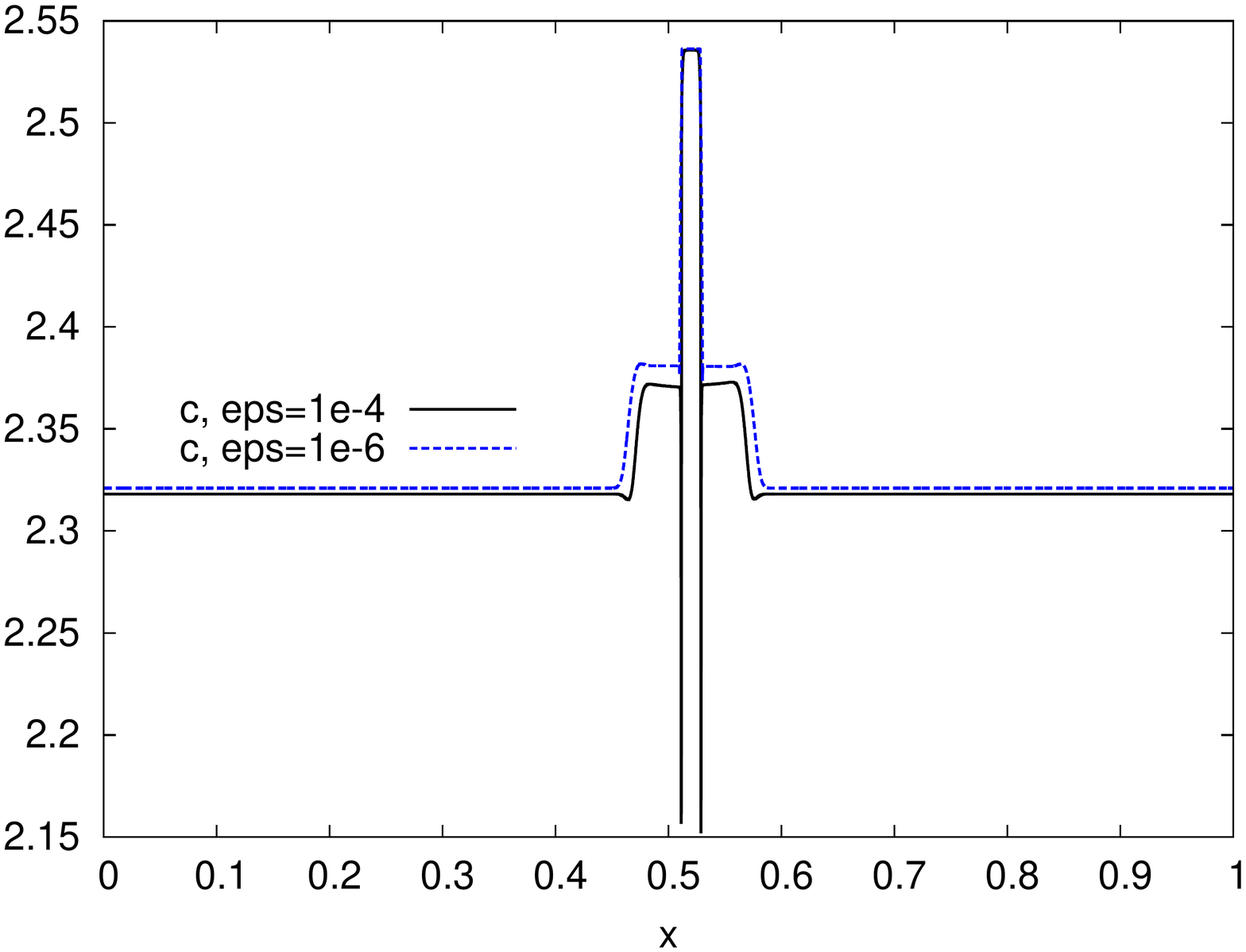}
  \includegraphics[width=6cm]{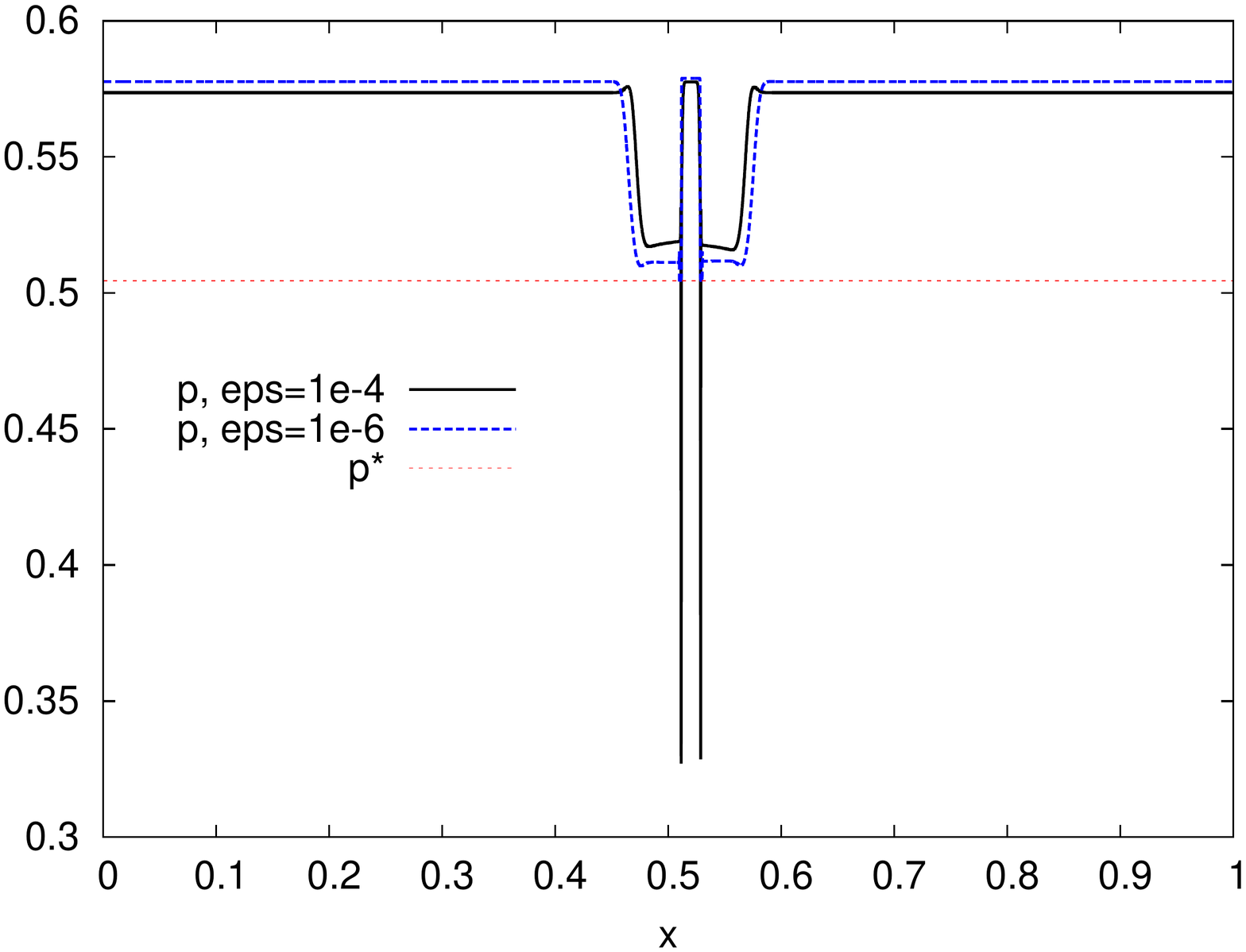}
  \includegraphics[width=6cm]{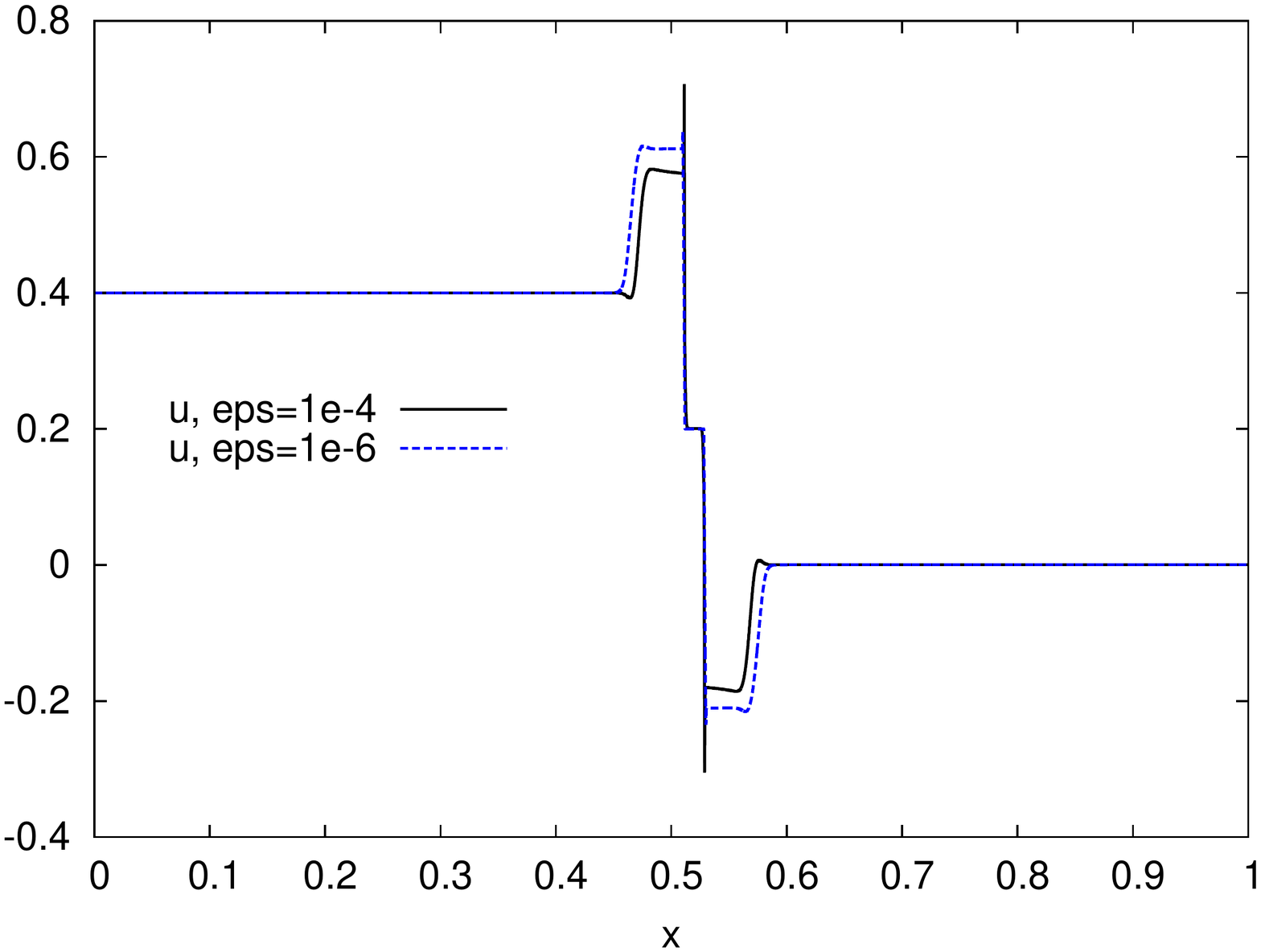}
\caption{Perturbation of a metastable state at t=0.1s. From top left to bottom right: density $\rho$,
    densities $rho_1$ and $\rho_2$, volume fraction $\alpha_1$, speed
    of sound $c$, pressure $p$ and
    velocity $u$.}
\label{fig:meta1}
\end{center}
\end{figure}

\begin{figure}
\begin{center}
  \includegraphics[width=6cm]{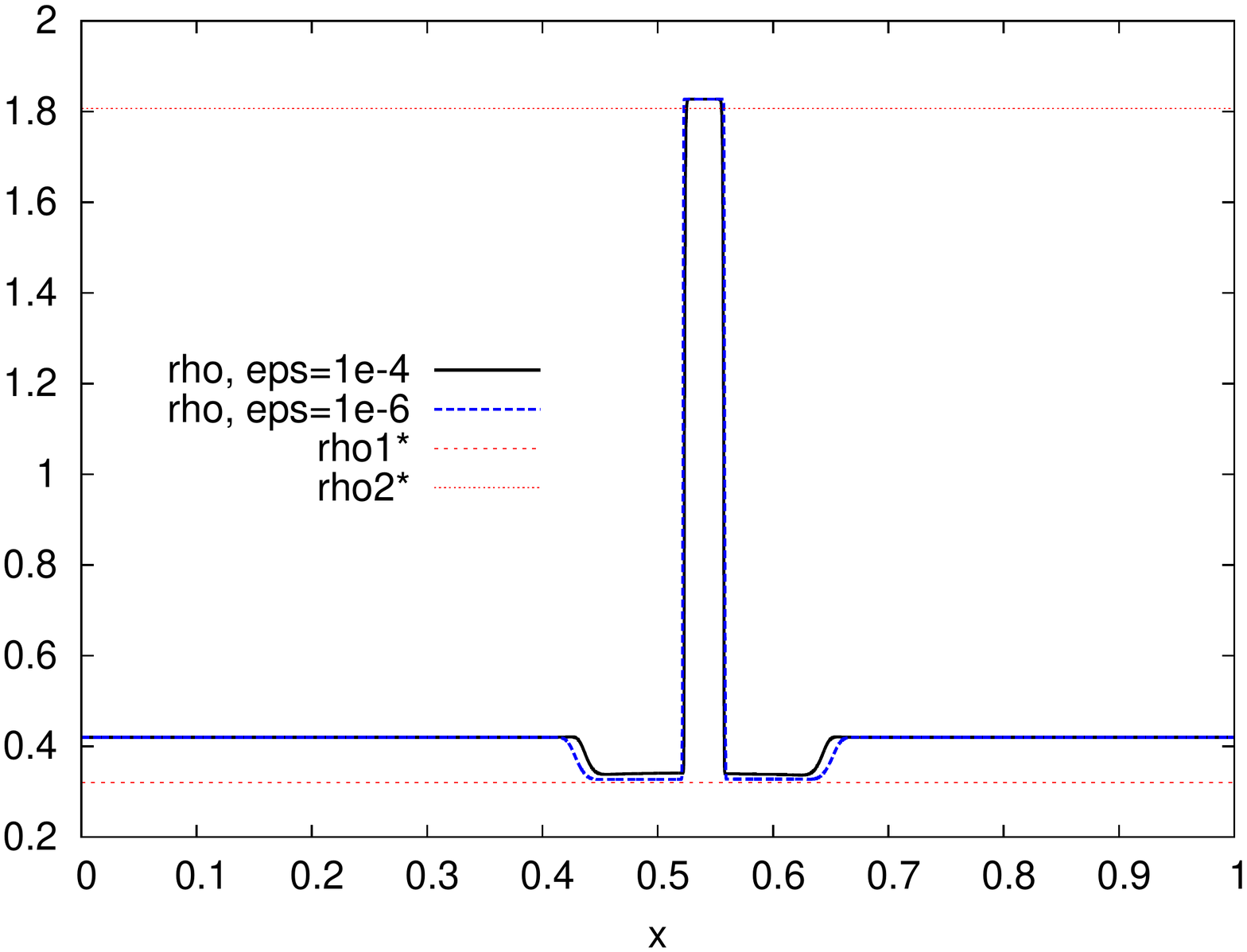}
  \includegraphics[width=6cm]{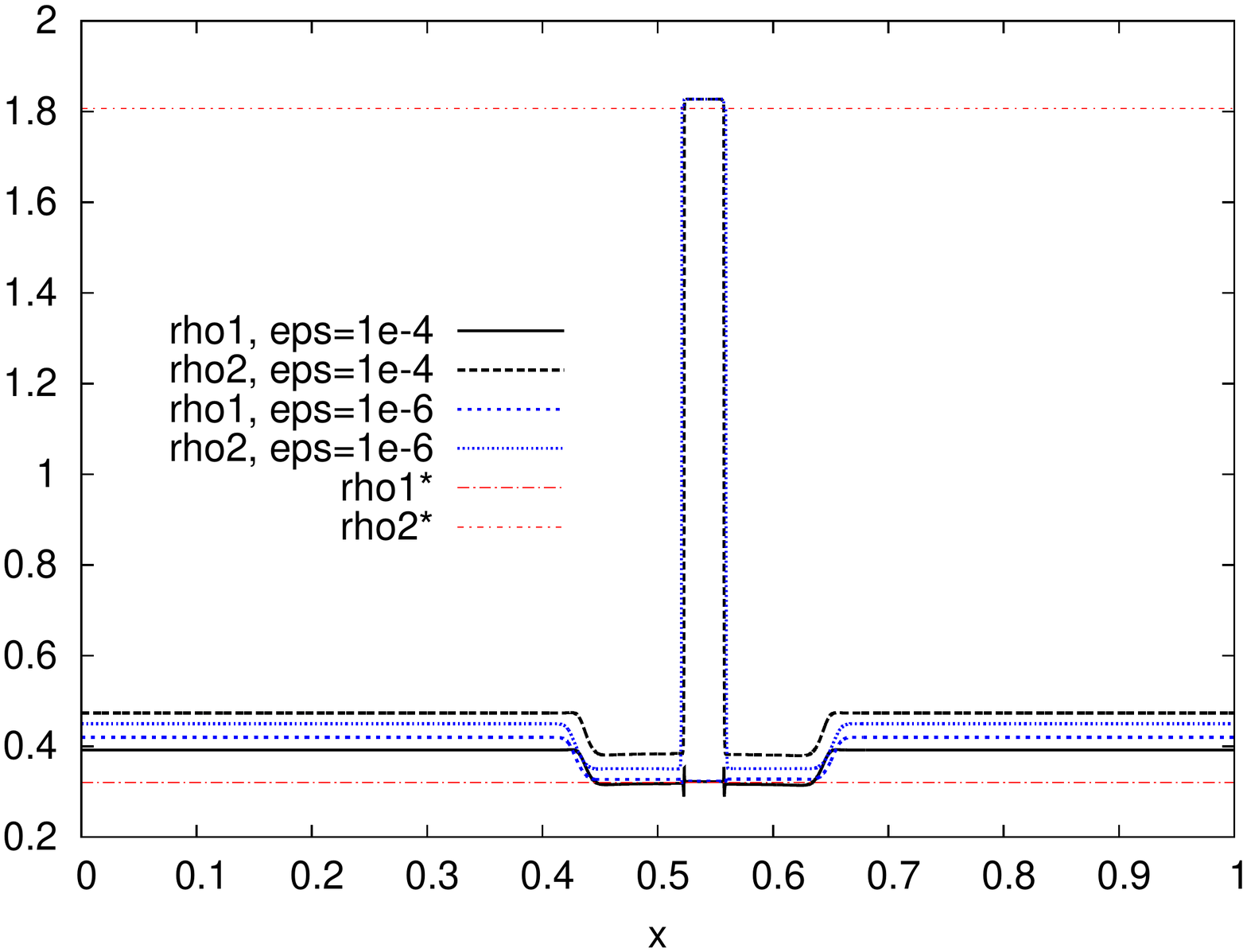}
  \includegraphics[width=6cm]{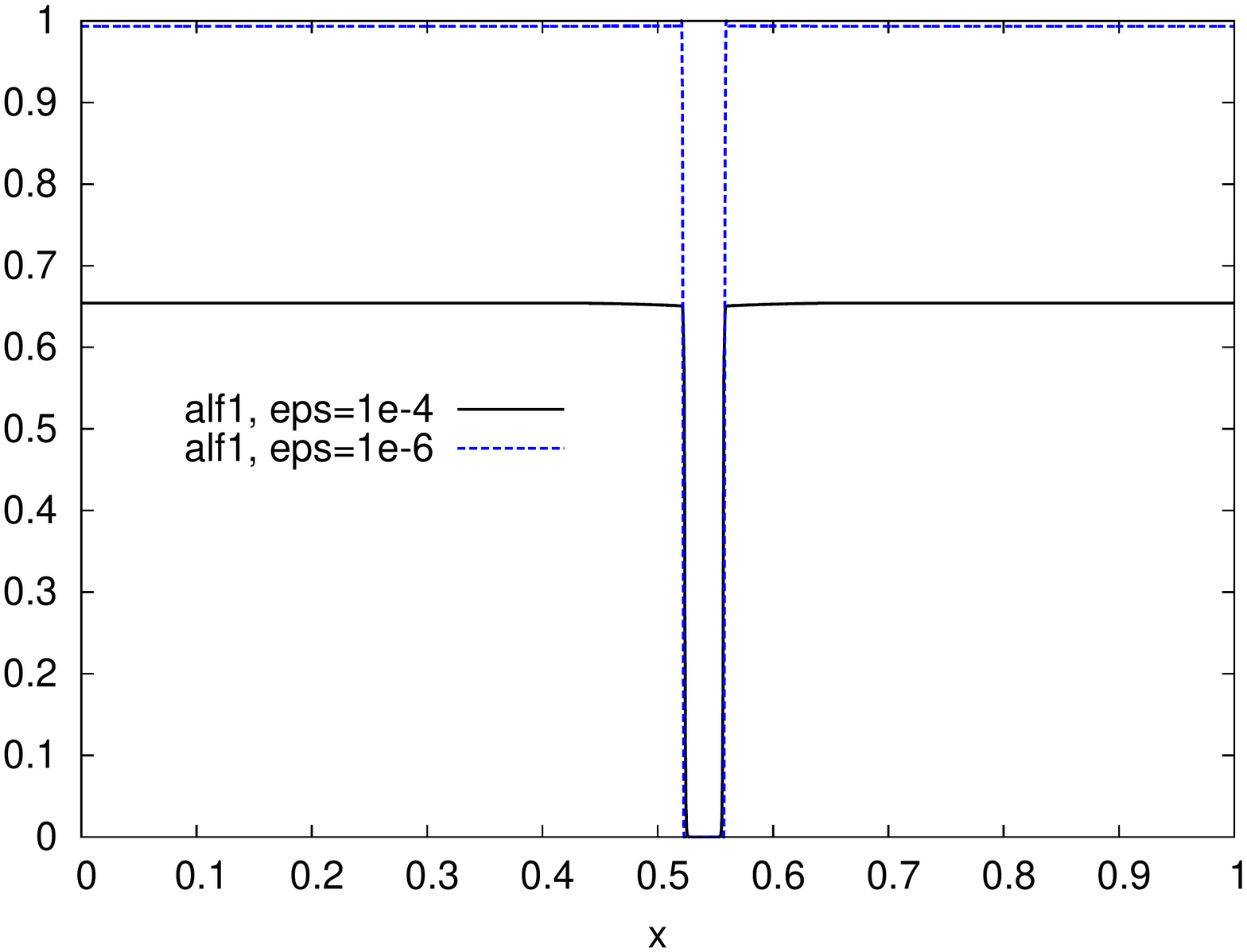}
  \includegraphics[width=6cm]{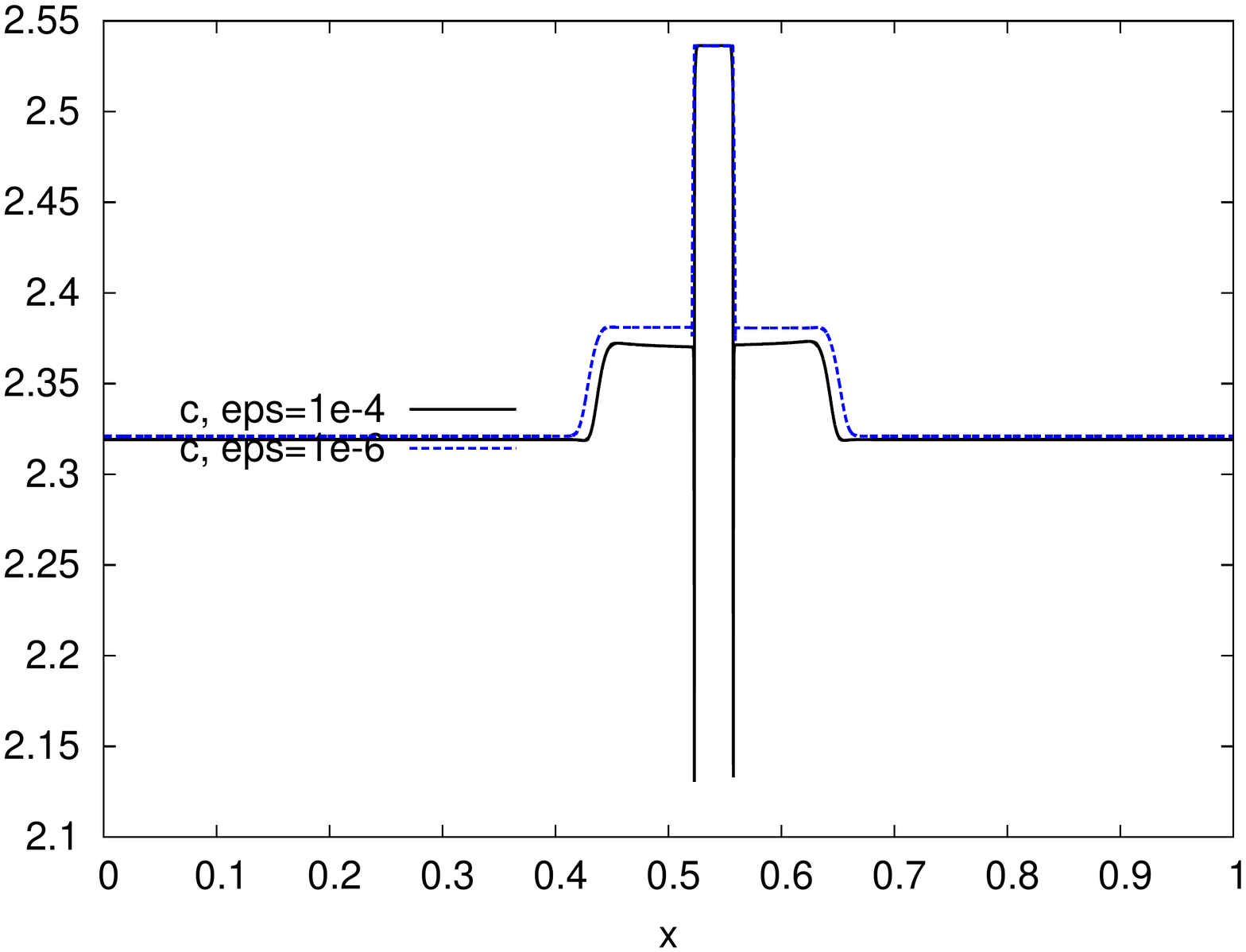}
  \includegraphics[width=6cm]{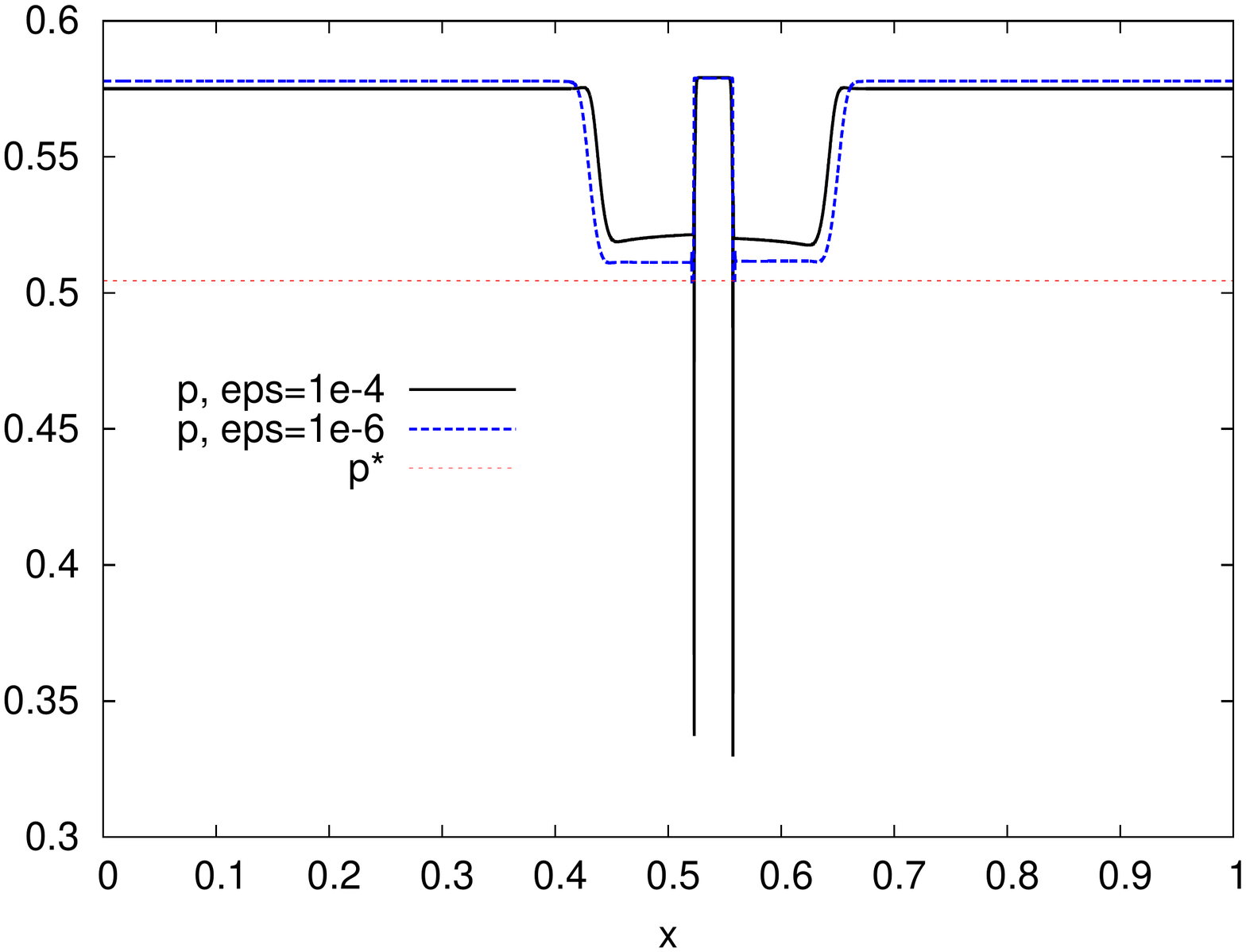}
  \includegraphics[width=6cm]{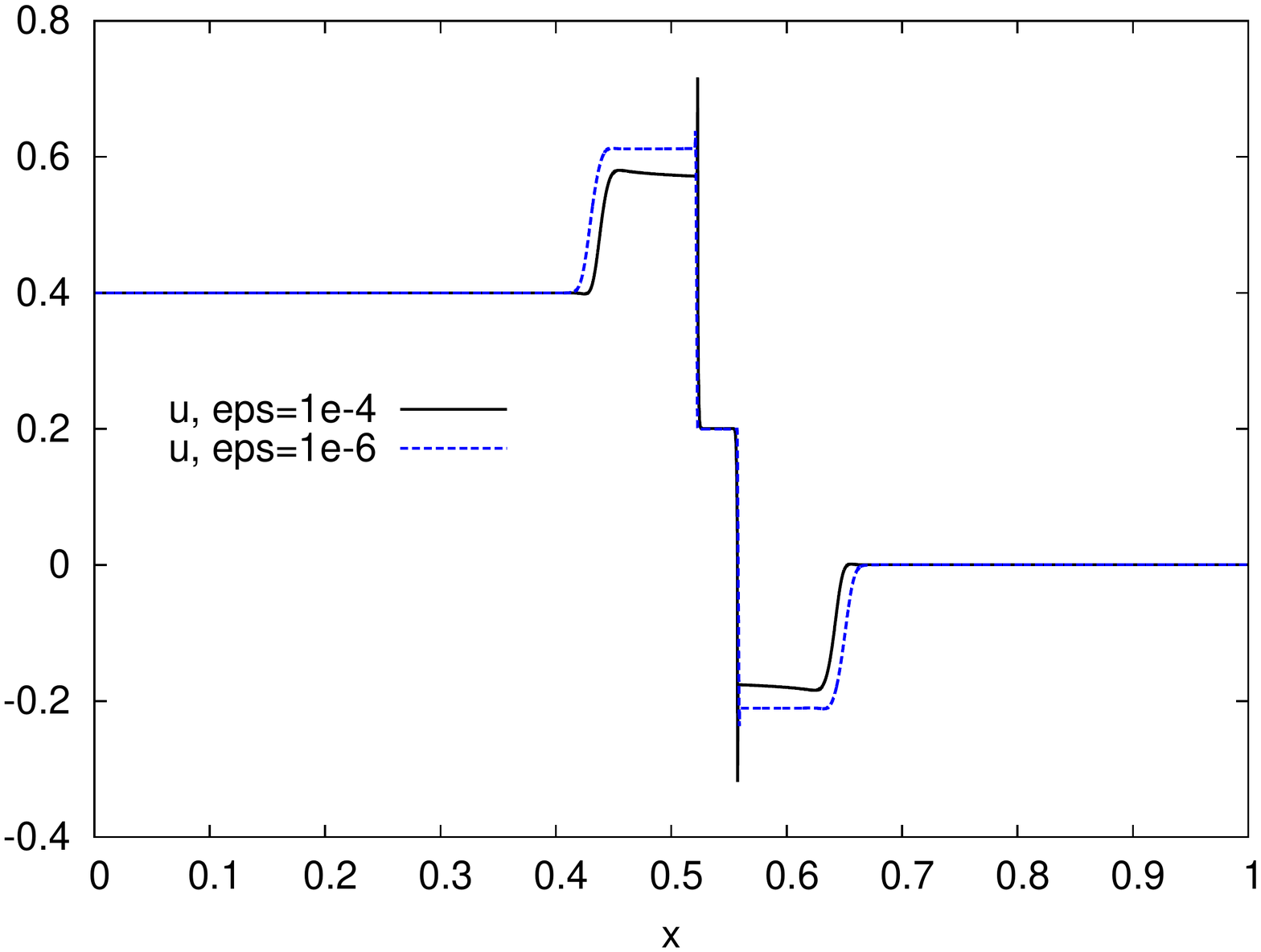}
\caption{Perturbation of a metastable state at t=0.2s. From top left to bottom right: density $\rho$,
    densities $rho_1$ and $\rho_2$, volume fraction $\alpha_1$, speed
    of sound $c$, pressure $p$ and
    velocity $u$.}
\label{fig:meta2}
\end{center}
\end{figure}

\begin{figure}
\begin{center}
  \includegraphics[width=6cm]{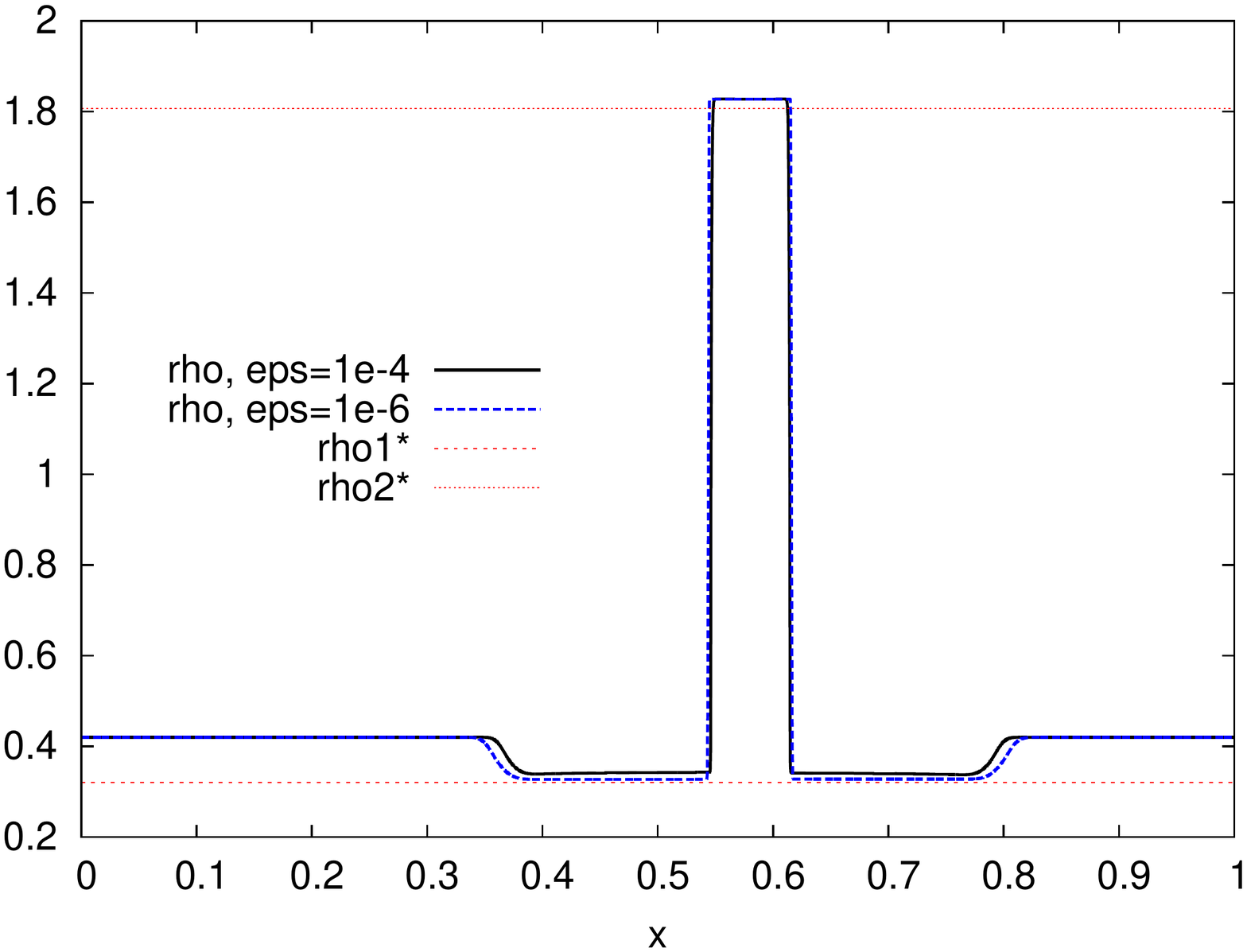}
  \includegraphics[width=6cm]{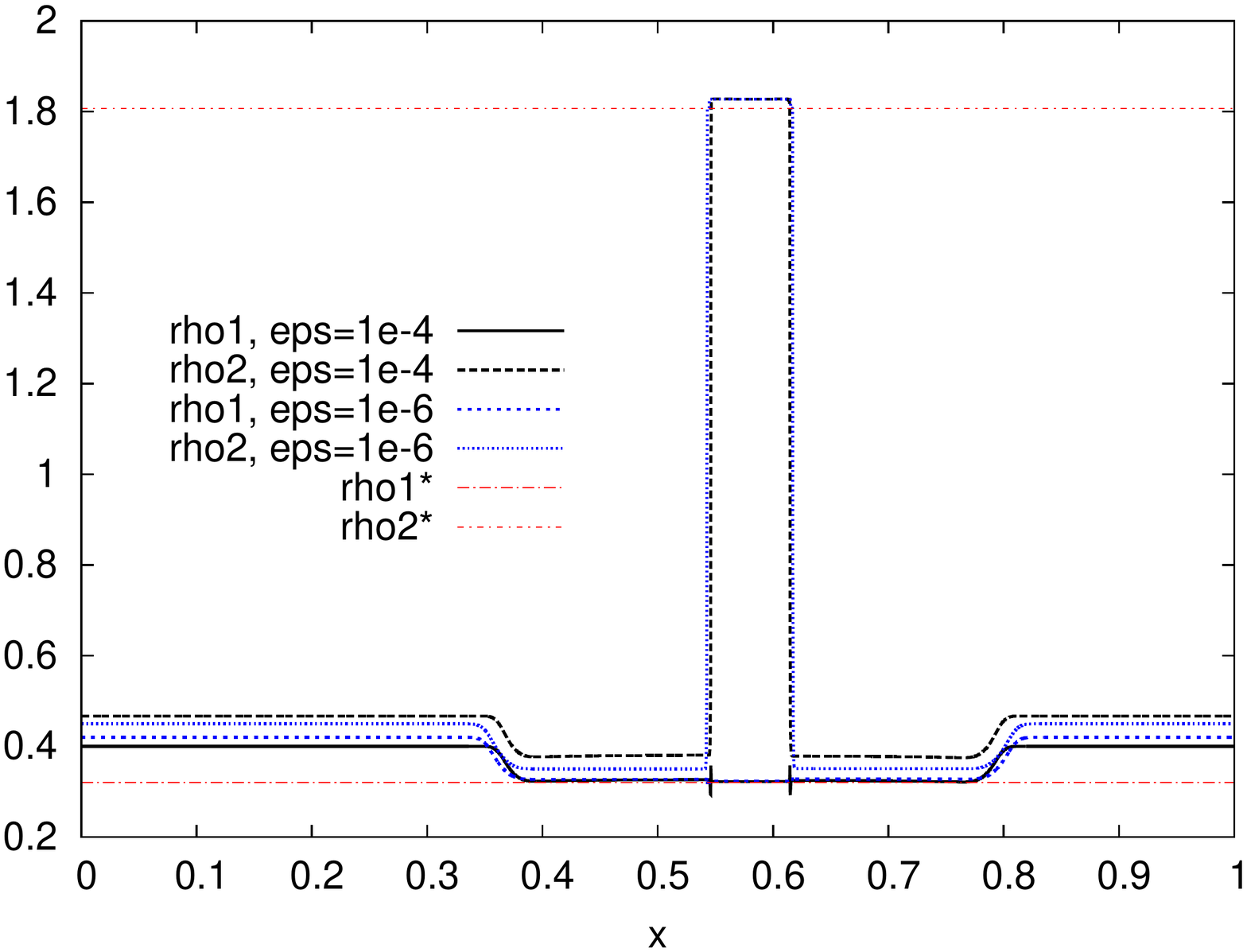}
  \includegraphics[width=6cm]{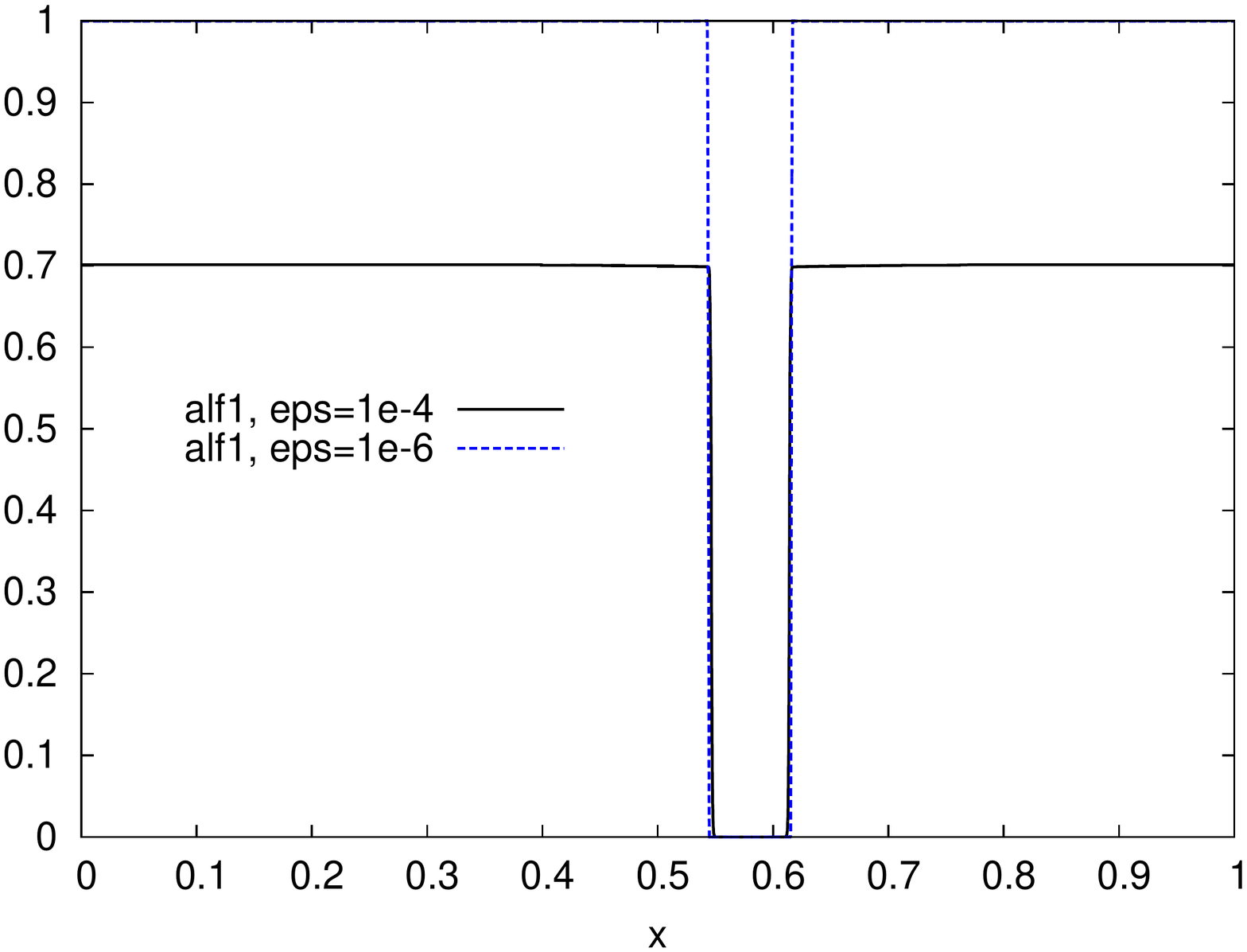}
  \includegraphics[width=6cm]{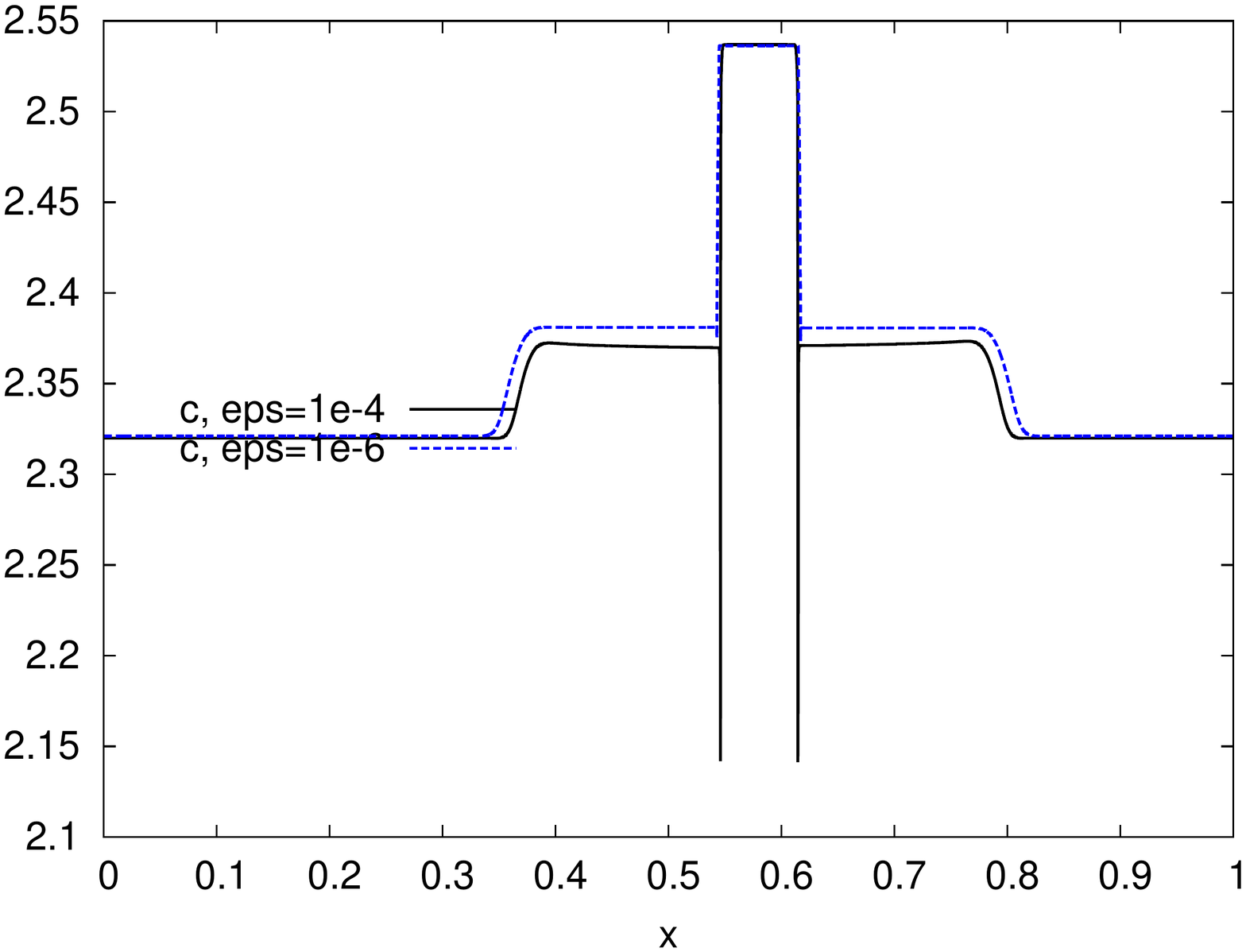}
  \includegraphics[width=6cm]{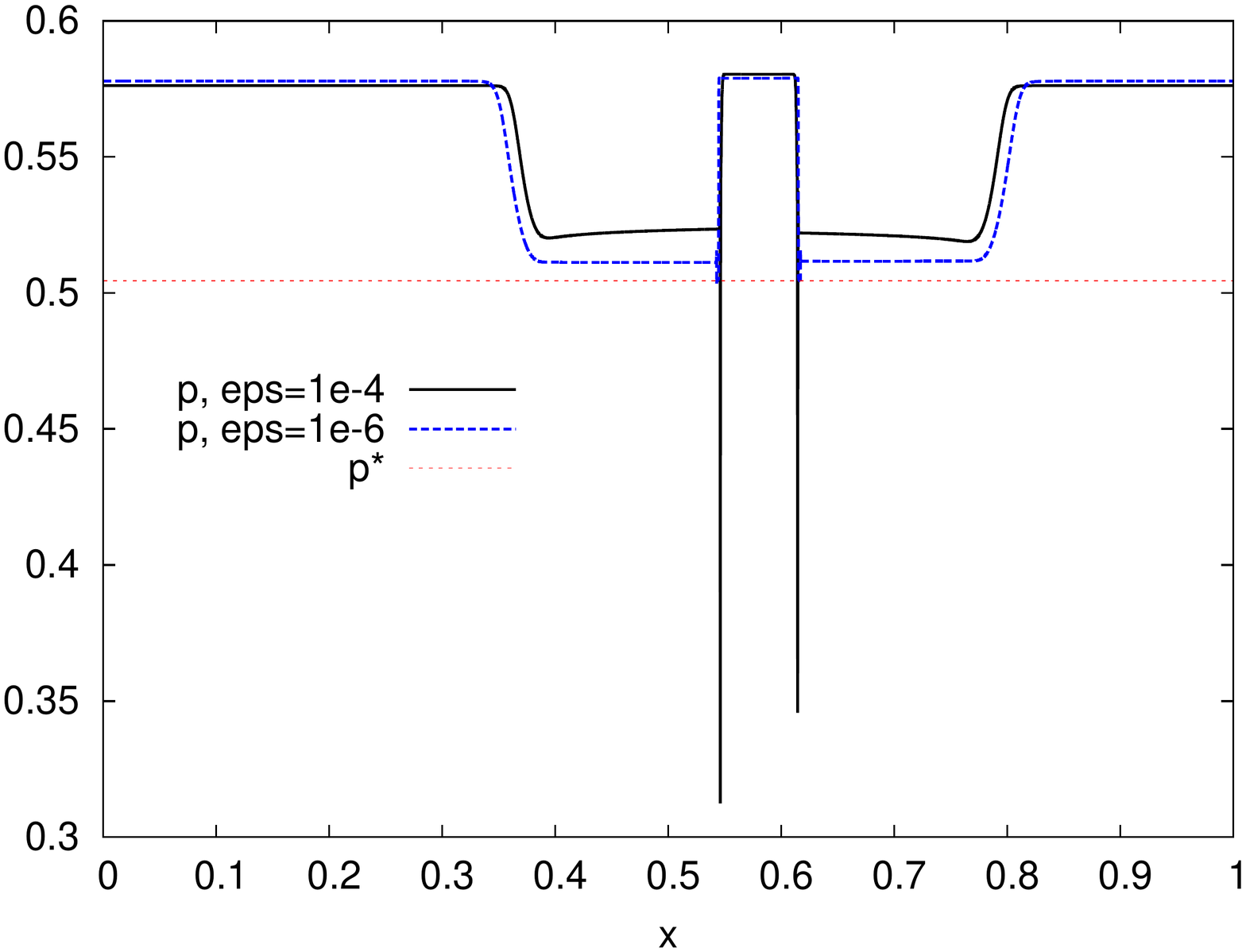}
  \includegraphics[width=6cm]{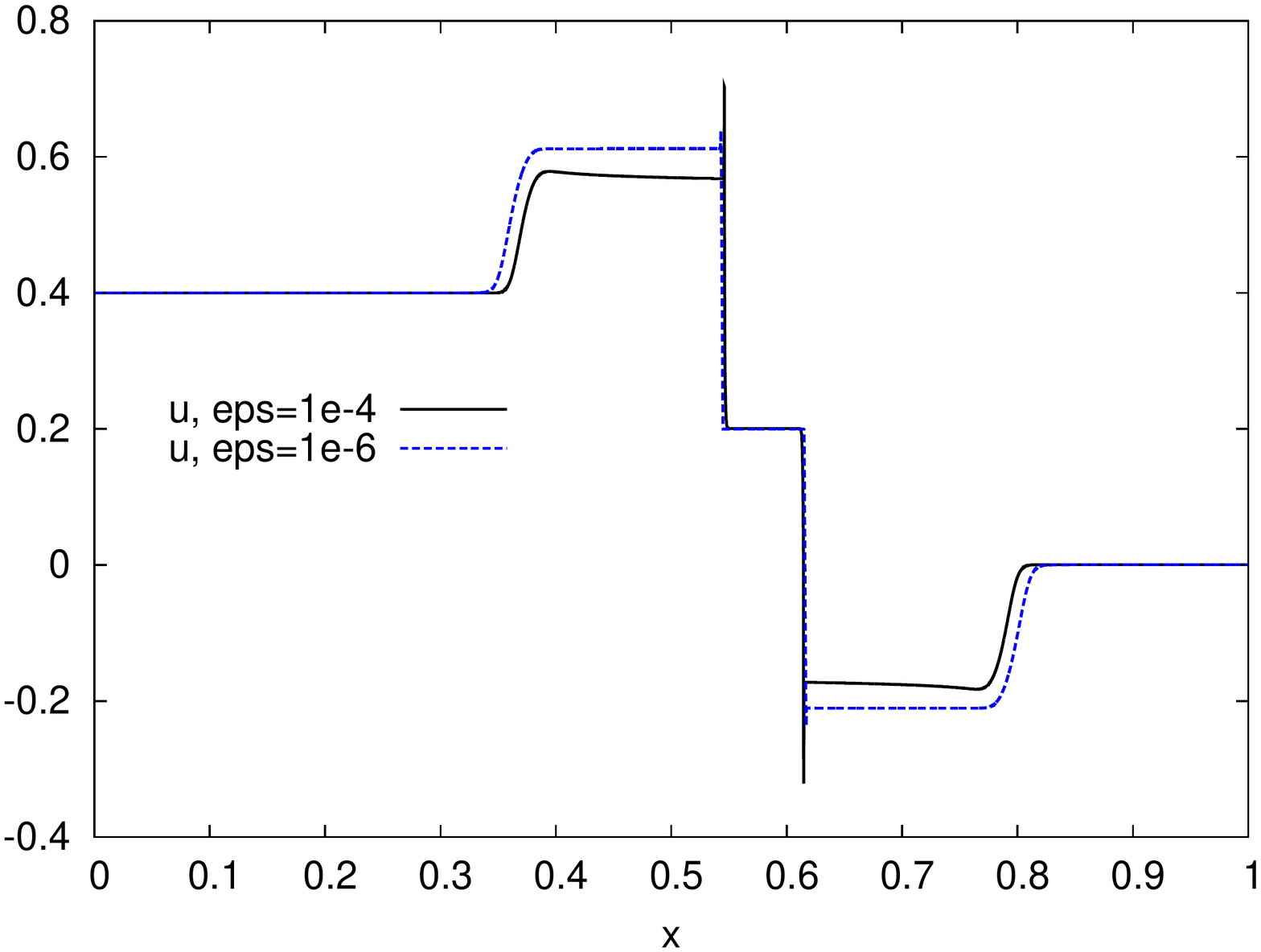}
\caption{Perturbation of a metastable state at t=0.4s. From top left to bottom right: density $\rho$,
    densities $rho_1$ and $\rho_2$, volume fraction $\alpha_1$, speed
    of sound $c$, pressure $p$ and
    velocity $u$.}
\label{fig:meta4}
\end{center}
\end{figure}

%%%%%%%%%%%%%%%%%%%%%%%%%%%%%%%%%%%%%%%%%%%%%%%%%%%%%%%%%%%%%
\section{Conclusion}
\label{sec:conclusion}
%%%%%%%%%%%%%%%%%%%%%%%%%%%%%%%%%%%%%%%%%%%%%%%%%%%%%%%%%%%%%
The core of this work is the formalization in terms of a dynamical system
of the thermodynamics of phase transition, using the van der Waals EoS. It leads to a mathematical 
characterization of metastable states, compared to stable coexistent two-phase
states. The dynamical behaviour of the solutions is crucial here, and seems to preclude the use
of instantaneous exchange kinetic. When coupled to a simple hydrodynamic model, namely the isothermal Euler
equations, it evidences abilities to cope with metastable states as well as
bubble or droplet generation. 
This preliminary study gives rise to a wide bunch of open questions and problems, and we believe that 
the methodology can be used in a much larger context. 

First, in the same isothermal context, the construction of the dynamical system (the right-hand side in the extended Euler equations)
can be addressed. We deliberately used a simple and readable function, which possibly could be improved. In any case, the behaviour
of the coexistence zone around the interface has to be 
investigated in more details, as well as the role of $\varepsilon$. 

Next, an obvious mandatory issue is the numerical
treatment of the coupled system. We have chosen here the simplest numerical strategy that allowed us to illustrate our purpose. 
 The explicit treatment of the stiff relaxation term enforces tough constraints on the time step, and prevents the simulation
of more realistic metastable cases.

Finally, we attend to
include temperature dependance to obtain a fully heat, mass and
mechanical transfer model in order to compare our results to those
of \cite{saurel08} and \cite{Zein10}.

\bibliographystyle{plain}
\bibliography{meta}

\end{document}